\documentclass[
  reprint,            % or preprint
  superscriptaddress, % author affiliations as superscripts
  amsmath,amssymb,    % AMS math packages
  aps,                % American Physical Society
  prl,                % Physical Review Letters
  floatfix,           % resolve figure-placement issues
]{revtex4-2}

\usepackage{graphicx}% Include figure files
\usepackage{dcolumn}% Align table columns on decimal point
\usepackage{bm}% bold math
\usepackage{comment}
\usepackage{braket}
\usepackage{multirow}
\usepackage{amsthm}
\usepackage{subcaption}
\newtheorem{example}{Example}

\newtheorem{thm}{Theorem}

\newcommand{\tr}{\text{tr}}
\usepackage[usenames,dvipsnames]{xcolor}
\newcommand{\liubov}[1]{\textcolor{black}{[liubov] #1}}
\usepackage{hyperref}% add hypertext capabilities
\begin{document}

\preprint{APS/123-QED}

\title{Nonparametric Learning Non-Gaussian Quantum States of Continuous Variable Systems}% Force line breaks with \\

\author{Liubov Markovich}
\email{markovich@mail.lorentz.leidenuniv.nl}
\author{Xiaoyu Liu}%
%email{xiaoyu_liu0330@hotmail.com}
 \author{Jordi Tura}
%\email{tura@lorentz.leidenuniv.nl}
\affiliation{%
 Instituut-Lorentz, Universiteit Leiden, P.O. Box 9506, 2300 RA Leiden, The Netherlands}
\affiliation{$\langle \text{aQa}^\text{L} \rangle$ Applied Quantum Algorithms Leiden, The Netherlands
}%


\date{\today}% It is always \today, today,

\begin{abstract}
Continuous-variable quantum systems are foundational to quantum computation, communication, and sensing. While traditional representations using wave functions or density matrices are often impractical, the tomographic picture of quantum mechanics provides an accessible alternative by associating quantum states with classical probability distribution functions called tomograms.
Despite its advantages, including compatibility with classical statistical methods, the tomographic method remains underutilized due to a lack of robust estimation techniques. This work addresses this gap by introducing a non-parametric \emph{kernel quantum state estimation} (KQSE) framework for reconstructing quantum states and their trace characteristics from noisy data, without prior knowledge of the state. In contrast to existing methods, KQSE yields estimates of the density matrix in various bases, as well as trace quantities such as purity, higher moments, overlap, and trace distance, with a near-optimal convergence rate of $\tilde{O}\bigl(T^{-1}\bigr)$, where $T$ is the total number of measurements. KQSE is robust for multimodal, non-Gaussian states, making it particularly well suited for characterizing states  essential for
quantum science. %achieving universal quantum computation, entanglement distillation, quantum error correction, and metrological advantage.
%Moreover, by addressing key challenges such as bandwidth selection and integral estimation over bounded domains, we propose a hybrid parametric and non-parametric approach that enables accurate tomographic state reconstruction even under realistic experimental constraints.
\end{abstract}

\maketitle

%\tableofcontents

%\section{\label{sec:level1}Introduction}

\par \textit{Introduction.} Continuous-variable (CV) quantum states describe systems with observables taking values in continuous spectra~\cite{serafini2017quantum}. In quantum optics, their realizations, known as qumodes, serve as CV counterparts to qubits, as light fields naturally exhibit continuous quadrature degrees of freedom. CV platforms support a wide range of quantum information tasks, including computation, communication, and sensing~\cite{lau2016universal, arrazola2019machine, killoran2019continuousvariable, bromley2020applications,jia2025continuousvariable, nielsen2025variational}.
As a result, significant effort has been devoted to learning CV quantum states~\cite{wu2024efficient, mele2024learning, holevo2024estimates}, with particular progress for Gaussian and mildly non-Gaussian states~\cite{walschaers2021nongaussian}, establishing tight bounds on the trace distance in terms of first moments. Alternatively, shadow
tomography for truncated CV systems or energy constraint states is used~\cite{becker2024classical,aaronson2018shadow,volkoff2021universal}.
  \par Due to the continuous nature of qumodes, the density matrix formalism is often impractical.  However, CV states can  be described in phase space by joint quasiprobability distribution functions (QPDFs), such as the Wigner, Q and P‐functions~\cite{wigner1932quantum, husimi1940formal, cahill1969density}. These QPDFs are related via  one-to-one transforms~\cite{manko2020integral} providing a complete description of the state. However, since position $\boldsymbol{q}$ and momentum $\boldsymbol{p}$ operators cannot be measured simultaneously, we lack a pair of observable random variables to reconstruct their joint QPDFs. Instead, quantum state tomography~\cite{smithey1993measurement, lvovsky2009continuousvariable, babichev2004homodyne, wallentowitz1995reconstruction} measures a finite number of rotated quadratures $\boldsymbol{X}_{\mu,\nu}=\mu \boldsymbol{q}+\nu \boldsymbol{p}$ at various  phase-space directions  $\mu,\nu\in  \mathbb{R}$. Using these statistics,  the desired QPDF is reconstructed. 
  \par
The Wigner function is reconstructed by inverting the Radon transform of quadrature distributions estimated via histograms and maximum-likelihood estimation (MLE)~\cite{leonhardt1995measuring,vogel1989determination,banaszek1998maximumlikelihood}. This process is sensitive to sampling noise and the inversion of the Radon transform can be ill-posed, necessitating the use of regularization~\cite{vogel1989determination,lvovsky2009continuousvariable}. Detector inefficiencies blur nonclassical features such as Wigner negativity~\cite{dariano2007homodyne}. 
These propagate into the reconstructed density matrix, leading to loss of nonclassicality, reduced purity, or even nonphysical states~\cite{lvovsky2009continuousvariable, hradil1997quantumstate}.
To avoid such instabilities, modern experiments reconstruct the  state directly from quadrature data using parametric methods
that guarantee physicality~\cite{christ2009spatial,wu2024efficient,fluhmann2020direct,jia2025continuousvariable}. The quadrature $X_{\mu,\nu}$ is a random variable dependent on the measurement settings  $\mu, \nu$, and its PDF $\mathcal{W}(x|\mu,\nu)$, called the symplectic tomogram, is informationally complete~\cite{amosov2008information}, uniquely specifying the density operator $\boldsymbol{\rho}$. Hence, every CV state can be fully described by the classical PDF, i.e., the tomogram, of its quadratures~\cite{mancini1996symplectic,ibort2009introduction, manko1997quantum, dodonov1974even, chernega2023dynamics,dudinets2024entangled,manko2020integral}. 
\par In this Letter, we propose a state reconstruction and trace characteristic estimation method based on nonparametric kernel density estimation (KDE)~\cite{silverman2018density} and a kernel characteristic function estimation (KCFE), called \textit{kernel quantum state estimation} (KQSE). KQSE estimates the state characteristics, using their integral representations via the tomogram and characteristic functions, achieving $\tilde{O}(T^{-1})$
in the uniform and $L_2$ norms, where $T$ is the total amount of measurements with no prior assumptions on the state (Theorem~\ref{thm0}).
%requiring $\tilde{O}(\varepsilon^{-2})$ measurements to achieve accuracy $\varepsilon$ in the uniform and $L_2$ norms with no prior assumptions on the state (Theorem~\ref{thm0}). 
This approach is robust to measurement noise and performs equally well for Gaussian and non-Gaussian states.
\par  We first prove that the trace distance  $D(\boldsymbol{\rho}_1,\boldsymbol{\rho}_2)=\tfrac{1}{2}\|\boldsymbol{\rho}_1-\boldsymbol{\rho}_2\|_1$ between any quantum states $\boldsymbol{\rho}_1$, $\boldsymbol{\rho}_2$ is lower bounded by the total variation between the corresponding tomograms  (Theorem~\ref{thm1}).
In Theorem~\ref{thm2} we show that achieving precision $\varepsilon$ in trace distance using KQSE requires $O(\varepsilon^{-5/2})$ measurements. While slightly looser than the $O(\varepsilon^{-2})$ rate known for Gaussian states~\cite{holevo2024estimates,mele2024learning,devroye2023total}, our bound holds universally without assumptions on the state.  Next, we use our previous work~\cite{markovich2024not},  where the characteristic function (CF) $\phi(t;\mu,\nu)$ is defined as the Fourier transform of the tomogram. Theorem~\ref{thm_3} provides necessary and sufficient conditions for a CF to define a valid quantum state. Moreover, the CF enables not only the reconstruction of the density matrix $\rho(y,y')=\langle y|\boldsymbol{\rho}|y'\rangle$  in various bases, but also allows for direct analytic evaluation of quantities such as $\tr{(\boldsymbol{\rho}^d)}$ $(d\in\mathbb{Z}^{+})$, and state distinguishability measures like the trace distance $D(\boldsymbol{\rho}_1,\boldsymbol{\rho}_2)$ and overlap $\tr{(\boldsymbol{\rho}_1\boldsymbol{\rho}_2)}$. We show that this CF coincides with that of the Wigner function~\cite{wu2024efficient, fluhmann2020direct}. Fig.~\ref{fig_1} illustrates the theoretical and empirical correspondence among the tomogram, CF, density matrix, and Wigner function~\cite{manko2020integral}.
Finally, we introduce the new nonparametric KCFE that achieves a mean squared error (MSE) convergence rate of ${O}(1/n)$, where $n$ is the amount of measurements for the fixed $\mu,\nu$, matching optimal parametric rates of MLE (Theorem~\ref{thm4}). In signal processing~\cite{stefanski1990deconvolving, sinitsyn2005kernel,fan1991optimal, delaigle2004practical}, PDFs of the signals are often estimated via the inverse Fourier transform of a CF, as this approach empirically yields more stable results. 
%To the best of our knowledge, we provide the theoretical justification for this practice for the first time, making it of interest to both the quantum and statistical communities. 
\par  The KQSE method uses a collection of KCFEs $\{\hat{\phi}_X(1;\mu_i,\nu_j\}$, $i\in[1,N_{\mu}]$, $j\in[1,N_{\nu}]$ to estimate  the state and its trace characteristics from the noisy data:
\begin{thm}(Convergence of KQSE)\label{thm0}
Let $\widehat{\rho(y,y')}$ and $\widehat{\mathrm{tr}(\boldsymbol{\rho}_1 \boldsymbol{\rho}_2)}$ denote the KQSE estimators of the density matrix $\rho(y,y')$ and the overlap $\mathrm{tr}(\boldsymbol{\rho}_1 \boldsymbol{\rho}_2)$, respectively, constructed from $T_{\mu} = n N_\mu$ and $T_{\mu,\nu} = n N_\mu N_\nu$ quantum state preparations, where $n$ samples are collected per  setting $\{\mu_i, \nu_j\}$, $ N_\mu, N_\nu=\log{n}$. Then the estimator achieves:
    \begin{eqnarray}\label{1723}
\!\!\!&&   L_{\infty}(\widehat{\rho}(y,y')) \equiv \sup\limits_{y,y'} \mathbb{E}|\rho(y,y') - \widehat{\rho(y,y')}|^2 = \tilde{O}(1/T_{\mu}),\\\nonumber
\!\!&&\mathrm{MSE}\big(\widehat{\mathrm{tr}(\boldsymbol{\rho}_1 \boldsymbol{\rho}_2)}\big) \!\equiv \!\mathbb{E}\big[|\mathrm{tr}(\boldsymbol{\rho}_1 \boldsymbol{\rho}_2) - \widehat{\mathrm{tr}(\boldsymbol{\rho}_1 \boldsymbol{\rho}_2)}|^2\big]\! = \!\tilde{O}(1/T_{\mu,\nu}).\end{eqnarray}
\end{thm}
%Then $T = \tilde{O}(\varepsilon^{-2})$ suffices to achieve accuracy $\varepsilon$ in the uniform norm and  MSE precision $\varepsilon$.
Using the KQSE to estimate the pure state $\boldsymbol{\rho}$ with $\boldsymbol{\rho}_n$ we get the optimal rate $
\mathrm{MSE}(\widehat{D^2(\boldsymbol{\rho},\boldsymbol{\rho}_n)})=\tilde{O}(T_{\mu,\nu}^{-1})$,
and the probability of KQSE to give a squared distance greater than 
$\delta$ decays like $\tilde{O}\left(T_{\mu,\nu}^{-1}\delta^{-2}\right)$.
\par Comparable to other state reconstruction and characterization methods that use parametric methods like MLE, we do not assume Gaussianity or use any prior information, only measurement data, reaching the same rate of convergence. Moreover, the estimator incorporates a built‐in filtering step, requiring no additional measurements even in the presence of noisy data. 
%Although similar ideas exist in signal processing~\cite{stefanski1990deconvolving,fan1991optimal,delaigle2002datadriven}, to the best of our knowledge, this is the first application of nonparametric estimation in quantum theory.
\begin{figure}[ht]
\includegraphics[width=\linewidth]{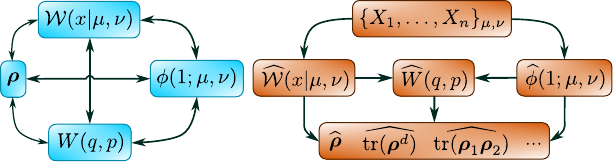}% Here is how to import EPS art
\caption{In the left panel, a schematic diagram shows how the density operator $\boldsymbol{\rho}$, the tomogram $\mathcal{W}(x| \mu,\nu)$, the Wigner function $W(q,p)$, and the CF $\phi(1;\mu,\nu)$ are related. These representations are equivalent and interconvertible via one‐to‐one mappings. In the right panel, the flowchart depicts how each representation is reconstructed from the measured data $\{X_1,\dots,X_n\}_{\mu,\nu}$. Here, hats denote estimates rather than operators (e.g.\ $\widehat{\mathcal{W}}(x| \mu,\nu)$ is an estimate of $\mathcal{W}(x| \mu,\nu)$). In particular, the Wigner function can be obtained from either the tomogram or the CF. 
\label{fig_1}}
\end{figure}

%\textcolor{black}{Maybe I need to add numbers on the pathes I take. }
\par
In~\cite{manko1997quantum,manko2023quantum}, the tomograms of quantum harmonic oscillators (HO) and inverted HO are shown to belong to the Gaussian-Hermite class~\cite{barndorff-nielsen2014information}, with the ground state yielding a simple Gaussian tomogram~\cite{manko2023quantum, markovich2024not}. Superpositions of Gaussian states form Gaussian mixtures~\cite{moya-cessa2014optical}, resulting in multimodal distributions with Gaussian tails. Non-Gaussian states such as cat states~\cite{yurke1986generating} and GKP states~\cite{gottesman2001encoding} exhibit even richer structures, with crystallized cat tomograms recently analyzed in~\cite{lopez-saldivar2022symplectic,markovich2024not}. While parametric models like Gaussian mixtures require prior assumptions about the number of modes and may fail for complex shapes, KQSE offer data-driven, flexible alternatives that adapt naturally to multi-modal distributions without model bias (see the later section \textit{Simulation Study}).
%The paper is organized as follows. We briefly review the tomographic picture of quantum mechanics, showing how CF simplifies this formalism and directly connects to the state and its trace characteristics. The trace distance between quantum states is analyzed in detail.  We introduce the KDE method for tomogram and CF estimation for the first time, incorporating a built-in filtering technique to address typical noise in homodyne tomography.  We presents density operator reconstruction from the KCFE, along with a detailed error analysis.  Finally, we conduct a simulation study for various initial non-Gaussian states under measurement noise.
\\\textit{Kernel Quantum State Estimation.} We consider a quantum state in an infinite-dimensional Hilbert space $ \mathcal{H} $, represented by the density matrix operator $ \boldsymbol{\rho} $. By definition, it is Hermitian, positive semi-definite, and has unit trace.  
%Its kernel  is given by   $\rho{(\psi,\psi')}=\bra{\psi}\boldsymbol{\rho}\ket{\psi'}$where $ \ket{\psi} $ and $ \ket{\psi'} $ are quantum states.
The  connection between $\boldsymbol{\rho}$ and its symplectic tomogram is given by $\mathcal{W}(x|\mu,\nu)=\mathrm{tr}\bigl[\boldsymbol{\rho}\,\delta(x\,\boldsymbol{1}-\mu\boldsymbol{q}-\nu\boldsymbol{p})\bigr]$ and
\begin{align}\label{1_1}
\boldsymbol{\rho}=\frac{1}{2\pi}\!\iint
  \left[\int\!\mathcal{W}(x|\mu,\nu)e^{ix}\,dx\right]
e^{-i(\mu\boldsymbol{q}+\nu\boldsymbol{p})}\,d\mu\,d\nu,
\end{align}
where $\delta(\cdot)$ denotes the Dirac delta function of an operator. 
 The tomogram is a PDF (always nonnegative and normalized) of a continuous random variable $X_{\mu,\nu}$ corresponding to the quadrature operator $\boldsymbol{X}_{\mu,\nu}$ (see Supplemental Material). 
\par Consider identically prepared CV systems in an unknown state with density operator $\boldsymbol{\rho}$. For each system and a fixed setting $(\mu_i,\nu_j)$, performing $n$ independent measurements yields independent and identically distributed (i.i.d.) random variables
$\{X_{l,\mu_i,\nu_j}\}$ from an unknown PDF $\mathcal{W}(x|\mu_i,\nu_j)$.  
%Repeating this for $i\in[1,N_{\mu}]$, $j\in[1,N_{\nu}]$ distinct settings produces $N_{\mu} N_{\nu}$ data sets.
%Our first goal is to construct an estimator  which is $\varepsilon$ close in $L_1$ and/or $L_2$ norms to the true state with high probability using no assumptions on the state.
%\par 
We introduce the KDE of this tomogram  as:
\begin{eqnarray}\label{1500}
    \widehat{\mathcal{W}}_{nh}(x|\mu,\nu) = \frac{1}{nh} \sum_{l=1}^{n} K_{\mu,\nu}\left(\frac{x-X_{l,\mu,\nu} }{h}\right).
\end{eqnarray} 
Note that, following mathematical convention, in the following section we  use hats for the estimators. 
The kernel function $K_{\mu,\nu}(\cdot )$ is typically chosen to be a PDF (e.g., Gaussian). 
KDE estimates probability densities from finite samples without assuming any parametric form. By adjusting the bandwidth $h > 0$,  it can resolve complex and multimodal structures more accurately than histograms or MLE, which often suffer from binning artifacts or model bias~\cite{chen2017tutorial,ChaconDuong2020}. 
\par The total variation distance $T(\mathcal{W}_1, \mathcal{W}_2)$ between two tomograms is defined as half of the $L_1(\mathcal{W}_1, \mathcal{W}_2)= \int \left| \mathcal{W}_1(x|\mu,\nu) - \mathcal{W}_2(x|\mu,\nu) \right|\, dx$.
Then:
\begin{thm}\label{thm1}(Lower bound to closeness of any
states)
 The trace distance $D(\boldsymbol{\rho}_1, \boldsymbol{\rho}_2)$ between two CV quantum states is lower bounded by the total variation $T(\mathcal{W}_1, \mathcal{W}_2)$ between their tomograms.
\end{thm}
Unlike previous results limited to Gaussian states~\cite{wu2024efficient,mele2024learning,holevo2024estimates}, our inequality holds for arbitrary quantum states. The upper bound we derive is not tight in general and becomes formally infinite unless the total variation is either compactly supported or decays sufficiently fast at infinity, as is the case for PDFs with Gaussian tails. This implies that some prior knowledge about the PDF class is needed to obtain a tighter bound using our method. For Gaussian states, tighter upper bounds have been derived in recent works~\cite{holevo2024estimates,mele2024learning} based on the first two moments.
\par   For the $L_1(\widehat{\mathcal{W}}_{nh},\mathcal{W})$ distance, KDE achieves a convergence rate of $O(n^{-2/5})$. Then:
\begin{thm}\label{thm2}
Let $\boldsymbol{\rho}$ be an unknown state.
  The number of copies $n$ of $\boldsymbol{\rho}$ required to estimate it by  the KDE method with precision $\varepsilon$ in trace distance has to scale at least as $O(\varepsilon^{-5/2})$. %This result holds using only a single measurement setting $\mu,\nu$.
\end{thm}
 For the $L_2$ distance (i.e., mean integrated squared error (MISE)), KDE achieves a rate of $O(n^{-4/5})$ (see Supplemental Material), outperforming the optimal histogram rate of $O(n^{-2/3})$, but remaining slower than the parametric MLE rate of $O(n^{-1})$~\cite{silverman2018density}. This reflects the tradeoff for model flexibility: KDE requires only smoothness of the underlying density, without assuming a specific parametric form. The resulting tomograms can be directly used to reconstruct the quantum state via Eq.~\eqref{1_1}, but the convergence remains slower than $O(n^{-1})$. In the following, we show that incorporating the CF allows us to recover the optimal $O(n^{-1})$ rate.
\\\textit{Kernel Characteristic Function Estimation.}
One can rewrite \eqref{1_1} using the CF $\phi(t=1;\mu,\nu)$, which is the Fourier transform $e^{ix t}$ of the tomogram. The CF always exists, fully characterizes the distribution, and satisfies $\phi(0) = 1$, $|\phi(t)| \leq 1$, and $\phi(-t) = \phi^{\star}(t)$. In Ref.~\cite{markovich2024not}, we formulated conditions on the CF ensuring $\boldsymbol{\rho}$ reconstruction. Here, we present its equivalent integral formulation:
\begin{thm}{(Hermiticity,  Normalization, Positivity)}\label{thm_3}
The integral \eqref{1_1}  defines the density operator corresponding to a quantum state if and only if the CF of the tomogram satisfies:
 \begin{eqnarray}\label{th3}
%{\rho}^{\star}(y,y')={\rho}(y',y):\quad 
&&\phi(1;\mu,\nu)= \phi(-1;-\mu,-\nu) \quad \forall \mu,\nu,\quad \phi(1;0,0)=1,\nonumber\\
&& 0\leq  \frac{1}{2\pi}
    \iint \phi(1;\mu,\nu)\phi_{{\chi}}^{\star}(1;\mu,\nu)d\mu d\nu\leq 1, \quad \forall \ket{\chi}.
\end{eqnarray}
\end{thm}
Similar result is provided in Ref.~\cite{holevo2011probabilistic} in no connection to the tomographic formalism, and the function satisfying the left inequality in the condition \eqref{th3} is named $\Delta$-positive defined.
We show that the density matrices $\boldsymbol{\rho}_1$ and $\boldsymbol{\rho}_2$ of two quantum states satisfy the noncommutative Parseval relation:
\begin{eqnarray}\label{1357}
   \tr{( \boldsymbol{ \rho}_1\boldsymbol{\rho}_2)}
  &=& \frac{1}{2\pi}\iint \phi_1(1;\mu,\nu) \phi_2^{\star}(1;\mu,\nu) d\mu d\nu,
\end{eqnarray}
where $\phi_1$ and $\phi_2$ are the CFs corresponding to the tomograms of the states.
If $\boldsymbol{\rho}_1=\boldsymbol{\rho}_2$ the latter provides the purity, while 
$\tr{\left(\boldsymbol{\rho}^d\right)}$ and $\tr{\left(\prod_{i=1}^{d}\boldsymbol{\rho}_i\right)}$ are given in Supplemental Material. 
%Then the kernel element $\rho(\chi,\chi)=\langle \chi|\boldsymbol{\rho}|\chi\rangle=\tr{[\boldsymbol{\rho}\boldsymbol{\rho}_{\chi}]}$ of the density operator in any state $\ket{\chi}$ representation can be written in the form \eqref{1357}.
%Since $\boldsymbol{\rho}$ is positive semidefinite and its matrix elements cant be bigger then one, then  $0\leq  \rho( \chi, \chi)\leq 1$, holds $\forall \ket{\chi}\neq 0$.
 Since the distance between  two pure states is $D(\boldsymbol{\rho}_1,\boldsymbol{\rho}_2)=\sqrt{1-\tr{( \boldsymbol{ \rho}_1\boldsymbol{\rho}_2)}}$~\cite{holevo2024estimates}, we conclude that it is fully determined by the CFs in Eq.~\eqref{1357}.
Finally, in the coordinate representation, the  elements of $\boldsymbol{\rho}$ follow from the CF via Fourier transform
\begin{eqnarray}\label{1154}
  \!\!\!  \!\!\! \! \rho(y,y') = \frac{1}{2\pi}\int\phi\bigl(1;\mu,y-y'\bigr)\exp\Bigl(-i\,\tfrac{\mu(y+y')}{2}\Bigr)\,d\mu,
    \label{eqn: integral from matrix elements}
\end{eqnarray}
and similar expressions can be written in any basis. 
\par Thus, knowing the CF of a state allows one to fully characterize it and efficiently estimate both the distance and overlap between states, forming the basis of our KQSE framework. To achieve optimal convergence rates, we introduce the nonparametric KCFE as
\begin{eqnarray}\label{1041_1}
\hat{\phi}_{nh}(t;\mu,\nu)=\hat{\phi}_{X,n}(t;\mu,\nu) \phi_K(th;\mu,\nu)
\end{eqnarray}
with the optimal bandwidth selection. 
We see that the latter estimator is a product of two CFs: 
$\hat{\phi}_{X,n}(t;\mu,\nu)=\frac{1}{n} \sum_{j=1}^{n}\exp{(it X_{j,\mu,\nu})}$
depends on the measured data, while 
$\phi_K(th;\mu,\nu)\equiv\int_{-\infty}^{\infty}K_{\mu,\nu}(z)\exp{(ithz)}dz$
depends on the kernel's bandwidth parameter. The form of $\phi_{K}(t h;\mu,\nu)$ is known analytically because the chosen $K_{\mu,\nu}(z)$ has a known analytical expression.
\par 
Homodyne tomography typically neglects various losses in the detection process, such as mode mismatching and detector inefficiencies, which distort the observed distribution relative to the idealized case. In~\cite{leonhardt1995measuring}, the following simplified error model is introduced:
\begin{eqnarray}\label{1033}
   Z \equiv^d \kappa X + (1 - \kappa) Y, \quad 0 < \kappa < 1,
\end{eqnarray}
where $\equiv^d$ denotes equality in distribution. Here, $\kappa$ represents the overall detection efficiency, and $Y$ is an independent noise variable, typically modeled as i.i.d. Gaussian. While the true signal-noise interaction may be nonlinear and the noise non-Gaussian, the analysis of such models lies beyond the scope of this paper.
\par Our method requires the CF of the signal variable $X$ to be filtered from noise. If $X$ and $Y$ are independent, the PDF of their sum equals the convolution of the individual PDFs, and the CF of their sum is the product of individual CFs. Thus, if the CF of the noise $\phi_Y((1-\kappa)t)$ is invertible, the CF of the target variable is given by ${\phi}_X(t) = {\phi}_Z\left(\tfrac{t}{\kappa}\right)
    \phi_Y^{-1}\left(\tfrac{(1 - \kappa)t}{\kappa}\right)$, $ \kappa \in (0,1)$, 
where $\phi_Z(t)$ is the CF of the observable $Z$.
A similar idea arises in the deconvolution of signal densities from noisy data~\cite{stefanski1990deconvolving,fan1991optimal,delaigle2004practical}, where one typically assumes the CF of the noise satisfies the nonvanishing condition: $|\phi_Y(t)| > 0$, $\forall t$. In this deconvolution formula, division by $\phi_Y(t)$ is required, and zeros in the noise CF would render the inversion ill-posed or undefined. For instance, the CF of a normal distribution is strictly nonzero, ensuring stability, whereas distributions such as the uniform law have CFs that vanish at specific points, causing potential singularities.
\par 
Thus, the model \eqref{1033} and the assumed noise distribution determine the form of the correction in the CF estimation. Using \eqref{1041_1}, we define the kernel estimate of the signal CF $\phi_X$ filtered from the noise:
\begin{eqnarray}\label{CF_fil}
&&\hat{\phi}_{X,nh}( t ;\mu,\nu) = 
\hat{\phi}_{Z,nh}( t ;\mu,\nu)
\phi_Y^{-1}\left(\frac{(1-\kappa)t}{\kappa}\right),\\\nonumber
&&\hat{\phi}_{Z,nh}( t ;\mu,\nu) \equiv \hat{\phi}_{Z,n}\left(\frac{t}{\kappa};\mu,\nu\right)
\phi_K\left(\frac{t}{\kappa} h;\mu,\nu\right),
\end{eqnarray}
where $\phi_Y^{-1}$ is known from the noise model, and $\phi_K$ depends on the kernel choice.
\begin{thm}\label{thm4}
(Convergence of KCFE)
Let $\hat{\phi}_{X,nh}(t;\mu,\nu)$ be the KCFE of  $\phi(t;\mu,\nu)$, constructed from $n$ i.i.d. measurements of a fixed quadrature $X_{\mu,\nu}$ (or $Z_{\mu,\nu}$ in the noisy case). Then, for an optimally chosen bandwidth $h_{\phi} = O(n^{-1/2})$, the mean squared error satisfies
$\mathrm{MSE}(\hat{\phi}_{X,nh_{\phi}}) = O(n^{-1})$.
\end{thm}
% To our best knowledge, KCFE is largely unexplored outside signal processing, making our contribution relevant to both quantum tomography and nonparametric statistics.
The KCFE yields a discrete set of CF values, allowing only partial reconstruction of the integrals in Eqs.~\eqref{1357} and~\eqref{eqn: integral from matrix elements}. In KQSE, we decompose each integral into a bounded central region and tail contributions. The central part, supported on $|\mu| \leq \mu_{\max}, |\nu| \leq \nu_{\max}$, is estimated using KCFE combined with a discrete Fourier transform, while the tails are approximated based on the known asymptotic decay of the CF. This decomposition is justified by the rapid decrease of quadrature distributions~\cite{lvovsky2002quantumoptical,markovich2024not}, which ensures that truncation errors are well-controlled. The convergence rate of KQSE is established in Theorem~\ref{thm0}, and a detailed error analysis including finite-window effects, discretization errors, and KCFE accuracy is given in the Supplemental Material.
\\\textit{Simulation Study.}
Let the observable of interest $X$ be generated by the coherent cat state (CCS): $\ket{\psi}_{\mathrm{CCS}}=N_c\sum_{j=1}^{3} \ket{a_j}$, where $\ket{a_j} = \ket{a\, \exp{(\tfrac{2\pi i (j-1)}{3})}}$, $j = 1, 2, 3$, with $\ket{a}$ being Glauber coherent states of the HO and $N_c$ the normalization constant. The corresponding tomogram is multimodal, non-Gaussian PDF, with a Gaussian tail. We select  $\mu = 0.8$, $\nu = 1.2$, and $a = 1.0 + 0.5i$ and the KDE of this tomogram is shown in Fig.~\ref{fig:1}.
\begin{figure}[h]
\includegraphics[width=0.7\linewidth]{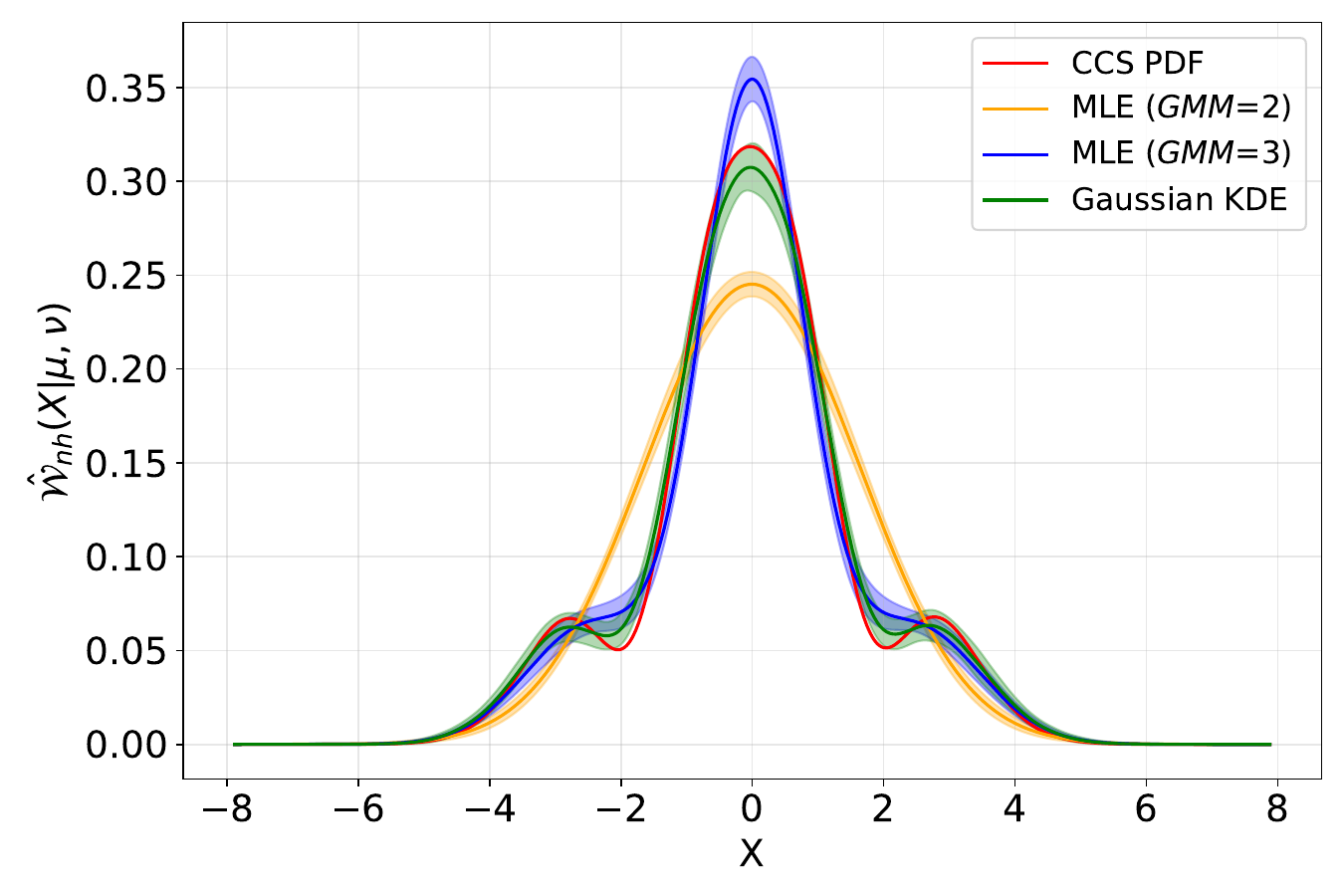}\caption{
Comparison of CCS tomogram reconstructions using MLE ($GMM=2,3$) and Gaussian KDE for $n=1000$ with $\mu=0.8$, $\nu=1.2$, and $a=1+0.5i$.
}
%{Figures/fig_2_main.pdf}
%\caption{KDEs of the coherent cats state tomogram with $\mu = 0.8$, $\nu = 1.2$, and $a = 1.0 + 0.5i$ for a sample size $n=500$.}
  \label{fig:1}
\end{figure}
%The KDE estimation is performed using the Gaussian and Epanechnikov kernels with bandwidth selected via cross-validation and adaptive Silverman's rule of thumb~\cite{silverman2018density}. 
\textcolor{black}{The KDE estimation is performed using a Gaussian kernel with the bandwidth selected via cross-validation. For comparison, we also perform MLE-based reconstructions under different parametric assumptions, modeling the tomogram as a Gaussian mixture. In particular, we consider a two-component mixture ($GMM=2$), which represents a deliberately misspecified model, and a three-component mixture ($GMM=3$), which provides a more adequate description. This comparison is used to illustrate the effect of model mismatch on the performance of MLE-based estimators. The symmetrical kernel selection is not important for the asymptotics, however changing the bias of the estimate (see Supplemental Material).}
\par \textcolor{black}{Next we consider the case where the observable data follows the mixture~\eqref{1033}, with the noise component $Y$ drawn from a normal distribution $N(0,1)$ and fixed $\kappa = 0.85$. The CF is estimated via~\eqref{CF_fil}, using $N_{\mu} = 160$ evaluation points uniformly sampled in $\mu \in [0,8]$, with $\nu=1.2$, as shown in Fig.~\ref{fig:2}. A KCFE with the Gaussian kernel and the optimal bandwidth $h_{\phi}=O(n^{-1/2})$ is used. For comparison, CF estimates obtained by Fourier transforming the MLE-based tomograms with $GMM=2$ and $GMM=3$ are also shown. The built-in correction procedure is observed to significantly enhance the estimation performance.
}
\begin{figure}[h]
\includegraphics[width=0.7\linewidth]{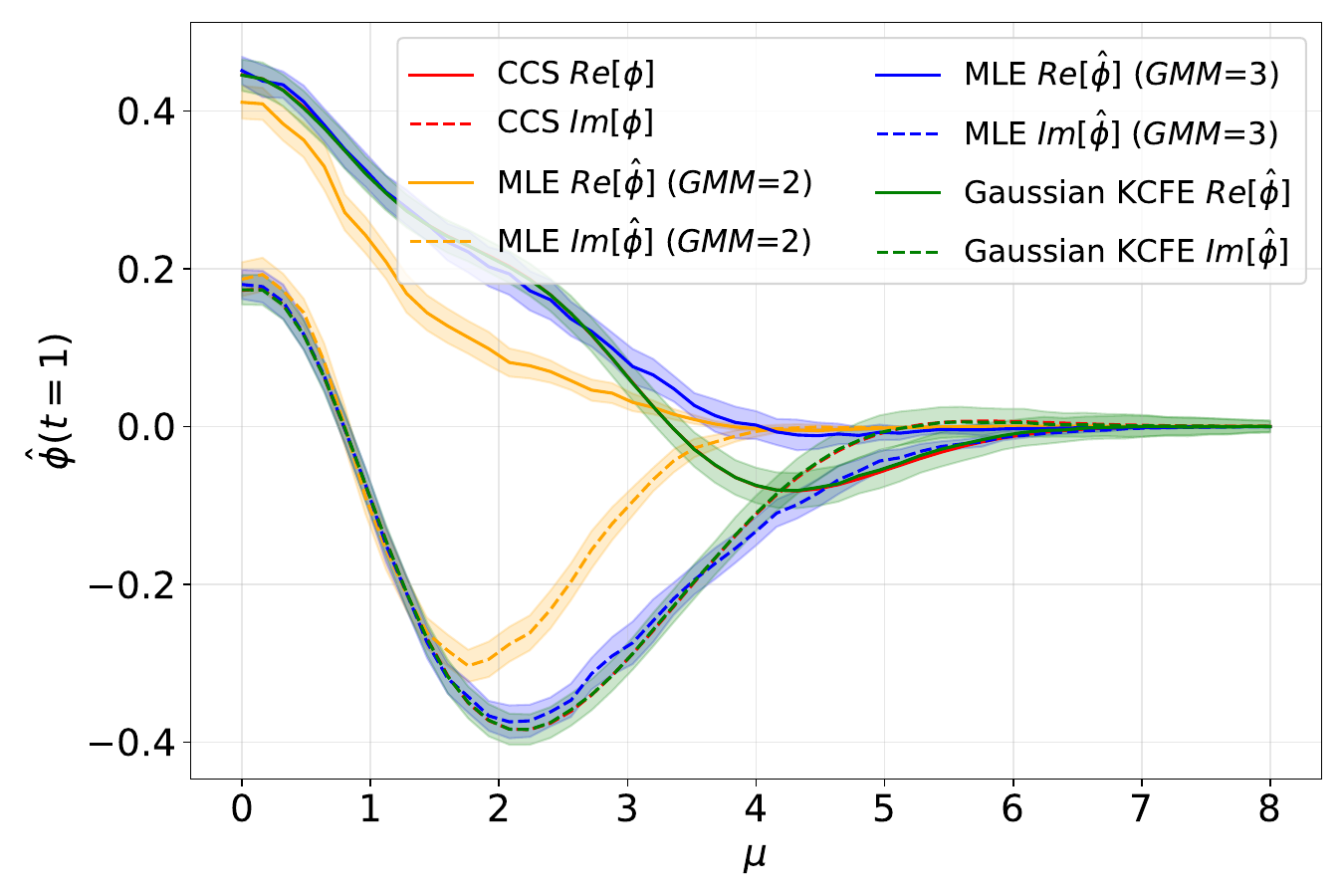}
  \caption{Comparison of CCS CF estimator obtained using MLE ($GMM=2,3$) and KCFE for $n=1000$ and  $\nu=1.2$.}%{Figures/fig_3_main.PDF} \caption{KCFE of the coherent cats state CF estimated from a noisy data without/with correction for the  sample size $n=500$, $\kappa = 0.85$,  $\nu=1.2$ from $\mu\in[0, 8]$ for $N_{\mu}=160$ points.}
  \label{fig:2}
\end{figure}
\par Using the set of KCFEs $\{\hat{\phi}_{nh}(1;\mu_i,y-y')\}_{i=1}^{N_{\mu}}$, we use KQSE procedure to reconstruct the kernel \eqref{1154} of cat state density operator. With the set  $\{\hat{\phi}_{nh}(1;\mu_i,\nu_j)\}_{i,j=1}^{N_{\mu},N_{\nu}}$  KQSE reconstructs the overlap \eqref{1357} between the known coherent cat state $\boldsymbol{\rho}$ and its estimate from noisy data $\boldsymbol{\rho}_n$ for various total sample sizes $T_{\mu,\nu}$. The results are provided in Tab.~\ref{tab_1}. Among the parameters varied, the sample size $n$ and the number of measurement settings $N_{\mu}$, $N_{\nu}$ exhibit the strongest influence on the reconstruction accuracy, yielding  reduction in the total errors, while the effect of increasing $\mu_{\max}$, $\nu_{\max}$ is slightly weaker. A detailed analysis of the parameter sensitivities and a comparison with MLE-based methods are provided in the Supplemental Material.
\par \textcolor{black}{In addition to numerical simulations, we validate the proposed KQSE framework using experimentally measured homodyne-tomographic data. Specifically, we analyze non-Gaussian optical data corresponding to a Schr{\"o}dinger kitten state obtained via conditional measurements in a quantum-optical experiment [see Ref.~\cite{Lvovsky_2002,Lvovsky_2004}]. From the experimental quadrature samples, we reconstruct the tomogram, CF, and state kernel using KQSE, and compare the results with MLE under both correctly specified and misspecified parametric models. KQSE consistently yields the lowest reconstruction errors for both simulated and experimental data, outperforming MLE-based approaches in the presence of model mismatch and noise.
 Full details of the experimental dataset, reconstruction procedure, and quantitative comparison are presented in the Supplementary Material.}
\begin{table}[b]
\caption{\label{tab_1}%
Uniform and MSE errors of the KQSE reconstruction for fixed $\mu_{\max}=6$ with total number of measurements $T_{\mu}=T_{\mu,\nu}=T$.}
\begin{ruledtabular}
\begin{tabular}{cccc}
$T$ & Data type & $L_{\infty}(\widehat{\rho}(y,y'))$ & $\mathrm{MSE}(\widehat{\mathrm{tr}(\boldsymbol{\rho}\boldsymbol{\rho}_n)})$ \\
\colrule
                 & Noiseless & $1.781\times10^{-4}$ & $1.048\times10^{-4}$ \\
$5.40\times10^5$ & Noisy     & $1.407\times10^{-2}$ & $1.116\times10^{-2}$ \\
                 & Corrected & $1.814\times10^{-4}$ & $1.104\times10^{-4}$ \\
\colrule
                 & Noiseless & $9.114\times10^{-5}$ & $8.742\times10^{-5}$ \\
$7.70\times10^5$ & Noisy     & $1.393\times10^{-2}$ & $1.071\times10^{-2}$ \\
                 & Corrected & $9.319\times10^{-5}$ & $9.382\times10^{-5}$ \\
\colrule
                 & Noiseless & $4.183\times10^{-5}$ & $6.322\times10^{-5}$ \\
$1.21\times10^6$ & Noisy     & $1.354\times10^{-2}$ & $1.008\times10^{-2}$ \\
                 & Corrected & $4.219\times10^{-5}$ & $7.152\times10^{-5}$ \\
\end{tabular}
\end{ruledtabular}
\end{table}

\textit{Conclusions.} We study CV quantum states within the framework of the probabilistic tomographic representation of quantum mechanics. Unlike QPDFs, the quantum tomogram is a genuine PDF, depending on a directly observable random variable. We prove that the trace distance between two quantum states is bounded from below by the total variation distance between their tomograms. We introduce a KDE for the tomogram and show that, given a single set of tomographic parameters $\mu,\nu$ and $n = O(\varepsilon^{-5/2})$ samples, one can obtain a lower bound on the trace distance with precision $\varepsilon$.  Furthermore, we demonstrate that the CF of the tomogram fully determines the quantum state and all its trace-based quantities.
By proposing a novel KCFE, we introduce our KQSE method for reconstructing quantum states and their trace characteristics from noisy measurement data. Our estimator achieves a near-optimal convergence rate in squared distance of $\tilde{O}(1/T)$, where $T$ denotes the total number of samples. The method is universal: it performs reliably for arbitrary states, including non-Gaussian multimodal ones, which are crucial for achieving universal quantum computation~\cite{menicucci2006universal, takagi2018convex, hanamura2024implementing}, entanglement distillation~\cite{yokoyama2013ultralargescale}, quantum error correction~\cite{gottesman2001encoding,sivak2023realtime}, and metrological advantage~\cite{xu2022metrological}.
\par In contrast to existing methods, our approach is independent of the shape of the tomographic distribution. While many non-Gaussian states retain Gaussian tails, our technique remains robust for any light-tailed PDF (and even for heavy-tailed distributions with a proper $K_{\mu,\nu}(\cdot)$ selection~\cite{markovich2018light}), which frequently arise in practical scenarios and are relevant for quantum modeling. 
\textcolor{black}{Our future work will address the extension of KQSE to multivariate states by leveraging established multivariate kernel methods and their known convergence properties~\cite{ChaconDuong2020}.}
%Finally, the proposed nonparametric framework naturally extends to multivariate and dependent data settings, offering a promising direction for our future work. 

\textit{Acknowledgments.} 
The authors would like to thank V.I. Man'ko and
A.S. Holevo for insightful comments and corrections.
We thank A. Lvovsky for providing the experimental homodyne-tomographic data used in this work and for valuable discussions. The numerical experiments of this work were performed using the compute resources from the Academic Leiden Interdisciplinary Cluster Environment (ALICE) provided by Leiden University. 
L.M. was  supported by the Netherlands Organisation for Scientific Research (NWO/OCW), as part of the Quantum Software Consortium program (project number 024.003.037 / 3368).
This work has received support from the European Union’s Horizon Europe research and innovation programme through the ERC StG FINE-TEA-SQUAD (Grant No. 101040729). This work is supported by the Dutch National Growth Fund (NGF), as part of the Quantum Delta NL programme. 
Funded by the European Union. Views and opinions expressed are however, those of the author(s) only and do not necessarily reflect those of the European Union or the European Commission. Neither the European Union nor the granting authority can be held responsible for them.

%\subsection{\label{sec:citeref}Citations and References}
%\nocite{*}
\bibliographystyle{apsrev4-2}
\bibliography{tomogram}% Produces the bibliography via BibTeX.

\appendix

\onecolumngrid 
\tableofcontents

\section*{S1. Operator Symbols and Tomographic Probability Representation of Quantum States}\label{sup_1}
\subsection*{1.1. Quantizer and Dequantizer Formalism}
\par Let $\mathcal{H}$ be an infinite-dimensional Hilbert space, and $\mathcal{B}(\mathcal{H})$ the set of linear operators on $\mathcal{H}$. In quantum mechanics, physical observables $A$ correspond to Hermitian operators $\boldsymbol{A} \in \mathcal{B}(\mathcal{H})$. Such operator can be represented as a matrix with complex elements.
 Consider a set of Hermitian positive definite operators $\boldsymbol{U}(z) \in \mathcal{B}(\mathcal{H})$, called the \textit{dequantizer}, where $z$ takes continuous values. We define the function $ f_A(z) $ as  
\begin{eqnarray}\label{0}
     f_A(z) = \operatorname{tr} (\boldsymbol{A} \boldsymbol{U}(z)).
\end{eqnarray}
This function, known as the \textit{symbol} of operator $\boldsymbol{A}$, establishes a linear mapping from the operator space to the space of functions.
If \eqref{0} establishes a one-to-one correspondence between the symbol and the operator, the inverse transform is given by  
\begin{eqnarray}\label{2}
   \boldsymbol{A}=\int  f_A(z)\boldsymbol{D}(z) dz,
\end{eqnarray}
where the \textit{quantizer} operators $ \boldsymbol{D}(z) \in \mathcal{B}(\mathcal{H}) $ enable the reconstruction of $ \boldsymbol{A} $ from its symbol $ f_A(z) $.
If the Hilbert space $ \mathcal{H}_d $ has finite dimension $ d $, the variable $ z $ is discrete, and $ \boldsymbol{U}(z), \boldsymbol{D}(z) $ are linear bounded operators on $ \mathcal{H}_d $. If $ z $ takes $ N $ values (outcomes), \eqref{2} becomes  
\begin{eqnarray}
    \boldsymbol{A}=\sum\limits_{k=1}^{N} f_A(k) \boldsymbol{D}(k).
\end{eqnarray}
For the inverse map to exist, the set of dequantizers $ \boldsymbol{U}(z) $ must contain $ d^2 $ linearly independent operators, with $ N \geq d^2 $.
 It is easy to check that the pair of quantizer and dequantizer operators satisfy the condition
\begin{eqnarray}
    \tr{\boldsymbol{U}(z') \boldsymbol{D}(z)}=\delta(z-z').
\end{eqnarray}
\par The choice of the quantizer-dequantizer pair for $\boldsymbol{A}$ is not unique. In general, the only requirement for the dequantizer is to be Hermitian, but in this case, the symbol is not a probability density function (PDF). For example, it can take negative values, as the Wigner function. This does not allow for a direct probabilistic interpretation of measurements.  Imposing the additional condition that $\boldsymbol{U}(z)$ is positive semi-definite and form a resolution of the identity on $\mathcal{H}$ makes the dequantizer a positive-operator valued measure (POVM). In this case, \eqref{0} defines the PDF of $z$. 

\subsection*{1.2. Tomographic Probability Density Function}
\par We consider the quantum state in an infinite-dimensional Hilbert space $ \mathcal{H} $, represented by a positive Hermitian operator $ \boldsymbol{\rho} \in \mathcal{B}(\mathcal{H}) $, known as the \textit{density matrix operator}. By definition, $ \boldsymbol{\rho} $ is Hermitian, positive semi-definite, and has unit trace.  The kernel of $ \boldsymbol{\rho} $ is given by  
$\rho{(\psi,\psi')}=\bra{\psi}\boldsymbol{\rho}\ket{\psi'}$
where $ \ket{\psi} $ and $ \ket{\psi'} $ are quantum states.
 \par Let us define the dequantizer and quantizer as  
\begin{eqnarray}\label{1439}
    \boldsymbol{U}_{\mu,\nu}(x)=\delta\left(\hat{1} x-\mu \boldsymbol{q}-\nu\boldsymbol{p}\right),\quad \boldsymbol{D}_{\mu,\nu}(x)= (2\pi)^{-1}\exp{i(x\hat{1}-\mu\boldsymbol{q}-\nu\boldsymbol{p})},\end{eqnarray}
where $\mu, \nu \in \mathbb{R}$, and $\boldsymbol{q}$, $\boldsymbol{p}$ are position and momentum operators.  We verify that $\boldsymbol{U}_{\mu,\nu}(x)$ is Hermitian, positive semi-definite, and satisfies  
\begin{eqnarray}
    \int\limits_{-\infty}^{\infty} \boldsymbol{U}_{\mu,\nu}(x) dx = \hat{1}.
\end{eqnarray}
Here, the Dirac delta function of an operator is given by  
\begin{eqnarray}
    \delta(\boldsymbol{A}) = \frac{1}{2\pi} \int e^{ik\boldsymbol{A}} dk.
\end{eqnarray}
Thus, $\boldsymbol{U}_{\mu,\nu}(x)$ forms a POVM. 
Using \eqref{1439} in the correspondence \eqref{0}, the mapping of $\boldsymbol{\rho}$ onto a parametric set of PDFs, depending on two real parameters $\mu$ and $\nu$, is given by~\cite{mancini1996symplectic}:
\begin{eqnarray}\label{1}
\mathcal{W}(x|\mu,\nu)=\tr{\boldsymbol{\rho} \boldsymbol{U}_{\mu,\nu}(x)}=\langle x;\mu,\nu|\boldsymbol{\rho}|x;\mu,\nu\rangle,\quad\!\! 
\boldsymbol{\rho}=\iiint\limits_{-\infty}^{\infty} \mathcal{W}(x|\mu,\nu) \boldsymbol{D}_{\mu,\nu}(x)dxd\mu d\nu,
\end{eqnarray}
where $|x;\mu,\nu\rangle$ is an eigenvector of the Hermitian operator  
\begin{eqnarray}\label{1051_2}
    \boldsymbol{X}_{\mu,\nu} = \mu \boldsymbol{q} + \nu \boldsymbol{p},
\end{eqnarray}
called the \textit{quadrature}.  A continuous random variable $X$ corresponds to \eqref{1051_2}, with sample space $S_X \subset \mathbb{R}$.  
The function $\mathcal{W}(x|\mu,\nu)$, known as the \textit{symplectic tomogram}, represents the \textit{conditional PDF} over $S_X$, dependent on $\mu$ and $\nu$. As a PDF, the tomogram is nonnegative and normalized:  
\begin{eqnarray}
\int\limits_{-\infty}^{\infty} \mathcal{W}(x|\mu,\nu) dx = 1.
\end{eqnarray}
Since the delta function is homogeneous, $\delta(\lambda y)=|\lambda|^{-1}\delta(y)$, for any constant $\lambda$, the tomogram inherits the homogeneity property:
\begin{eqnarray}\label{1453}
    \mathcal{W}(\lambda x|\lambda\mu,\lambda\nu)=|\lambda|^{-1}\mathcal{W}(x|\mu,\nu),\quad    \mathcal{W}(y| \lambda \mu, \lambda \nu) =
     |\lambda|^{-1}
     \mathcal{W}(y\lambda^{-1}|\mu,\nu),\quad y=\lambda x.
\end{eqnarray}
\par  Since tomograms are PDFs, the \textit{total variation distance} between $\mathcal{W}_1(x|\mu,\nu)$ and $\mathcal{W}_2(x|\mu,\nu)$ over a continuous variable $x$ is defined as  
\begin{eqnarray}\label{totalvar}
 T(\mathcal{W}_1, \mathcal{W}_2)  = \frac{1}{2} \int \Big|\mathcal{W}_1(x|\mu,\nu) - \mathcal{W}_2(x|\mu,\nu)\Big| \, dx.
\end{eqnarray}
By Scheffé’s theorem, the $L_1$-distance between two PDFs equals twice the total variation:  
$L_1(W_1,W_2)=2 T(W_1,W_2)$. The total variation distance ranges from $0$ to $1$, where $0$ indicates identical distributions and $1$ means no overlap.
\par Our result establishes a connection between the total variation of tomograms and the trace distance of the corresponding density operators. The following theorem holds:
\begin{thm}\label{thm1}
  The trace distance  $D(\boldsymbol{\rho}_1,\boldsymbol{\rho}_2)=\frac{1}{2}\tr{|\boldsymbol{\rho}_1-\boldsymbol{\rho}_2|}=1/2 \|\boldsymbol{\rho}_1-\boldsymbol{\rho}_2\|_1$ between the density matrices $\boldsymbol{\rho}_1$ and  $\boldsymbol{\rho}_2$ of two quantum CV states is   
 \begin{eqnarray}
   T(\mathcal{W}_1(x|\mu,\nu), \mathcal{W}_2(x|\mu,\nu))\leq D(\boldsymbol{\rho}_1,\boldsymbol{\rho}_2).
  \end{eqnarray}
\end{thm}
\begin{proof}
 We write
   \begin{eqnarray}
&& T(\mathcal{W}_1(x|\mu,\nu), \mathcal{W}_2(x|\mu,\nu))=\frac{1}{2} \int |\tr(\boldsymbol{U}_{\mu,\nu}(x)(\boldsymbol{\rho}_1 - \boldsymbol{\rho}_2))|dx \\\nonumber
&\leq &\frac{1}{2} \int \tr(\boldsymbol{U}_{\mu,\nu}(x)|\boldsymbol{\rho}_1 - \boldsymbol{\rho}_2|)dx =\frac{1}{2}  \tr\left(\int\boldsymbol{U}_{\mu,\nu}(x)dx|\boldsymbol{\rho}_1 - \boldsymbol{\rho}_2|\right)=D(\boldsymbol{\rho}_1,\boldsymbol{\rho}_2),
\end{eqnarray} 
where we used the POVM properties of the dequantizer. 
%The inequality compares tomographic distinguishability (the classical distinguishability of measured data) to quantum distinguishability (trace distance).
%The trace distance is can be written as  
%\begin{eqnarray}\label{957}
%  \|\boldsymbol{\rho}_1-\boldsymbol{\rho}_2\|_1 &=& \frac{1}{2\pi} \Bigg\|\int (\mathcal{W}_1(x|\mu,\nu)- \mathcal{W}_2(x|\mu,\nu))\boldsymbol{D}_{\mu,\nu}(x) \, dx d\mu d\nu\Bigg\|_1\\\nonumber
 %   &\leq& \frac{1}{2\pi}\int |\mathcal{W}_1(x|\mu,\nu)- \mathcal{W}_2(x|\mu,\nu)|\big\|\boldsymbol{D}_{\mu,\nu}(x) \big\|_1\, dx d\mu d\nu,
%\end{eqnarray}
%where we used the triangle inequality for integrals of operators $\|\int f\boldsymbol{D}\|_1\leq \int |f| \|\boldsymbol{D}\|_1$.
%Using \eqref{1026}, we conclude
%\begin{eqnarray}
 %   \frac{1}{2}\|\boldsymbol{\rho}_1 - \boldsymbol{\rho}_2\|_1 
%\leq \frac{1}{8\pi} \int \left| \mathcal{W}_1(x|\mu,\nu) - \mathcal{W}_2(x|\mu,\nu) \right| dx \, d\mu \, d\nu=\frac{1}{4\pi}\int T(\mathcal{W}_1(x|\mu,\nu), \mathcal{W}_2(x|\mu,\nu))d\mu d\nu.
%\end{eqnarray}
%The trace distance between the quantum states is controlled by the average total variation of their tomograms over the measurement directions.
\end{proof} 
\begin{example}
       The wave function of the excited state of the harmonic oscillator (HO) is
\begin{eqnarray}\Psi_m(y)=\frac{1}{\sqrt[4]{\pi}\sqrt{2^m m!}}e^{-y^2/2}H_m(y),\end{eqnarray}
where $H_m(y)$ is the Hermite polynomial. The tomogram and CF of the excited oscillator state are the following
\begin{eqnarray}\label{1100}
     \mathcal{W}_m(x|\beta)
    &=& \mathcal{W}_0(x|\beta) \frac{1}{2^m m!} H^2_m \left(\frac{x}{\beta}\right),\quad 
    {\phi}_m(1;\beta)
    = {\phi}_0(1;\beta) L_m\left(\frac{\beta^2}{2}\right),
\end{eqnarray}
where the tomogram and CF of the ground state  are 
\begin{eqnarray}\label{1530}
\mathcal{W}_0(x|\beta) = \frac{1}{\sqrt{\pi}\beta} \exp\!\Bigl[-\tfrac{x^2}{\beta^2}\Bigr], \quad 
{\phi}_0(t;\beta) = \exp\!\Bigl(-\tfrac{t^2\beta^2}{4}\Bigr),\quad \beta\equiv\sqrt{\mu^2+\nu^2}.
\end{eqnarray}
Here the Laguerre polynomial series decomposition is
\begin{eqnarray}
    L_m(x)=\sum\limits_{k=0}^m\frac{(-1)^km!}{(k!)^2(m-k)!}x^k.
\end{eqnarray}
We conclude that any quantum state formed as a superposition of HO eigenstates exhibits Gaussian tails, both in its wavefunction and in the corresponding CF with respect to the parameters $\mu$ and $\nu$. A more detailed analysis of the families of probability distributions arising in experimental settings is provided in Ref.~\cite{markovich2024not}.
\end{example}

\subsection*{1.3. Characteristic Function of the Tomogram}
\par In~\cite{markovich2024not}, the \textit{characteristic function} (CF) of the tomogram $ \mathcal{W}(x|\mu,\nu)$ is defined as  
\begin{eqnarray}\label{1113}
    \phi(t;\mu,\nu) = \int\limits_{-\infty}^{\infty} \mathcal{W}(x|\mu,\nu) e^{itx}dx.
\end{eqnarray}
Since the PDF is always integrable, the CF always exists. It is called the CF because it fully characterizes the distribution.  
The function $\phi(t)$ is continuous for all real $t$ and satisfies $\phi(0) = 1$, $|\phi(t)| \leq 1$, and $\phi(-t) = \phi^{\star}(t)$. By Theorem 1 in~\cite{polya1949remarks}, if the corresponding probability density is even and continuous everywhere except possibly at $x = 0$, then $\phi(t)$ has additional properties: it is real-valued, $\lim\limits_{t\rightarrow \infty} \phi(t) = 0$, even, and convex for $t > 0$.

From \eqref{1453} one can derive the following property
\begin{eqnarray}\label{1118}
   && \phi(t; \lambda\mu, \lambda\nu) = \phi(t\lambda; \mu, \nu),\quad \phi(t \lambda^{-1}; \lambda\mu, \lambda\nu) = \phi(t; \mu, \nu).
\end{eqnarray}
Using the CF, one can rewrite  \eqref{1}  as follows:
\begin{eqnarray}\label{1050}
 \phi(1;\mu,\nu)=\tr{\boldsymbol{\rho} e^{i\boldsymbol{X}_{\mu,\nu}}},\quad    \boldsymbol{\rho}=\frac{1}{2\pi} \int\limits_{-\infty}^{\infty} \phi(1;\mu,\nu)e^{-i\boldsymbol{X}_{\mu,\nu}}d\mu d\nu.
\end{eqnarray}
 As noted earlier, for a symmetric random variable around the origin, the CF is real and even. Since the Fourier transform of a real even function remains real and even, the resulting operator also inherits these properties.
\par Using the fact that $\tr{e^{-i\boldsymbol{X}_{\mu,\nu}}}=\tr(e^{-i(\nu\boldsymbol{p}+\mu\boldsymbol{q})})=2\pi e^{i\frac{\mu\nu}{2}}\delta(\mu)\delta(\nu)$ holds, one can derive the connections between the trace characteristics of the quantum state and the CF of the tomogram:
\begin{eqnarray}\label{1051}
   &&\tr{ \boldsymbol{\rho}}= \phi(1;0,0),\\\nonumber
   && \tr{(\boldsymbol{\rho}^2)} %=\frac{1}{2\pi} \iint \limits_{-\infty}^{\infty} \phi(1;\mu,\nu) \phi(1;-\mu,-\nu)d\mu d\nu
   =\frac{1}{2\pi} \iint \limits_{-\infty}^{\infty} |\phi(1;\mu,\nu)|^2 d\mu d\nu,\\\nonumber
   &&\tr{\left(\boldsymbol{\rho}^d\right)}\! =\!
   \frac{1}{(2\pi)^{d-1}} \iint \limits_{-\infty}^{\infty} \prod\limits_{i=1}^{d-1}\phi(1;\mu_i,\nu_i) \phi^{\star}\left(1;\sum\limits_{i=1}^{d-1}\mu_i,\sum\limits_{i=1}^{d-1}\nu_i\right)\\\nonumber
   &\times&\exp\!\Biggl[-\frac{i}{2}\sum_{1\le j<k\le d-1} (\mu_j\nu_k-\nu_j\mu_k)\Biggr]d\mu_1\dots d\mu_{d-1} d\nu_1\dots d\nu_{d-1}.
   %\\\nonumber &=&
   %\frac{1}{(2\pi)^{d-1}} \iint_{-\infty}^{\infty} \phi(1;\mu,\nu)^{d-1} \, \phi^{\star}\left(d-1;\mu,\nu\right) \, d\mu\, d\nu, \quad d>1.
\end{eqnarray}
One can see that the density matrices $\boldsymbol{\rho}_1$ and $\boldsymbol{\rho}_2$ of two quantum states satisfy the \textit{noncommutative Parseval relation}:
   \begin{eqnarray}\label{1357}
   \tr{( \boldsymbol{ \rho}_1\boldsymbol{\rho}_2)}
  &=& \frac{1}{2\pi}\iint \phi_1(1;\mu,\nu) \phi_2^{\star}(1;\mu,\nu) d\mu d\nu\\\nonumber
  &=&
  \frac{1}{2\pi}\iint\limits_{-\infty}^{\infty}  \tr{\left[\boldsymbol{\rho}_1 e^{i(\mu\boldsymbol{q}+\nu\boldsymbol{p})}\right]}  \left[\tr{\left[\boldsymbol{\rho}_2 e^{i(\mu\boldsymbol{q}+\nu\boldsymbol{p})}\right]}\right]^{\star} d\mu d\nu, 
\end{eqnarray}
where $\phi_1(1;\mu_1,\nu_1)$ and $\phi_2(1;\mu_1,\nu_1)$ are the CFs corresponding to the tomograms of the states.  
This result can be easily extended to the product of $d$ density matrices:
\begin{eqnarray}
&&\tr{\left(\prod\limits_{i=1}^{d}\boldsymbol{\rho}_i\right)}\! =\!\frac{1}{(2\pi)^{d-1}} \iint \limits_{-\infty}^{\infty} \prod\limits_{i=1}^{d-1}\phi_i(1;\mu_i,\nu_i) \phi_d^{\star}\Bigl(1; \sum_{j=1}^{d-1}\mu_j,\,\sum_{j=1}^{d-1}\nu_j\Bigr)\\\nonumber
&\times&\exp\!\Biggl[-\frac{i}{2}\sum_{1\le j<k\le d-1} (\mu_j\nu_k-\nu_j\mu_k)\Biggr]  d\mu_1\, d\mu_2\cdots d\mu_{d-1}\,d\nu_1\, d\nu_2\cdots d\nu_{d-1},
\end{eqnarray}
where $\phi_i(t;\mu_i,\nu_i)$ are the CFs of the corresponding states $\boldsymbol{\rho}_i$. The trace from the square rooted density matrices can be expressed in terms of CF as
\begin{eqnarray*}
\tr{\sqrt{\boldsymbol{\rho}_1} \sqrt{\boldsymbol{\rho}_2}}=\frac{1}{2\pi}\tr{\left(\!\Big[\!\!\iint\limits_{-\infty}^{\infty}\!\! \phi_{1}(1;\mu_1,\nu_1)\phi_{2}(1;\mu_2,\nu_2)e^{-i(\nu_1\boldsymbol{p}+\mu_1\boldsymbol{q})}\,e^{-i(\nu_2\boldsymbol{p}+\mu_2\boldsymbol{q})}d\mu_1 d\nu_1  d\mu_2 d\nu_2\Big]^{\frac{1}{2}}\!\right)},
\end{eqnarray*}
providing the upper and lover bounds in the following inequality for the trace distance~\cite{kholevo1972quasiequivalence}:
\begin{eqnarray}
1 - \mathrm{tr} \sqrt{\boldsymbol{\rho}_1} \sqrt{\boldsymbol{\rho}_2}  \leq D(\boldsymbol{\rho}_1,\boldsymbol{\rho}_2) \leq  \sqrt{1 - \left( \mathrm{tr} \sqrt{\boldsymbol{\rho}_1} \sqrt{\boldsymbol{\rho}_2} \right)^2 }.
\end{eqnarray}
\par The kernel element $\rho(\chi,\chi)=\langle \chi|\boldsymbol{\rho}|\chi\rangle=\tr{[\boldsymbol{\rho}\boldsymbol{\rho}_{\chi}]}$ of the density operator in any state $\ket{\chi}$ representation can be written as
\begin{eqnarray}\label{1514_1}
    \rho( \chi, \chi)
    &=&\frac{1}{2\pi}
    \iint \phi(1;\mu,\nu)\phi_{{\rho}_{\chi}}^{\star}(1;\mu,\nu)d\mu d\nu,
\end{eqnarray}
where $\phi_{{\rho}_{\chi}}^{\star}(1;\mu,\nu)$ is the CF for the tomogram corresponding to the quantum state $\ket{\chi}$. Since $\boldsymbol{\rho}$ is a density matrix, it is positive semidefinite and its matrix elements can not be bigger than one, then $0\leq  \rho( \chi, \chi)\leq 1$, holds $\forall \ket{\chi}\neq 0$.
\par In the coordinate representation ($ \rho(y,y')=\langle y|\boldsymbol{\rho}|y'\rangle$), the kernel elements of $\boldsymbol{\rho}$ can be written as a Fourier transform of the CF:
\begin{eqnarray}
    \rho(y,y')\!=\!\frac{1}{2\pi}\int\limits_{-\infty}^{\infty} \phi(1;\mu,y-y')\exp{\left(-i\frac{\mu(y+y')}{2}\right)}  d\mu\equiv\frac{1}{\sqrt{2\pi}}\overline{\phi}\left(1;\frac{y+y'}{2},y-y'\right),
    \label{eqn: integral from matrix elements}
\end{eqnarray}
where $\overline{\phi}(1;\tilde{\mu},{\nu})$ denotes the Fourier transform of $\phi(1;\mu,\nu)$. 
Hence, to reconstruct the latter kernel of the density matrix operator, one needs to estimate the CF for a fixed parameter $\nu=y-y'$, varying the parameter $\mu$.
\subsection*{1.4. Tomograms for Pure States}
For the density operator $\boldsymbol{\rho}=\ket{\psi}\bra{\psi}$ of the pure state $\ket{\psi}$, the relation \eqref{1} converts into a simpler one
\begin{eqnarray}\label{1645_5}
    \mathcal{W}(x|\mu,\nu)= \frac{1}{2\pi|\nu|}|(\mathcal{F}_{\mu,\nu}\Psi)(x)|^2,\quad (\mathcal{F}_{\mu,\nu}\Psi)(x)=
\int\limits_{-\infty}^{\infty} \Psi(y)\exp{\left(\frac{i\mu}{2\nu}y^2-\frac{ix}{\nu}y\right)} dy.
\end{eqnarray}
Here $(\mathcal{F}_{\mu,\nu}\Psi)(x)$ is a two-parameter set of unitary transforms in the space $\mathcal{H}$.
One can see that if $\mu=0$ and $\nu=1$, the latter is a Fourier transform. 
\par In the case of general pure states there is exact expression for the trace distance between two density operators~\cite{holevo2024estimates}:
\begin{eqnarray}
D(\boldsymbol{\rho}_1,\boldsymbol{\rho}_2)=\sqrt{1-\tr{(\boldsymbol{\rho}_1\boldsymbol{\rho}_2)}}=\sqrt{1-\frac{1}{2\pi}\iint \phi_1(1;\mu,\nu) \phi_2^{\star}(1;\mu,\nu) d\mu d\nu}.
\end{eqnarray}
In this case the knowledge of the CFs can provide the trace distance between two pure quantum states. 
\subsection*{1.5. Optical Tomograms}
\par The optical tomogram is a special case of the symplectic tomogram. Using the change of variables 
\begin{eqnarray}\label{1359}
\mu=r \cos\theta, \quad \nu=r \sin\theta,\quad x=ry, \quad r\in\mathbb{R}^+,\quad \theta\in[0,2\pi],
\end{eqnarray}
we can write the density operator \eqref{1} as follows
\begin{eqnarray}
\boldsymbol{\rho}=\frac{1}{2\pi}\int\limits_{-\infty}^{\infty} \int\limits_{0}^{\infty}\int\limits_{0}^{2\pi}r\mathcal{W}(y|  \cos{\theta},  \sin{\theta}) \exp{(ir(y-\boldsymbol{q}\cos{\theta}-\boldsymbol{p}\sin{\theta}))}dy d rd\theta
\end{eqnarray}
where $ \mathcal{W}(y|  \cos{\theta},  \sin{\theta}) $ is called \textit{optical tomogram}.
%The connection between two tomograms is
%\begin{eqnarray}\label{1401}
%   |r|^{-1} \mathcal{W}(y| \cos{\theta},\sin{\theta})=\mathcal{W}(x| \mu,\nu ).
%\end{eqnarray} 
\par Experimentally one uses the homodyne detection for a state reconstruction. The quantum system to be measured is laser light with a fixed frequency~\cite{wu1998optical,lvovsky2004iterative,dariano2007homodyne}.  A coherent laser source provides the input light field. The key component of the setup is a homodyne detector. The input signal is mixed by a $50-50$ beam-splitter  with a strong local oscillator,  which is coherent with the input field
and is in a strong coherent state.    The relative phase $\theta\in[0,\pi]$ between the signal and the local oscillator can be adjusted and is varied  with uniform probability. Two photo-detectors with efficiency $\eta$ measure the field and the resulting photo-currents are  subtracted electronically yielding the output signal 
\begin{eqnarray}\label{1112}
    I_{H}(\eta)=\langle \boldsymbol{Y}_{\theta}\rangle+\sqrt{\frac{1-\eta}{2\eta}}\langle u_{\theta}+v_{\theta}\rangle,
\end{eqnarray}
where the modes $u$ and $v$ are in the vacuum state
and 
\begin{eqnarray}
    \boldsymbol{Y}_{\theta}=\cos{\theta}\boldsymbol{q}+\sin{\theta}\boldsymbol{p},\end{eqnarray}
    are the field's quadrature amplitudes. The goal is to reconstruct the quantum state of the light field based on measurements $ I_{H}(\eta)$    which are obtained by varying the relative phase $\theta$ between the input signal and the local oscillator in the homodyne detector.
The phase is a parameter controlled by the experimenter.
\par By varying the local oscillator phase 
$\theta$ continuously over a range that is typically $[0,\pi]$, one can  obtain a series of measurements for different quadratures $Y_{\theta}$. Each fixed value of 
$\theta$ gives a marginal PDF $\mathcal{W}(y|  \cos{\theta},  \sin{\theta})$
 for the quadrature $Y_{\theta}$. 
 %Typically, one does the inverse Radon transform to first reconstruct the Wigner function and then one more inversion transform to reconstruct the density matrix kernel in a specific basis. 
\par One can see that the optical tomogram and the symplectic tomogram, are connected by the transformation \eqref{1359}, that gives
\begin{eqnarray}
   \frac{1}{|r|} \mathcal{W}(y|  \cos{\theta},  \sin{\theta})\equiv \mathcal{W}(x|  \mu,  \nu)
\end{eqnarray}That is why all the results deduced for the symplectic tomogram are easily translated to the language of optical tomograms  and vice versa. The CF of the optical tomogram $\mathcal{W}(y|  \cos{\theta},  \sin{\theta})$ is introduced as
\begin{eqnarray}
     \phi(r;  \cos{\theta},  \sin{\theta})=\int\limits_{-\infty}^{\infty} \mathcal{W}(y|  \cos{\theta},  \sin{\theta})e^{iry}dy.
\end{eqnarray}
%The latter characteristic function is connected with \eqref{1113}  as follows
%\begin{eqnarray}\label{1400}
%\phi_{\mathcal{W}}(1;\mu,\nu)=\phi_{\mathcal{W}}(r;\cos{\theta},\sin{\theta}).  
%\end{eqnarray}
Then, we can immediately write the density operator, using the optical tomogram, as
\begin{eqnarray}\label{1312}
\boldsymbol{\rho}
%&=&\frac{1}{2\pi}\int\limits_{0}^{\infty} \int\limits_{0}^{2\pi} r\phi(r;\cos{\theta},\sin{\theta}) \exp{(-ir\boldsymbol{Y}_\theta)} d r d\theta\\\nonumber
&=&\frac{1}{2\pi}\int\limits_{-\infty}^{\infty} \int\limits_{0}^{\pi} |r|\phi(r;\cos{\theta},\sin{\theta}) \exp{(-ir\boldsymbol{Y}_\theta)} d r d\theta,
\end{eqnarray}
 where we  use the fact that $\cos(\theta + \pi) = -\cos\theta$ and $\sin(\theta + \pi) = -\sin\theta$, so $\boldsymbol{Y}_{\theta + \pi} = -\boldsymbol{Y}_\theta$. Moreover, from \eqref{1118} we can see that $\phi(r;\cos{(\theta+\pi)},\sin{(\theta+\pi)})=\phi(-r;\cos{(\theta)},\sin{(\theta)})$ holds. One can conclude that it is enough to vary the angle $\theta$ only in $(0,\pi)$ to reconstruct the state, that is exactly done in the experiment mentioned before.
\subsection*{1.6. From Characteristic Function to Wigner Function}
\par In quantum optics, a widely used connection between the Wigner function and the quantum tomogram  by Radon transformation exists:
\begin{eqnarray}\label{1431}
&&W(q, p) = \frac{1}{2\pi} \iiint_{-\infty}^{\infty}   \mathcal{W}(x| \mu, \nu) \exp\left(i(x - \mu q - \nu p)\right) \, dx \, d\mu \, d\nu,\\\nonumber
   &=& \frac{1}{4\pi^2} \int_{-\infty}^{\infty} \int_{0}^{\infty} \int_{0}^{2\pi} \mathcal{W}(y| \cos{\theta},\sin{\theta}) \exp\left[ir\left(y - q \cos{\theta} - p \sin{\theta}\right)\right]r dy  dr  d\theta.
\end{eqnarray}
One can rewrite it in terms of the CF as follows:
\begin{eqnarray}\label{1432}
&&W(q, p) = \frac{1}{2\pi} \iint_{-\infty}^{\infty} \phi_{\mathcal{W}}(1;\mu,\nu) \exp\left(-i( \mu q + \nu p)\right)  \, d\mu \, d\nu,\\\nonumber
   &=& \frac{1}{4\pi^2} \int_{0}^{\infty} \int_{0}^{2\pi} \phi_{\mathcal{W}}(r|\cos{\theta},\sin{\theta}) \exp\left[-ir\left(q \cos{\theta} + p \sin{\theta}\right)\right]r  dr  d\theta.
\end{eqnarray}
\par We can conclude, that the CF of a symplectic tomogram is connected to the Wigner function of a quantum state as follows:
\begin{eqnarray}
    W(q,p)=\frac{1}{2\pi}\int \phi(1;\mu,\nu)e^{-i(\mu q+\nu p)}d\mu d\nu.
\end{eqnarray}
Using the Wigner function, similar to \eqref{eqn: integral from matrix elements}, the kernel element can be written as:
\begin{eqnarray}\label{1624}
    \rho(y,y')&=&
    \frac{1}{2\pi}\int\limits_{-\infty}^{\infty} W\left(\frac{y+y'}{2};p\right)\exp{\left(-ip(y-y')\right)}  dp\equiv\frac{1}{\sqrt{2\pi}}\overline{W}\left(\frac{y+y'}{2};y-y'\right),
\end{eqnarray}
where  $\overline{W}(q,\omega)$  denotes the Fourier transform of the Wigner function. One can conclude that the following connection of CF and the Wigner function holds
\begin{eqnarray}
    \overline{\phi}\left(1;\frac{y+y'}{2},y-y'\right)=\overline{W}\left(\frac{y+y'}{2};y-y'\right),\quad \forall y,y'.
\end{eqnarray}
\par In many traditional methods, the reconstruction of the quantum state typically proceeds via an intermediate step, first obtaining the Wigner function. In contrast, our approach bypasses this intermediary by directly estimating the tomogram or its CF from the measured data. This direct reconstruction not only simplifies the overall procedure but also reduces potential sources of error that may accumulate during intermediate transformations. 
%Furthermore, by working directly with these representations, we are able to more efficiently harness the intrinsic connections between the tomogram, characteristic function, and the key trace characteristics of the quantum state, as formalized in Theorem~\ref{thm1}. 

\section*{S2. Nonparametric Kernel Tomogram Estimation }\label{app1}
\par Let us have $n$ i.i.d. data points $\{X_{j,\mu,\nu}\}_{j=1}^{n}$ from an unknown PDF $\mathcal{W}(x|\mu,\nu)$. 
The kernel density estimate (KDE) of the tomogram  $\mathcal{W}(x|\mu,\nu)$ is~\cite{silverman2018density,tsybakov2009introduction}:
\begin{eqnarray}\label{1500_1}
    \hat{\mathcal{W}}_{nh}(x|\mu,\nu) = \frac{1}{nh} \sum_{j=1}^{n} K\left(\frac{x-X_{j,\mu,\nu} }{h}\right).
\end{eqnarray}  
In this section, we denote the estimates with hats and use the notation $ \hat{\mathcal{W}}_{nh}(x)= \hat{\mathcal{W}}_{nh}(x|\mu,\nu)$.
Here $K(x)$ is the kernel function, which is generally  smooth, and $h > 0$ is the bandwidth parameter, which controls the width of the kernel.
A kernel function is generally symmetric and normalised $\int K(x)dx = 1$. Moreover
\(\lim\limits_{{x \to -\infty}} K(x) = \lim\limits_{{x \to +\infty}} K(x) = 0\), holds. 
These requirements are needed to guarantee that the KDE $\hat{\mathcal{W}}_{nh}(x)$ is a PDF. 
\par The kernel $K(x)$ is typically chosen from the set of  PDFs. The three most common kernel functions used are the Gaussian, Epanechnikov, and rectangular (uniform) kernels. These kernels have different shapes and characteristics, and the choice of kernel depends on the specific characteristics of the data and the goals of the analysis. 
The Epanechnikov is a special kernel that has the lowest (asymptotic) mean square error. 
%Despite most kernel functions are positive, some kernel functions could be negative. Special types of kernel functions, known as higher-order kernel functions, will take negative values in some regions. These higher-order kernel functions, though very counterintuitive, might have a smaller bias than the usual kernel functions.
\par The mean squared error (MSE) of the KDE is~\cite{silverman2018density,tsybakov2009introduction}:
 \begin{eqnarray}
\text{MSE}(\hat{\mathcal{W}}_{nh}(x_0)) &=& \mathbb{E}[(\hat{\mathcal{W}}_{nh}(x_0)-\mathcal{W}(x_0))^2]=\text{Bias}^2(\hat{\mathcal{W}}_{nh}(x_0)) + \text{Var}(\hat{\mathcal{W}}_{nh}(x_0)) \nonumber \\
&=& \frac{h^4}{4} ({\mathcal{W}}''(x_0))^2 \mu^2_K + \frac{1}{nh} {\mathcal{W}}(x_0) \sigma^2_K + o(h^4) + o\left(\frac{1}{nh}\right)\\\nonumber
&=& O(h^4) + O\left(\frac{1}{nh}\right),
\end{eqnarray}
where \(\mu_K = \int y^2 K(y) \, dy\neq 0\) and  \(\sigma^2_K \equiv \int K^2(y) \, dy\).
%The first two terms, \( \frac{1}{4h^4} ({f}''(x_0))^2\mu^2_K + \frac{1}{nh}{f}(x_0)\sigma^2_K \), is called the asymptotic mean square error (AMSE). 
The bandwidth $h$ plays a crucial role in the balance between bias and variance in the estimation. The bandwidth minimizing the MSE is given by 
\begin{eqnarray}
    h_{\mathrm{opt}}(x_0) = \left(\frac{1}{n} \cdot \frac{{\mathcal{W}}(x_0)\sigma^2_K}{({\mathcal{W}}''(x_0))^2 \mu^2_K}\right)^{1/5} = O( n^{-1/5}).\end{eqnarray}
   For this choice of smoothing bandwidth the MSE has the rate 
   \begin{eqnarray}\label{1640}
     \text{MSE}(\hat{f}_{nh_{\mathrm{opt}}}(x_0)) = O(n^{-4/5}) . 
   \end{eqnarray} 
\par In the above analysis, we only consider a single point $ x_0 $. A straightforward generalization is the mean integrated square error (MISE):
   \begin{eqnarray}\label{1641}
\text{MISE}(\hat{\mathcal{W}}_{nh}) 
= \int \text{MSE}(\hat{\mathcal{W}}_{nh}(x)) \, dx.
 \end{eqnarray} 
One can show that
  \begin{eqnarray}
\text{MISE}(\hat{\mathcal{W}}_{nh}) 
= \frac{\mu^2_K}{4}  h^4  \int |\mathcal{W}''(x)|^2 \, dx 
+ \frac{\sigma^2_{K}}{nh} + o(h^4) + o\left( \frac{1}{nh} \right)= O(h^4) + O\left( \frac{1}{nh} \right).
\end{eqnarray}
The optimal smoothing bandwidth is often chosen by minimizing the first two terms of the latter expression. Namely, we have
\begin{eqnarray}
h_{\mathrm{opt}} = \left( \frac{1}{n} \cdot \frac{4  \sigma^2_{K}}{\mu^2_{K} \int |\mathcal{W}''(x)|^2 \, dx} \right)^{1/5}= O( n^{-1/5}),
\end{eqnarray}
providing the optimal rate 
\begin{eqnarray}\label{1640_10}
     \text{MISE}(\hat{\mathcal{W}}_{nh_{\mathrm{opt}}}(x|\mu,\nu))\equiv\mathbb{E}\left[\int \left(\hat{\mathcal{W}}_{nh_{\mathrm{opt}}}(x|\mu,\nu)-{\mathcal{W}}(x|\mu,\nu)\right)^2dx\right] = O(n^{-4/5}).
   \end{eqnarray} 
\par Note that the bandwidth depends on the unknown derivatives of the density. To overcome that difficulty, a variety of automatic, data-based methods have been developed for the bandwidth selection. Several review studies have been undertaken to compare their efficiencies, with the general consensus that the plug-in selectors and cross validation~\cite{hall1992smoothed} selectors are the most useful over a wide range of data sets.
\par The classical result \eqref{1640_10} is obtained minimizing the $L_2$- norm distance and gives little information about how close the estimate is to the true PDF in $L_1$- norm, that is required for total variation \eqref{totalvar}. However, in Ref.~\cite{hall1988minimizing} under condition on the PDF to be bounded and having at least $p$ bounded, continuous and integrable derivatives, 
the kernel estimate achieves the following rate with respect to $L_1$- norm:
 \begin{eqnarray}
    T(\hat{\mathcal{W}}_{nh_u}(x|\mu,\nu), {\mathcal{W}}(x|\mu,\nu))\equiv  \frac{1}{2}\!\ \int\limits_{-\infty}^{\infty}\!\!\mathbb{E} \left[|\hat{\mathcal{W}}_{nh_u}(x|\mu,\nu)-{\mathcal{W}}(x|\mu,\nu)|\right]dx=O\left(n^{-\frac{p}{2p+1}}\right)\!,
 \end{eqnarray}
 with the optimal bandwidth $h_u=O\left(n^{-1/(p+1)}\right)$.
 The kernel function used in \eqref{1500_1} is selected to be of order $p$ - that is $\int |z^pK(z)|dz<\infty$, $K(z)$ is bounded and its $p-1$ moments are zero, while its $p$-th moment is finite. For example, second-order kernels, e.g., Gaussian or Epanechnikov, has the property that it integrates to one and its first moment (mean) is zero, while its second moment (variance) is finite. For such  kernels the rate of convergence in $L_1$- norm   is $O(n^{-2/5})$. Higher-order kernels (e.g., biweight or triweight) can have more than two zero moments,  approximating smoother density functions more effectively. However, they are less frequently applied due to their complexity and potential for negative values in the density estimate.
 \par KDE can perform poorly at the boundaries of the domain, as kernels may be only partially within the data range. Moreover, for large datasets, KDE can be computationally intensive, as it requires summing over all data points for each evaluation of the
 PDF. Techniques like Fast Fourier Transform (FFT) and kernel density trees help to mitigate this issue.
 
\subsection*{2.1. Nonparametric Characteristic Function Estimation}
\par For our approach, the estimation of the CF is needed to characterize the quantum state \eqref{1050}, the kernel function \eqref{eqn: integral from matrix elements} and the trace characteristics \eqref{1051} mentioned in the previous section. 
We define the kernel CF estimate (KCFE) as 
\begin{eqnarray}\label{1040}
\hat{\phi}_{nh}(t;\mu,\nu)&=& \int\limits_{-\infty}^{\infty} \hat{\mathcal{W}}_{nh}(x|\mu,\nu)e^{it x}  dx = \frac{1}{nh} \int\limits_{-\infty}^{\infty}\sum_{j=1}^{n} K_{\mu,\nu}\left(\frac{x-X_{j,\mu,\nu}}{h}\right) e^{{i}t x} dx\\\nonumber
&=& \frac{1}{n}\sum_{j=1}^{n} e^{itX_{j,\mu,\nu} }\int\limits_{-\infty}^{\infty}K_{\mu,\nu}(z)e^{ith z}dz.
\end{eqnarray}
%\liubov{write more accurate} 
\par  The probability density estimate can be introduced as the sum of the Dirac delta functions built on a variational series of vectors of samples $X$~\cite{chentsov1981correctness}:
\begin{eqnarray}
\hat{\mathcal{W}}(x|\mu,\nu)  =\frac{1}{n} \sum_{j=1}^{n} \delta(x - X_{j,\mu,\nu}).
\end{eqnarray}
Its CF can be defined as follows:
\begin{eqnarray}\label{1229_1}
\hat{\phi}_{X,n}(t;\mu,\nu) = \int\limits_{-\infty}^{\infty}  \hat{\mathcal{W}}(x|\mu,\nu)e^{it x} \,dx = \int\limits_{-\infty}^{\infty}\frac{1}{n} \sum_{j=1}^{n} \delta(x-X_{j,\mu,\nu})e^{it x} dx=\frac{1}{n} \sum_{j=1}^{n}e^{it X_{j,\mu,\nu}}.
\end{eqnarray}
Assuming that the kernel function $K_{\mu,\nu}(z)$ is also selected to be a PDF, we can introduce its CF
\begin{eqnarray}
\phi_K(th;\mu,\nu)\equiv\int\limits_{-\infty}^{\infty}K_{\mu,\nu}(z)e^{ithz}dz.
\end{eqnarray}
Using this, we rewrite \eqref{1040} as follows:
\begin{eqnarray}\label{1041_1}
\hat{\phi}_{nh}(t;\mu,\nu)=\hat{\phi}_{X,n}(t;\mu,\nu) \phi_K(th;\mu,\nu).
\end{eqnarray}
 We see that the latter estimator is a product of two CFs, one dependent on the measurement data and one on the bandwidth parameter of the kernel function. Since we define the kernel function ourselves, the form of $\phi_K(th;\mu,\nu)$ is analytically known. In signal processing \eqref{1041_1}, this is called smoothed empirical CF~\cite{stefanski1990deconvolving}.
 If $\hat{\phi}_{nh}( t ;\mu,\nu)$ is absolutely integrable, one can estimate the PDF $\mathcal{W}_{X}(x|\mu,\nu)$ of $X$ by the inverse Fourier transform of the CF:
\begin{eqnarray}\label{tom_fil}
 \hat{\mathcal{W}}_{nh}(x|\mu,\nu)&=&
 \frac{1}{2\pi } \int\limits_{-\infty}^{\infty}  \hat{\phi}_{nh}( t ;\mu,\nu)e^{-itx} dt.
\end{eqnarray} 
Since the  empirical CF is always bounded by one, e.g., $|\hat{\phi}_{X,n}(t;\mu,\nu)|\leq 1$,
it suffices that the kernel's Fourier transform $\phi_K(th;\mu,\nu)$ is integrable. For example, if we choose a Gaussian kernel in the form
\begin{eqnarray}\label{1004_2}
    K_{\mu,\nu}(z) =
     \frac{1}{\sqrt{\pi(\mu^2+\nu^2)}}\exp{\left[-\frac{z^2}{\mu^2+\nu^2}\right]},
\end{eqnarray}
the CF of this kernel is
\begin{eqnarray}\label{1614}
    \phi_K(th;\mu,\nu)=\exp{\left(-\frac{t^2h^2(\mu^2+\nu^2)}{4}\right)},
\end{eqnarray}
that is a Gaussian function  in $t$  and it is known to be integrable. Therefore, under this choice the inversion integral exists.
\par The KCFE with the Gaussian kernel  can be written as
\begin{eqnarray}\label{1231}
  \hat{\phi}^G_{nh}(t;\mu,\nu)=\hat{\phi}_{X,n}(t;\mu,\nu) \exp{\left[-\frac{t^2h^2(\mu^2+\nu^2)}{4}\right]}.
\end{eqnarray}
Now we need to find the optimal bandwidth parameter minimizing the distance between the CF and its kernel estimate in some specific norms. 

\subsection*{2.2. Asymptotic Properties of the Kernel Characteristic Function Estimator}
\par Now we analyze the estimation error of the KCFE. 
Only the sample CF $\hat{\phi}_{X,n}(t)$ depends on the observable.
%We assume that for every pair of $\mu_k,\nu_k$ the errors are the same, that is why we omit the index $k$ in this subsection. 
By linearity of expectation we have~\cite{tsybakov2009introduction}:
\begin{eqnarray}
    \mathbb{E}[\hat{\phi}_{X,n}(t)]&=&\frac{1}{n} \mathbb{E}\Big[\sum_{j=1}^{n}e^{itX_{j}}\Big]=\mathbb{E}\Big[e^{itX}\Big]=
\int e^{it x} f(x)dx=\phi(t),\\\nonumber
%E[\hat{\phi}^2_{n}(t)]&=&\frac{1}{n^2} E\Big[\sum_{j,l=1}^{n}e^{it(X_{j}+X_{l})}\Big]=\frac{1}{n^2} E\Big[\sum_{j,l=1;j\neq l}^{n}e^{it(X_{j}+X_{l})}+\sum_{j=1}^{n}e^{2itX_{j}}\Big]\\\nonumber
%&=&\frac{1}{n^2}\left(n(n-1)\phi^2(t)+n\phi(2t)\right)=\frac{1}{n}\left((n-1)\phi^2(t)+\phi(2t)\right),\\\nonumber
\mathbb{E}[|\hat{\phi}_{X,n}(t)|^2]&=&\frac{1}{n^2} \mathbb{E}\Big[\sum_{j,l=1}^{n}e^{it(X_{j}-X_{l})}\Big]=\frac{1}{n^2} \mathbb{E}\Big[\sum_{j,l=1;j\neq l}^{n}e^{it(X_{j}-X_{l})}+\sum_{j=1}^{n}e^{1}\Big]\\\nonumber
&=&\frac{1}{n^2}\left(n+n(n-1)\phi(t)\phi(t)^{\star}\right)=\frac{1}{n}+\left(1-\frac{1}{n}\right)|\phi(t)|^2,
\end{eqnarray}
where we used that for $X$ and $Y$ which are i.i.d. random variables, the CF of their sum is given by:
$\phi_{X+Y}(t) = \phi_X(t) \cdot \phi_Y(t)$. Then we can deduce the expectation of the KCFE \eqref{1041_1}:
\begin{eqnarray}
    \mathbb{E}[\hat{\phi}_{nh}(t)]&=& \phi_K(th)\mathbb{E}[\hat{\phi}_{X,n}(t)]=\phi_K(th)\phi(t),\\\nonumber
    \mathbb{E}[|\hat{\phi}_{nh}(t)|^2]&=&|\phi_K(th)|^2\mathbb{E}[|\hat{\phi}_{X,n}(t)|^2]=|\phi_K(th)|^2 \left(\frac{1}{n}+\left(1-\frac{1}{n}\right)|\phi(t)|^2\right).
\end{eqnarray}
Using these results, we obtain the bias and  variance of the estimate:
\begin{eqnarray}
   \mathrm{Bias}( \hat{\phi}_{nh}(t))=(\phi_K(th)-1)\phi(t),
\end{eqnarray}
\begin{eqnarray}\label{1038}
    \text{Var}(\hat{\phi}_{nh}(t))&=&\mathbb{E}[|\hat{\phi}_{nh}(t)|^2]-|\mathbb{E}[\hat{\phi}_{nh}(t)]|^2=\frac{|\phi_K(th)|^2 }{n}\left(1-|\phi(t)|^2\right).
\end{eqnarray}
One can see that if $h\rightarrow 0$, the bias is tending to zero. If $n\rightarrow \infty$, the variance is tending to zero. Since the CF can take complex values, we define the MSE for the absolute value as
\begin{eqnarray}\label{1640_10_2}
    \operatorname{MSE}(\hat{\phi}_{nh}( t ;\mu,\nu))&\equiv&\mathbb{E}\left(|\hat{\phi}_{nh}( t ;\mu,\nu)-{\phi}( t ;\mu,\nu)|^2\right)= |\operatorname{Bias}( \hat{\phi}_{nh}(t))|^2+ \operatorname{Var}(\hat{\phi}_{nh}(t))\\\nonumber
    &=&|\phi_K(th)-1|^2|\phi(t)|^2+\frac{|\phi_K(th)|^2 }{n}\left(1-|\phi(t)|^2\right).%\\\nonumber
 %   &=& |(\phi_K(th)-1)\phi(t)|^2+\frac{\phi^2_K(th)}{n}\left(\phi(2t)-\phi^2(t)\right).
   \end{eqnarray} 
To find the rate of convergence in $h$ one has to specify the kernel function.
   \par For the ground state of the HO,
the corresponding tomogram is given by \eqref{1530}, that is, a Gaussian PDF. Hence, the obvious choice of the kernel function is \eqref{1004_2}.
The corresponding CF \eqref{1614} is also Gaussian in $\mu,\nu$.
Then the KCFE  with the Gaussian kernel  is \eqref{1231}.
The MSE of this KCFE is
\begin{eqnarray}
 \!\!\!   \mathrm{MSE}(\hat{\phi}^G_{nh}(t;\mu,\nu)) 
  =\left(\exp{\left[-\frac{t^2h^2(\mu^2+\nu^2)}{4}\right]}-1\right)^2|\phi(t;\mu,\nu)|^2+\frac{\exp{\left[-\frac{t^2h^2(\mu^2+\nu^2)}{4}\right]}}{n}(1-|\phi(t;\mu,\nu)|^2).
\end{eqnarray}
For $|\phi(t;\mu,\nu)|<1$ the optimal  bandwidth parameter is
\begin{eqnarray}\label{1924}
h_{\phi} = \frac{2}{t(\mu^2+\nu^2)}\sqrt{-\ln\!\left(1-\frac{1-|\phi(t;\mu,\nu)|^2}{2n\,|\phi(t;\mu,\nu)|^2}\right)}
\approx \frac{2}{t(\mu^2+\nu^2)}\sqrt{\frac{1-|\phi(t;\mu,\nu)|^2}{2n\,|\phi(t;\mu,\nu)|^2}}=O(n^{-1/2}).
\end{eqnarray}
Since for any number $y$ with $0<y<1$, we have $\ln{(1-y)}<0$,  the expression under the square root is nonnegative if
\begin{eqnarray}
    0<1-\frac{1-|\phi(t;\mu,\nu)|^2}{2n\,|\phi(t;\mu,\nu)|^2}<1,
\end{eqnarray}
hold.  In practice, for big enough 
$n$ it is always the case. 
Then the optimal  MSE of the Gaussian KCFE is
\begin{eqnarray}\label{1224}
\mathrm{MSE}(\hat{\phi}_{nh_{\phi}}^G(t;\mu,\nu)) 
    &=&\frac{1-|\phi(t;\mu,\nu)|^2}{n}
-\frac{(1-|\phi(t;\mu,\nu)|^2)^2}{4n^2\,|\phi(t;\mu,\nu)|^2}
= O(1/n).
\end{eqnarray}
 The KCFE is asymptotically as efficient as the parametric maximum likelihood estimator (MLE).

\subsection*{2.3. Filtering Measurement Noise via Characteristic Functions}
The homodyne tomography does not take into account various losses (mode mismatching, failure of detectors)
in the detection process, which modify the distribution of results in a real
measurement compared with the idealized case. In~\cite{leonhardt1995measuring}, the following error model is introduced:
\begin{eqnarray}\label{1033}
   Z\equiv^d \kappa X+(1-\kappa)Y,\quad 0< \kappa< 1,
\end{eqnarray}
where $\equiv^d $ is an equality in distribution. 
Here, $\kappa$ comprises the overall detection efficiency, including all kinds of losses, and $Y$ is a sequence of i.i.d.  noise independent from the target variable $X$. 
\par If two random variables $X$ and $Y$ are independent, then the probability density of their sum is equal to the convolution of the probability densities of $X$ and $Y$. The independence of the two random variables implies that their joint PDF is equal to the product of their marginals, namely
$p_{X,Y}(x, y) = p_X(x) p_Y(y)$. Then we have
\begin{eqnarray}\label{1441}
 \mathcal{W}_{Z}(z) &=& \int_{-\infty}^{\infty} \int_{-\infty}^{\infty} p_{X,Y}(x, y) \delta\left(\kappa x +(1-\kappa) y - z\right)  dy  dx \\\nonumber
 &=& \frac{1}{1-\kappa}\int_{-\infty}^{\infty} p_{X,Y}\left(x, \frac{z - \kappa x}{1-\kappa}\right)  dx =  \frac{1}{1-\kappa}\int_{-\infty}^{\infty}  \mathcal{W}_X(x) p_Y\left(\frac{z - \kappa x}{1-\kappa}\right)  dx,%,\quad \kappa\neq 1,0,
\end{eqnarray}
where $\mathcal{W}_X(x)$ is a PDF (tomogram) of $X$, $p_Y(z)$ is a PDF of the noise $Y$.
\par For our approach we need the CF of $X$. It is known that the CF of the sum of two independent random variables is the
product of individual CFs. 
 The CF of $Z $, denoted by $\phi_{Z}(t)$, is:
\begin{eqnarray}
   \phi_{Z}(t) = \mathbb{E}[e^{itZ}] = \mathbb{E}[e^{it(\kappa X + (1 - \kappa) Y)}]=\mathbb{E}[e^{it\kappa X}] \cdot \mathbb{E}[e^{it(1 - \kappa) Y}]=\phi_X(\kappa t) \phi_Y((1-\kappa) t).
\end{eqnarray}
Here $\phi_X(\kappa t)$ is the CF of a target PDF, while $\phi_Y((1-\kappa) t)$ is the CF of $p_Y(z)$.
One can conclude, that if $\phi_Y((1-\kappa) t)$ is invertible, then
\begin{eqnarray}\label{1509}
    {\phi}_X( t)={\phi}_{Z}\left(\frac{t}{\kappa}
    \right)\phi^{-1}_Y\left(\frac{(1-\kappa)t}{\kappa}
    \right),\quad \kappa\in(0,1).
\end{eqnarray}
Similar idea appears in the context of \textit{deconvolution} of signals PDF from the noisy data~\cite{stefanski1990deconvolving,fan1991optimal,delaigle2004practical}. One imposes the nonvanishing condition on the CF of the noise, i.e., 
$    |\phi_Y(t)|>0$, $ \forall t$.
 When performing deconvolution \eqref{1509}, one needs to divide by $\phi_Y(t)$. If $\phi_Y(t)$ were to vanish at some points, the inversion process would become unstable or even impossible due to division by zero. For example, the CF of a normal distribution \eqref{1614} is nonvanishing since it is never zero. In contrast, some distributions, such as the uniform distribution on an interval, have CFs that vanish at certain values of $t$.
\par 
Thus, the model \eqref{1033} and the assumption about the noise distribution define the correction of the CF estimator.
Using \eqref{1041_1}, we can define the KCFE of the signal $X$ filtered from noise $Y$:
\begin{eqnarray}\label{CF_fil}
&&\hat{\phi}_{X,nh}( t ;\mu,\nu)= 
\hat{\phi}_{Z,nh}( t ;\mu,\nu)
\phi^{-1}_Y\left(\frac{(1-\kappa)t}{\kappa}
    \right),\\\nonumber
    &&\hat{\phi}_{Z,nh}( t ;\mu,\nu)\equiv \hat{\phi}_{Z,n}\left(\frac{t}{\kappa};\mu,\nu\right)
\phi_K\left(\frac{t}{\kappa} h;\mu,\nu\right),
\end{eqnarray}
where $\phi^{-1}_Y\left(\frac{(1-\kappa)t}{\kappa}
    \right)$ is assumed to be known from the noise model and $\phi_K\left(\frac{t}{\kappa} h;\mu,\nu\right)$ is dependent from the kernel function choice.  
If $\hat{\phi}_{X,nh}( t ;\mu,\nu)$ is integrable, one can estimate the PDF $\mathcal{W}_{X}(x|\mu,\nu)$ of $X$ by the inverse Fourier transform of the CF:
\begin{eqnarray}\label{tom_fil_2}
 \hat{\mathcal{W}}_{nh}(x|\mu,\nu)=
 %\frac{1}{2\pi}\int\limits_{-\infty}^{\infty} \hat{\phi}_X( t ;\mu,\nu)e^{-itx} dt=
 \frac{1}{2\pi } \int\limits_{-\infty}^{\infty}  \hat{\phi}_{Z,n}\left(\frac{t}{\kappa};\mu,\nu\right)
\phi_K\left(\frac{t}{\kappa} h;\mu,\nu\right)
\phi^{-1}_Y\left(\frac{(1-\kappa)t}{\kappa}
    \right)e^{-itx} dt.
\end{eqnarray} 
In order to guarantee that this
estimator is well defined we need to impose the conditions~\cite{delaigle2004practical}: \begin{eqnarray}
      &&|\phi_Y(t)|>0,\quad \forall t,\\\nonumber
      &&\sup_t |\phi_K\left(\frac{t}{\kappa} h;\mu,\nu\right)
\phi^{-1}_Y\left(\frac{(1-\kappa)t}{\kappa}
    \right)|<\infty,\\\nonumber
    &&\int |\phi_K\left(\frac{t}{\kappa} h;\mu,\nu\right)
\phi^{-1}_Y\left(\frac{(1-\kappa)t}{\kappa}
    \right)| dt<\infty.\end{eqnarray}
\par The trace product \eqref{1357} between two estimated states can be estimated as
 \begin{eqnarray}\label{1357_1}
   \widehat{\tr{( \boldsymbol{\rho}_{1,nh}\boldsymbol{\rho}_{2,nh})}}
  = \frac{1}{2\pi}\phi^{-2}_Y\left(\frac{1-\kappa}{\kappa}
    \right)\iint\limits_{-\infty}^{\infty} \hat{\phi}_{Z_1,nh}(1;\mu,\nu) \hat{\phi}_{Z_2,nh}(1;-\mu,-\nu) d\mu d\nu,
\end{eqnarray}
assuming for simplicity that the noise has a similar PDF for both measurement procedures. 
\par If $\boldsymbol{\rho}_1=\boldsymbol{\rho}$ is a true state and $\boldsymbol{\rho}_{2}=\boldsymbol{\rho}_{nh}$ is its estimate from the noisy data, then
 \begin{eqnarray}\label{1357_2}
   \widehat{\tr{( \boldsymbol{\rho}\boldsymbol{\rho}_{nh})}}
  = \frac{1}{2\pi}\phi^{-1}_Y\left(\frac{1-\kappa}{\kappa}
    \right)\iint\limits_{-\infty}^{\infty} {\phi}(1;\mu,\nu) \hat{\phi}_{Z_1,nh}(1;-\mu,-\nu) d\mu d\nu,
\end{eqnarray}
that shows the overlap between the sate and its estimate. 
The latter estimate must tend to purity for $n\rightarrow\infty$. The rate of convergence of this estimate is provided in the following sections. 
Other trace characteristics \eqref{1051} can be estimated in the same fashion.  The estimate of the kernel of the density matrix \eqref{eqn: integral from matrix elements} is
 \begin{eqnarray}
    \widehat{\rho_{nh}(y,y')}&=&\frac{1}{2\pi }\phi^{-1}_Y\left(\frac{1-\kappa}{\kappa}
    \right)\int\limits_{-\infty}^{\infty} \hat{\phi}_{Z,nh}(1;\mu,y-y')
\exp{\left(-i\frac{\mu(y+y')}{2}\right)}  d\mu
    \label{kel_fil}
\end{eqnarray}
Thus, with a collection of i.i.d. data sets $\{Z_{\mu_i, \nu_i}\}_{i=1}^{N} $ for a finite set of parameters 
$\{\mu_i, \nu_i\}_{i=1}^{N}$  on the interval, e.g.,  $\mu_i\in[-\mu_\mathrm{max},\mu_\mathrm{max}]$, $\nu_i\in[-\nu_\mathrm{max},\nu_\mathrm{max}]$, $\mu_\mathrm{max},\nu_\mathrm{max}>0$, $i=[1,N]$, we can estimate a finite number of CFs $\{\hat{\phi}(1;\mu_k,\nu_k)\}_{k=1}^N$  and tomograms $\{\hat{\mathcal{W}}(x|\mu_k,\nu_k)\}_{k=1}^N$. However, this is not sufficient for integration over the entire parameter space in \eqref{1357_1}- \eqref{kel_fil}. The solution to this problem is presented in the next section.

\section*{S3. Kernel quantum State Estimation}\label{sec_4}
\par We start with the  estimation of the kernel of the density operator \eqref{eqn: integral from matrix elements} in the coordinate basis. Let us rewrite  \eqref{eqn: integral from matrix elements}  as a sum of  integrals
\begin{eqnarray}\label{1227}
    &&{\rho}(y,y')={\rho}_{\mathrm{max}}(y,y')+{\rho}_{\mathrm{tail}+}(y,y')+{\rho}_{\mathrm{tail}-}(y,y'),
    \end{eqnarray}
    where we introduce the notations
    \begin{eqnarray}
    &&{\rho}_{\mathrm{max}}(y,y')\equiv \frac{1}{2\pi}\int\limits_{-\mu_{\mathrm{max}}}^{\mu_{\mathrm{max}}}\Big[{\phi}(1;\mu,y-y')\exp{\left(-i\frac{\mu(y+y')}{2}\right)}\Big]  d\mu,\label{1121_5}
    \end{eqnarray}
    \begin{eqnarray}
    &&{\rho}_{\mathrm{tail}\pm}(y,y')\equiv\frac{1}{2\pi}\int\limits_{\mu_{\mathrm{max}}}^{\infty}\Big[{\phi}(1;\pm\mu,y-y')\exp{\left(\pm i\frac{\mu(y+y')}{2}\right)}\Big]  d\mu.
    \label{1121_4}
\end{eqnarray}
\par  We have KCFE $\{\hat{\phi}(1;\mu_k,\nu=y-y')\}_{k=1}^{N_{\mu}}$ only to approximate the first integral. That means that no measurement data is available to directly estimate \eqref{1121_4}. 
%From this perspective, we can think of the continuous estimation function $\hat{\phi}(1;\mu,y-y')$, turning to zero for  $\mu>\mu_{\mathrm{max}}$, $\mu<-\mu_{\mathrm{max}}$. Using this assumption, we can approximate the continuous-time Fourier transform ${\rho}_{\mathrm{max}}(y,y')$  of a time-limited signal by the Discrete Fourier Transform (DFT).\
From \cite{markovich2024not} we know that the family of tomographic PDFs $\mathcal{W}(x|\mu,\nu)$  encountered in experiments typically exhibit a Hermite-Gaussian form. Such PDFs are bounded on infinity and  have finite moments.  Their multimodal structure makes them challenging to estimate using solely parametric methods. However, beyond a certain value of $x$, the Hermite modes decay, and the Gaussian tail becomes dominant.  Therefore, it is reasonable to adopt a hybrid estimation approach: to estimate ${\rho}_{\mathrm{max}}(y,y')$ nonparametrically from the measurement data while estimating the tail components ${\rho}_{\mathrm{tail} \pm}(y,y')$ parametrically using the prior knowledge of the theoretically prepared state. We use similar technique for $ {\tr{( \boldsymbol{\rho}_{1}\boldsymbol{\rho}_{2})}}$ estimation.

\subsection*{3.1. DFT-Based Approximation of Integrals Involving the Characteristic Function}
We want to approximate \eqref{1121_5} with a finite sum. Let us divide the interval of width $2\mu_{\mathrm{max}}$ into $N_{\mu}$ equally spaced sub-intervals with width \begin{eqnarray}\label{1008}
    \Delta\mu =2\mu_{\mathrm{max}}/N_{\mu}.
\end{eqnarray} 
We can approximate the Fourier integral \eqref{1121_5} by the sum
\begin{eqnarray}\label{1090}
  \tilde{\rho}_{\mathrm{max}}(y,y')\equiv  \frac{\Delta\mu}{2\pi}\sum\limits_{k=0}^{N_{\mu}-1}{\phi}(1;k \Delta\mu-\mu_{\mathrm{max}},y-y')\exp{\left(-i\frac{(k\Delta\mu-\mu_{\mathrm{max}})(y+y')}{2}\right)}.
\end{eqnarray}
By dividing the measurement range into small intervals and assigning a constant value of the function for each interval, we do a similar thing that is done in signal processing, when one performs time-sampling of the signal and the sampling frequency is given by $F=\Delta\mu^{-1}$. Further we use KCFE to write the estimate of \eqref{1090} as follows
\begin{eqnarray}\label{1260}
     \hat{\tilde{\rho}}_{\mathrm{max}}(y,y')&\equiv & \frac{1}{2\pi F}\exp{\left(i\frac{\mu_{\mathrm{max}}(y+y')}{2}\right)}\sum\limits_{k=0}^{N_{\mu}-1} \hat{\phi}_{nh}[k]\exp{\left(-i\frac{k(y+y')}{2F}\right)},
\end{eqnarray}
where we introduce the notation 
\begin{eqnarray}
    \hat{\phi}_{nh}[k]\equiv \hat{\phi}_{nh}\left(1;\frac{k}{F}-\mu_{\mathrm{max}},y-y'\right),
\end{eqnarray}
and the discrete Fourier transform (DFT) of this function is
\begin{eqnarray}\label{1259}
   \hat{\Phi}_{nh}[g]=\sum\limits_{k=0}^{N_{\mu}-1} \hat{\phi}_{nh}[k]\exp{\left(-i 2\pi \frac{ gk}{N_{\mu}}\right)}.
\end{eqnarray}
We use square brackets to stress that the function takes discrete input.
Comparing \eqref{1260} and \eqref{1259}, we conclude that the DFT calculates the Fourier transform only for  discrete frequencies which are given by $\frac{y+y'}{2}=2\pi\frac{g}{N_{\mu}}F$.
Then the final estimate of the  integral \eqref{1121_5} after truncation, discretization and CF estimation is 
\begin{eqnarray}\label{1012}
     \hat{\tilde{\rho}}_{\mathrm{max}}(y,y')&=& \frac{1}{2\pi F}\exp{\left(i\frac{\mu_{\mathrm{max}}(y+y')}{2}\right)}\hat{\Phi}_{nh}\Big[\frac{(y+y')N_{\mu}}{4\pi F}\Big].
\end{eqnarray}
\par We use $L_{\infty}$ norm as a measure of  mean error of approximating the kernel \eqref{eqn: integral from matrix elements} of the density matrix operator with \eqref{1012}:
\begin{eqnarray}\label{1033_33}
   {L}_{\infty}(\boldsymbol{\rho}_{nh})\equiv \sup\limits_{y,y'}  \mathbb{E}\left(|{\rho}(y,y')-\hat{\tilde{\rho}}_{\mathrm{max}}(y,y')|^2\right).
\end{eqnarray} 
It quantifies the maximum pointwise squared deviation between the true kernel and the estimated kernel over the domain of $y,y'\in\mathbb{R}$.
\par In our method we have three types of errors. First error is the one arising from the CF estimation. The second is the truncation of the integral over the infinite space by the integral over a final domain. The third is the discretization error arising from approximating the integral with the finite sum. 
\par Suppose we want to approximate
\begin{eqnarray}
    I(\theta)=\frac{1}{2\pi}\int\limits_{-\infty}^{\infty}\phi(\mu)e^{-i\mu \theta}d\mu
\end{eqnarray}
with the truncated integral on the finite interval $[-\mu_{\mathrm{max}},\mu_{\mathrm{max}}]$:
\begin{eqnarray}
     I_{\mathrm{trunc}}(\theta;\mu_{\mathrm{max}})\equiv \frac{1}{2\pi}\int\limits_{-\mu_{\mathrm{max}}}^{\mu_{\mathrm{max}}}\phi(\mu)e^{-i\mu \theta}d\mu
\end{eqnarray}
The truncation error is given by
\begin{eqnarray}
\epsilon_{\mathrm{trunc}}(\theta;\mu_{\mathrm{max}}) = \left| I(\theta)-I_{\mathrm{trunc}}(\theta;\mu_{\mathrm{max}})\right| 
= \left| \frac{1}{2\pi} \int_{|\mu| > \mu_{\max}} \phi(\mu) \, e^{-i\mu \theta} \, d\mu \right|.
\end{eqnarray}
Let $\phi(\mu)$  be analytic in the strip
\begin{eqnarray}
S_b=\{ \mu \in \mathbb{C} : |\operatorname{Im}(\mu)| < b \},\quad b > 0.
\end{eqnarray}
Assume that on the real line it satisfies 
\begin{eqnarray}\label{1621}
    |\phi(\mu)|\leq Ce^{-\tau \mu}, \quad \mu \in \mathbb{R},\quad \tau>0,
\end{eqnarray}
and $C>0$ is some constant. Then it is easy to see that the
truncation error is bounded by
\begin{eqnarray}\label{1621_1}
\epsilon_{\mathrm{trunc}}(\theta;\mu_{\mathrm{max}}) \leq \frac{C}{2\pi}\, e^{-\tau \mu_{\max}},\quad \forall \theta\in \mathbb{R}.
\end{eqnarray}
So the tail is exponentially small in $\mu_{\mathrm{max}}$.
This indicates that for analytic functions that decay exponentially on the real axis (or faster), the truncation error decreases exponentially with the truncation limit $\mu_{\max}$. The exponential, sine and cosine functions, as long as polynomials are analytic because their Taylor series expansions converge to the function values on their entire domain.
\par When approximating a Fourier integral with a discrete sum, a discretization error arises from replacing the continuous integral by a numerical quadrature. We approximate $I_{\mathrm{trunc}}(\theta,\mu_{\mathrm{max}})$ by 
\begin{eqnarray}\label{eq:discrete_sum}
I_{N_{\mu}} = \frac{\Delta \mu}{2\pi} \sum_{j=0}^{N_{\mu}-1} \phi(\mu_j) \, e^{-i\mu_j \theta},
\end{eqnarray}
with sample points
\begin{eqnarray}
\mu_j = -\mu_{\max} + j\Delta \mu,\quad \text{for } j = 0,1,\dots, N_{\mu}-1,
\end{eqnarray}
and grid spacing $\Delta \mu$.
If $f(\mu)$ is analytic in the strip
$S_b$ and is uniformly bounded $|f(\mu)|\leq M$, $\forall\mu\in S_b$  the discretization error
decays as~\cite{trefethen2014exponentially}:
\begin{eqnarray}\label{1644}
E_{\mathrm{dis}} =  \left| I_{\mathrm{trunc}} - I_{N_{\mu}}\right|\leq \frac{M}{\pi}\frac{e^{-2\pi \tau/\Delta\mu}}{1-e^{-2\pi \tau/\Delta\mu}}=O\!\left(e^{-\pi\tau N_{\mu}/\mu_{\mathrm{max}}}\right).
\end{eqnarray}
\par Let us use these results to our case. Considering the exponential family of PDFs we are typically dealing with, the   $\phi(1;\mu,y-y')$ is analytic and  we can use the result \eqref{1621_1} to upper bound the  truncation error introduced by limiting the infinite integration domain to a finite interval. Since $|\phi(\mu)|\leq 1$, we can find the constants $C$ and $\tau$ such that the condition \eqref{1621} will be satisfied. 
First we select:
\begin{eqnarray}
    |\phi(1;\mu,y-y')|\leq 1\leq Ce^{-\tau \mu_{\mathrm{max}}}, \quad \mu \in \mathbb{R},\quad \tau>0.
\end{eqnarray}
Since $\mu_{\mathrm{max}}>0$, the inequality holds for $\tau\leq \ln{(C)}/\mu_{\mathrm{max}}$, $C\geq 1$. The smallest truncation error is achieved by taking the largest admissible 
$\tau$. Then the truncation error is
\begin{eqnarray}\label{1743}
  \varepsilon_{\mathrm{trunc}} = |{\rho}(y,y')-{\rho}_{\mathrm{max}}(y,y')|\leq \frac{1}{2\pi}.\end{eqnarray}
  Choosing the maximal 
$\tau$  makes the product $Ce^{-\tau \mu_{\mathrm{max}}}$ exactly one, and drives the truncation error down to the universal ceiling 
\eqref{1743}, independent of $C$ and $\mu_{\mathrm{max}}$. Lets observe
\begin{eqnarray}
    |\phi(1;\mu,y-y')|\leq Ce^{-\tau \mu_{\mathrm{max}}}\leq 1, \quad \mu \in \mathbb{R},\quad \tau>0.
\end{eqnarray}
If $0<C\leq 1$ the inequality holds $\forall \tau\geq 0$. If $C>1$, then $\tau\geq \ln{(C)}/\mu_{\mathrm{max}}$ and the truncation error is
\begin{eqnarray}\label{17433}
  \varepsilon_{\mathrm{trunc}} \leq \frac{C}{2\pi}e^{-\tau\mu_{\mathrm{max}}}.\end{eqnarray}
The truncation error is tending to zero when $\mu_{\mathrm{max}}\rightarrow\infty$. However, we are always limited to some finite value of $\mu_{\max}$, so the truncation error always exists. Further we show how to mitigate this error with parametric correction.
\par Using the result \eqref{1644}, the discretization error arising from replacing the continuous integral with a sum over discrete points, can be estimated as 
\begin{eqnarray}
  \varepsilon_{\mathrm{dis}} =    |{\rho}_{\mathrm{max}}(y,y')-\tilde{\rho}_{\mathrm{max}}(y,y')|\leq  \frac{1}{\pi}\frac{e^{-\pi \tau N_{\mu}/ \mu_{\max}}}{1-e^{-\pi \tau N_{\mu}/ \mu_{\max}}},
\end{eqnarray}
where we can select $M=1$ since $|\phi(\mu)|\leq 1$ holds.
\par 
Finally, we take into account the estimation error, arising from the imperfections of the KCFE  method:
\begin{eqnarray}\label{234}
  \varepsilon_{K}^2 &=&   \mathbb{E}\left( |\tilde{\rho}_{\mathrm{max}}(y,y')-\hat{\tilde{\rho}}_{\mathrm{max}}(y,y')|^2\right)\\\nonumber
  &=& \mathbb{E}\left(\Bigg |\frac{1}{2\pi F}\exp{\left(i\frac{\mu_{\mathrm{max}}(y+y')}{2}\right)}\sum\limits_{k=0}^{N_{\mu}-1} ({\phi}[k]-\hat{\phi}_{nh}[k])  
  \exp{\left(-i\frac{k(y+y')}{2F}\right)}\Bigg|^2\right).
\end{eqnarray}
We use the Cauchy-Schwarz inequality
\begin{eqnarray}
    \Bigg|\sum_k a_kb_k\Bigg|^2\leq  \left(\sum_k |a_k|^2\right)\left(\sum_k |b_k|^2\right),
\end{eqnarray}
where we take $a_k={\phi}[k]-\hat{\phi}_{nh}[k]$ and 
$b_k=\frac{1}{2\pi F}\exp{\left(i\frac{\mu_{\mathrm{max}}(y+y')}{2}\right)} \exp{\left(-i\frac{k(y+y')}{2F}\right)}$.
Then the upper bound is %|x\cdot y|\leq \|x\| \|y\|
  \begin{eqnarray}\label{234_1}
  \varepsilon_{K}^2 
  &\leq&\frac{N_{\mu}}{(2\pi F)^2}\mathbb{E}\left(\sum\limits_{k=0}^{N_{\mu}-1} \Big|{\phi}[k]-\hat{\phi}_{nh}[k]\Big|^2\right).
\end{eqnarray}
%N_mu is from the second sum of |b_k|^2
Since the exponential term is of a unit modulus, it does not affect the magnitude of the expression. We use the point-wise MSE of the kernel CF estimation \eqref{1224} to write
  \begin{eqnarray}\label{234_2}
  \varepsilon_{K}^2 \leq \frac{N_{\mu}}{(2\pi F)^2n}\sum\limits_{k=0}^{N_{\mu}-1}(1-|{\phi}[k]|^2)
  &\leq & \frac{\mu^2_{\max}}{\pi^2 n}.
\end{eqnarray}
\par Collecting all the errors, the distance between the estimate and the exact integral  is bounded by
\begin{eqnarray}\label{1113_1}
 \mathbb{E}\left( |{\rho}(y,y')-\hat{\tilde{\rho}}_{\mathrm{max}}(y,y')|^2\right)\leq3(\varepsilon_{\mathrm{trunc}}^2+\varepsilon_{\mathrm{dis}}^2+  \varepsilon_{K}^2)= \frac{3}{\pi^2}\Big[
\frac{1}{4} e^{-2\tau \mu_{\max}}+\frac{e^{-2\pi \tau N_{\mu}/ \mu_{\max}}}{(1-e^{-\pi \tau N_{\mu}/ \mu_{\max}})^2}+\frac{\mu^2_{\max}}{ n}\Big].
%\\\nonumber&+&
% \frac{e^{-\pi \tau N_{\mu}/ \mu_{\max}-\tau \mu_{\max}}}{1-e^{-\pi \tau N_{\mu}/ \mu_{\max}}}
%+\frac{e^{-\pi \tau N_{\mu}/ \mu_{\max}}}{1-e^{-\pi \tau N_{\mu}/ \mu_{\max}}}\frac{2\mu_{\max}}{\sqrt{n}}
%+ e^{-\tau \mu_{\max}}\frac{\mu_{\max}}{ \sqrt{n}})\Big].
\end{eqnarray}
As we mentioned, we need to select the parameter  $\mu_{\mathrm{max}}$ such that the tail behavior of the CF would be dominating the oscillations. In theory we can select any $\mu_{\mathrm{max}}$ that would fit this demand, but in practice it is dictated by the measurement setup restrictions. To guarantee 
\begin{eqnarray}\label{1105}
\lim\limits_{n\rightarrow\infty}\frac{\mu_{\mathrm{max}}^2}{ n}=0, \quad \lim\limits_{n\rightarrow\infty}\frac{N_{\mu}}{\mu_{\mathrm{max}}}=\infty,
\end{eqnarray}
we must check that $ \mu_{\mathrm{max}}$ is such that
\begin{eqnarray}
    \mu_{\mathrm{max}}=o\left(  n^{\frac{1}{2}}\right)
\end{eqnarray}
and $N_{\mu}$ must grow strictly faster than $ \mu_{\mathrm{max}}$.
To this end, we can  select 
\begin{eqnarray}\label{1202}
    \mu_{\mathrm{max}}\approx\log{n},\quad N_{\mu}\approx\mu_{\mathrm{max}}\log{n}.
\end{eqnarray} 
Then \eqref{1105} are satisfied. 
The error \eqref{1113_1} is 
\begin{eqnarray}\label{1113_2}
 \mathbb{E}\left( |{\rho}(y,y')-\hat{\tilde{\rho}}_{\mathrm{max}}(y,y')|^2\right)=
\begin{cases}
O(n^{-2\tau}), & \quad 0<\tau<\dfrac12,\\[6pt]
O((\log n)^{2}\,n^{-1}), & \quad \tau\ge\dfrac12.
\end{cases}
\end{eqnarray}
\par The total amount of measurement that we must do for our reconstruction procedure per $(y,y')$ is $T_{\mu}=N_{\mu}\times n= n(\log{n})^2$. Since we deal with the Gaussian tailed PDFs, the $\tau$ can be always selected bigger than $1/2$. 
Using the Lambert 
$W$- function, we find $n=\exp{2W(\sqrt{T_{\mu}}/2)}$. For large $T_{\mu}$ the $W$- function is $W(z)\sim \log{z}-\log{\log{z}}$. We can derive that $n\approx T_{\mu}/4$. 
Then
\begin{eqnarray}\label{eq:epsilon_rho_upperbound}
 \mathbb{E}\left( |{\rho}(y,y')-\hat{\tilde{\rho}}_{\mathrm{max}}(y,y')|^2\right)\leq %\frac{3(\log{n})^4}{\pi^2 T}\approx 
  \frac{3(\log{T_{\mu}/4})^4}{\pi^2 T_{\mu}}=\tilde{O}\left(\frac{1}{T_{\mu}}\right). 
\end{eqnarray}
In KDE estimation, the difference between pointwise MSE and supremum MSE (i.e., in $L_{\infty}$ norm) does not change the convergence rate, provided the target function is smooth and the estimator is well-behaved~\cite{tsybakov2009introduction}. Since the reconstructed function is smooth and bounded, taking the supremum on the bounded domain provides
\eqref{1033_33}
\begin{eqnarray}\label{1033_33_1}
  {L}_{\infty}(\boldsymbol{\rho}_{nh})=\tilde{O}\left(\frac{1}{T_{\mu}}\right).
\end{eqnarray}

\par The upper bound \eqref{234_2} allows to see the rate of convergence in $n$ of the estimation error $\varepsilon_{K}^2$. However, the actual error, calculated by \eqref{234} is much smaller, still preserving the rate $O(1/n)$, that can be seen in Fig.~\ref{fig:epsilon_rho}. The reason of this difference is that the upper bound provided by the  Cauchy-Schwarz inequality is not tight.
\begin{figure}[h]
    \centering
    \begin{subfigure}[t]{0.48\linewidth}
        \centering
        \includegraphics[width=\linewidth]{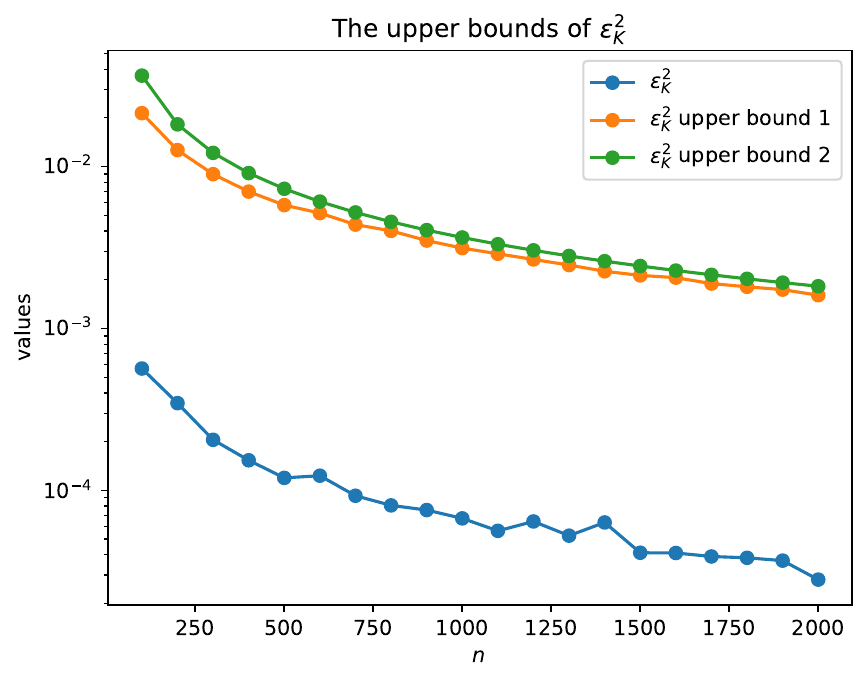}
        \caption{}
        \label{fig:epsilon_K}
    \end{subfigure}
    \hfill
    \begin{subfigure}[t]{0.48\linewidth}
        \centering
        \includegraphics[width=\linewidth]{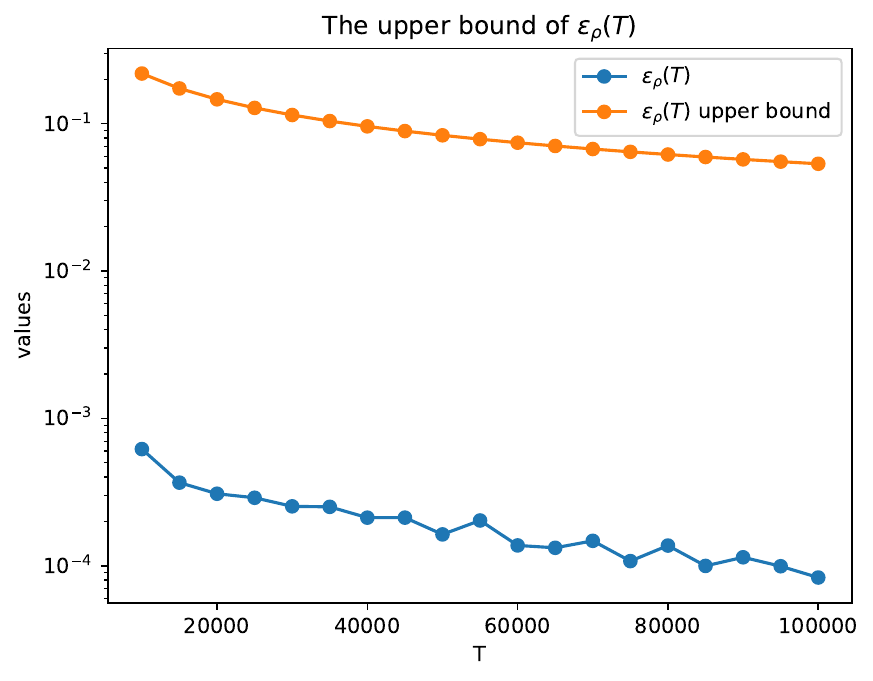}
          \caption{}
        \label{fig:epsilon_rho}
    \end{subfigure}
    \caption{Comparison of estimation errors with their theoretical upper bounds for the initial coherent cat state (a) Comparison between $\varepsilon_K^2$ (Eq.~\eqref{234}) and its upper bounds (Eqs.~\eqref{234_1} and~\eqref{234_2} for upper bound 1 and 2 in the labels, respectively) as a function of $n\in[100,2000]$ for $N_{\mu}=50$, $\mu_{\max}=6$, $(y,y')=(1.5,1.0)$. Each point is averaged over 100 trials. Observed ratios: $46.07$ and $55.37$. (b) Comparison between $\varepsilon_{\rho}$ (Eq.~\eqref{1033_33}) and its upper bound (Eq.~\eqref{eq:epsilon_rho_upperbound}) as a function of $T\in [10^4, 10^5]$.The observed ratio is $4.88 \times 10^2$.}
    \label{fig:epsilon_all}
\end{figure}

\par As we mentioned before, we can't increase the $\mu_{\mathrm{max}}$ with $n$ in realistic scenarios. That is why we  consider $\mu_{\mathrm{max}}$ as a constants, given by the experimental setup.
We must be assure that this constant is in accordance to \eqref{1105} $\mu_{\mathrm{max}}<< n$. In this case the third term in \eqref{1113_1} is decreasing even faster ($O(1/n)$), however the first term in in \eqref{1113_1} will give us a constant shift because of the truncation procedure. 
To reduce this truncation error we introduce a correcter of the tail component. In \cite{markovich2024not} we show that for the ground state of the HO the CF is \eqref{1530}.
Substituting it in \eqref{1121_4} we introduce the correction of the tail part as
 \begin{eqnarray}
    &&\boldsymbol{\rho}_{\mathrm{tail}\pm}(y,y')
\;=\;
\frac{1}{2\sqrt{\pi}}
\,\exp\!\Bigl(-\tfrac{y^2 + y'^2}{2}\Bigr)
\,
\mathrm{erfc}\!\Bigl(\tfrac{\mu_{\mathrm{max}} \mp i\,(y+y')}{2}\Bigr),
    %\equiv\frac{1}{2\pi}\exp{\left(\frac{-(y-y')^2)}{4}\right)}\int\limits_{\mu_{\mathrm{max}}}^{\infty}\Big[\exp{\left(\frac{-\mu^2}{4}\pm i\frac{\mu(y+y')}{2}\right)}\Big]  d\mu.
    \label{1121_4_0}
\end{eqnarray}
where the result involves the complementary error function:
\begin{eqnarray}
    \mathrm{erfc}(z)=\frac{2}{\sqrt{\pi}}\int\limits_{z}^{\infty}e^{-t^2}dt.
\end{eqnarray}
Then the KQSE with correction of the tail is
\begin{eqnarray}
  \boldsymbol{\rho}(y,y')=  \hat{\tilde{\rho}}_{\mathrm{max}}(y,y')+\boldsymbol{\rho}_{\mathrm{tail}\pm}(y,y').
\end{eqnarray}
 We cannot estimate how the parametric correction will reduce the error, since it is fully dependent on how well we guessed our parametric model. In the best-case scenario, the truncation error, already exponentially suppressed in $\mu_{\max}$, can be considered negligible. Theoretical knowledge of the CF can help estimate this error when $\mu_{\max}$ is insufficient to cover its oscillatory support. For instance, the CF of the HO decays as $e^{-t^2\beta/4}$, reaching $10^{-267}$ for the ground state and $10^{-258}$ for the $n=3$ Fock state when $\beta = 50$. Thus, given an experimental bound on $\mu_{\max}$, the expected truncation loss can be evaluated analytically, allowing one to decide whether a correction procedure is necessary. 
%\par Then we can write the total error as 
%\begin{eqnarray}\label{1113_123}
%&&\epsilon_{\rho}\leq \frac{1}{\pi^2}\Big[
%\frac{e^{-2\pi \tau N_{\mu}/ \mu_{\max}}}{(1-e^{-\pi \tau N_{\mu}/ \mu_{\max}})^2}+\frac{\mu^2_{\max}}{ n}
%+\frac{e^{-\pi \tau N_{\mu}/ \mu_{\max}}}{1-e^{-\pi \tau N_{\mu}/ \mu_{\max}}}\frac{2\mu_{\max}}{\sqrt{n}}
%\Big].
%\end{eqnarray}
%The asymptotically optimal allocation with respect to the total amount of measurements $T=nN_{\mu}$ is
%\begin{eqnarray}
 %   N_\mu^{\star} = \frac{\mu_{\max}}{2\pi\tau}\,\log T,\quad 
%n^{\star} \;=\; \frac{T}{N_\mu^{\star}}             =\frac{2\pi\tau}{\mu_{\max}}\,
 %                   \frac{T}{\log T}
%\end{eqnarray}
%We select $\tau^{\star}= \ln{(C)}/\mu_{\mathrm{max}}$, $C\geq 1$:
%\begin{eqnarray}
%    N_\mu^{\star} = \frac{\mu^2_{\max}}{2\pi\ln{(C)}}\,\log T,\quad 
%n^{\star}=\frac{2\pi \ln{(C)}}{\mu^2_{\max}}\,
 %                   %\frac{T}{\log T}
%\end{eqnarray}
%Then 
%\begin{eqnarray}
%\epsilon_{\rho}\bigl(n^{\star},N_\mu^{\star}\bigr)=
% \frac{\mu_{\max}^{3}}{2\pi^{3}\tau}\,
 %      \frac{\log T}{T}\,\bigl(1+o(1)\bigr)=\frac{\mu_{\max}^{4}}{2\pi^{3}\ln{(C)}}\,      \frac{\log T}{T}\,\bigl(1+o(1)\bigr)=\tilde{O}\left(\frac{1}{T}\right).
%\end{eqnarray}
\subsection*{3.2. Estimation Error of the Trace Characteristics of Quantum States}
\par To estimate the trace characteristic of the density matrix operators we do similar steps. However, in this case their is no need in DFT, it is a simple integral approximation with the finite sum. The tail correction can be done in the similar fashion. 
We denote the product of two CFs as
\begin{eqnarray}\label{1752}
    A(\mu,\nu) =\phi_1(1;\mu,\nu) \phi_2^{\star}(1;\mu,\nu),
\end{eqnarray}
and its point‐wise product and the KCFE as:
\begin{eqnarray}
   A(\mu_i,\nu_j)=\phi_1(1;\mu_i,\nu_j) \phi_2^{\star}(1;\mu_i,\nu_j) ,\quad \hat{A}(\mu_i,\nu_j)=\hat{\phi}_{1,nh}(1;\mu_i,\nu_j) \hat{\phi}_{2,nh}^{\star}(1;\mu_i,\nu_j).
\end{eqnarray}
\par We want to approximate the integral
   \begin{eqnarray}
I\equiv   \tr{ \boldsymbol{\rho}_{1}\boldsymbol{\rho}_{2}}
  &=& \frac{1}{2\pi}\iint A(\mu,\nu) d\mu d\nu.
\end{eqnarray}
Similarly to the previous section, we first truncate the latter integral:
\begin{eqnarray}
    I_{\mathrm{trunc}}
  &=& \frac{1}{2\pi}\int\limits_{-\mu_{\mathrm{max}}} ^{\mu_{\mathrm{max}}}\int\limits_{-\nu_{\mathrm{max}}} ^{\nu_{\mathrm{max}}} A(\mu,\nu) d\mu d\nu.
\end{eqnarray}
The product of two CFs \eqref{1752} is a CF too. We assume it has  a Gaussian tail (that is a typical situation for HO based states~\cite{markovich2024not}) and its modulus can be bounded by $Ce^{-\beta(\mu^2+\nu^2)}$. The truncation error is  
\begin{eqnarray}\label{1347}
|I-I_{\mathrm{trunc}}|^2=O\left( e^{-4\beta\min{(\mu_{\mathrm{max}}^2,\nu_{\mathrm{max}}^2)}}\right).    
\end{eqnarray}
Then we discretise the finite integral
   \begin{eqnarray}
\tilde{I}=\frac{1}{2\pi} \sum_{i=1}^{N_{\mu}}  \sum_{j=1}^{N_{\nu}} A(\mu_i, \nu_j)\Delta\mu \Delta\nu
\end{eqnarray}
For smooth integrands and uniform grid the discretization error is $|I_{\mathrm{trunc}}-\tilde{I}|^2=O(\Delta\mu^4 + \Delta\nu^4)$.
\par Finally, we use the KCFE:
\begin{eqnarray}
     \hat{\tilde{I}}=\widehat{\tr{( \boldsymbol{\rho}_{1,nh}\boldsymbol{\rho}_{2,nh})}}=\frac{1}{2\pi}\sum_{i=1}^{N_{\mu}}  \sum_{j=1}^{N_{\nu}} \hat{A}(\mu_i,\nu_j)\Delta \mu\Delta \nu.
\end{eqnarray}
The error arising from the kernel estimation method is
\begin{eqnarray}
  |\tilde{I}- \hat{\tilde{I}}|^2&=&\left(\frac{\Delta \mu\Delta \nu}{2\pi}\right)^2\left|\sum_{i=1}^{N_{\mu}}  \sum_{j=1}^{N_{\nu}} (A(\mu_i, \nu_j)-\hat{A}(\mu_i, \nu_j))\right|^2\\\nonumber
  &=&\left(\frac{\Delta \mu\Delta \nu}{2\pi}\right)^2\sum_{i=1}^{N_{\mu}}  \sum_{j=1}^{N_{\nu}} (A(\mu_i, \nu_j)-\hat{A}(\mu_i, \nu_j))(A(\mu_i, \nu_j)-\hat{A}(\mu_i, \nu_j))^{\star}
\end{eqnarray}
If we take the expectation, then
\begin{eqnarray}
       \mathbb{E}|\tilde{I}- \hat{\tilde{I}}|^2&=&\left(\frac{\Delta \mu\Delta \nu}{2\pi}\right)^2\mathbb{E}\left[\sum_{i=1}^{N_{\mu}}  \sum_{j=1}^{N_{\nu}} (A(\mu_i, \nu_j)-\hat{A}(\mu_i, \nu_j))(A(\mu_i, \nu_j)-\hat{A}(\mu_i, \nu_j))^{\star}\right]\\\nonumber
       &=&\left(\frac{\Delta \mu\Delta \nu}{2\pi}\right)^2\sum_{i=1}^{N_{\mu}}  \sum_{j=1}^{N_{\nu}}\mathbb{E}\left[|A(\mu_i, \nu_j)-\hat{A}(\mu_i, \nu_j)|^2\right].
\end{eqnarray}
We use $|ab-cd|^2\leq 2|a-c|^2|b|^2+2|c|^2|b-d|^2$, to write
\begin{eqnarray}
    &&\mathbb{E}\left[|A(\mu_i, \nu_j)-\hat{A}(\mu_i, \nu_j)|^2\right]
   \\\nonumber
   &\leq& 2\mathbb{E}\left[|{\phi}_{1}(\mu_i,\nu_j) -\hat{\phi}_{1,nh}(\mu_i,\nu_j)|^2|{\phi}_{2}^{\star}(\mu_i,\nu_j)|^2\right]
    +2\mathbb{E}\left[|{\phi}_{2}^{\star}(\mu_i,\nu_j) -\hat{\phi}_{2,nh}^{\star}(\mu_i,\nu_j)|^2|\hat{\phi}_{1,nh}(\mu_i,\nu_j)|^2\right]\\\nonumber
    &=&2\mathbb{E}\left[|{\phi}_{1}(\mu_i,\nu_j) -\hat{\phi}_{1,nh}(\mu_i,\nu_j)|^2\right]|{\phi}_{2}^{\star}(\mu_i,\nu_j)|^2+ 2\mathbb{E}\left[|{\phi}_{2}^{\star}(\mu_i,\nu_j) -\hat{\phi}_{2,nh}^{\star}(\mu_i,\nu_j)|^2\right]\mathbb{E}\left[|\hat{\phi}_{1,nh}(\mu_i,\nu_j)|^2\right]\\\nonumber
    &=&2\Bigg[\left(\frac{1-|\phi_1(1;\mu_i,\nu_j)|^2}{n}
-\frac{(1-|\phi_1(1;\mu_i,\nu_j)|^2)^2}{4n^2\,|\phi_1(1;\mu_i,\nu_j)|^2}\right)|{\phi}_{2}^{\star}(1;\mu_i,\nu_j)|^2\\\nonumber
&+& \left(\frac{1-|\phi_2(1;\mu_i,\nu_j)|^2}{n}
-\frac{(1-|\phi_2(1;\mu_i,\nu_j)|^2)^2}{4n^2\,|\phi_2(1;\mu_i,\nu_j)|^2}\right)|\phi_K(th)|^2 \left(\frac{1}{n}+\left(1-\frac{1}{n}\right)|\phi_1(1;\mu_i,\nu_j)|^2\right)
    \Bigg]\\\nonumber
&=&2\left(|{\phi}_{2}^{\star}(1;\mu_i,\nu_j)|^2(1-|\phi_1(1;\mu_i,\nu_j)|^2)+|\phi_K(th)|^2|{\phi}_{1}(1;\mu_i,\nu_j)|^2(1-|\phi_2(1;\mu_i,\nu_j)|^2)\right)O\left(\frac{1}{n}\right),
\end{eqnarray}
where we use the Gaussian kernel with  the optimal bandwidth \eqref{1924}, using the approximation \begin{eqnarray}
|\phi_K(th)|^2\approx\exp{-\frac{2}{\mu^2+\nu^2}\frac{1-|\phi_1(1;\mu_i,\nu_j)|^2}{2n|\phi_1(1;\mu_i,\nu_j)|^2}}=1+O(1/n).\end{eqnarray} 
%Then minimize for $h$ and uts the small $h$ exponent approximation. The bandwidth that minimizes the latter expression is
%\begin{equation}
%h_{\mathrm{\mathrm{opt}}} = \frac{2}{t \sqrt{\mu^2 + \nu^2}} \cdot 
%\sqrt{\frac{- \left(3n |\phi_1(t)|^2 |\phi_2(t)|^2 
%- n |\phi_1(t)|^2 
%- 2n |\phi_2(t)|^2 
%- 2 |\phi_1(t)|^2 |\phi_2(t)|^2 
%+ 2 |\phi_1(t)|^2 
%+ 2 |\phi_2(t)|^2 
%- 2 \right)
%}{2n^2 |\phi_1(t)|^2 |\phi_2(t)|^2 
%- 11n |\phi_1(t)|^2 |\phi_2(t)|^2 
%+ 3n |\phi_1(t)|^2 
%+ 8n |\phi_2(t)|^2 
%+ 8 |\phi_1(t)|^2 |\phi_2(t)|^2 
%- 8 |\phi_1(t)|^2 
%- 8 |\phi_2(t)|^2 
%+ 8
%}
%}
%\end{equation}
We assume that $\Delta\mu =2\mu_{\mathrm{max}}/N_{\mu}$, $\Delta\nu =2\nu_{\mathrm{max}}/N_{\nu}$. Then 
\begin{eqnarray}
\!\!\!\!\!\!\!\!\mathbb{E}|\tilde{I}- \hat{\tilde{I}}|^2\leq \frac{8\mu^2_{\mathrm{max}}\nu_{\mathrm{max}}^2}{\pi^2N_{\mu}^2 N_{\nu}^2}\sum_{i=1}^{N_{\mu}}  \sum_{j=1}^{N_{\nu}}\left(|{\phi}_{2}^{\star}(1;\mu_i,\nu_j)|^2(1-|\phi_1(1;\mu_i,\nu_j)|^2)+|{\phi}_{1}(1;\mu_i,\nu_j)|^2(1-|\phi_2(1;\mu_i,\nu_j)|^2)\right)O\left(
   \frac{1}{n}\right).
\end{eqnarray}
In this case, the total amount of measurements is 
$T_{\mu,\nu}=nN_{\mu}N_{\nu}$, since we vary both $\mu$ and $\nu$ parameters. We select the parameters $\mu_{\mathrm{max}}$ and $\nu_{\mathrm{max}}$ to be of the rate \eqref{1202}.
Then the latter error is
\begin{eqnarray}
    \mathbb{E}|\tilde{I}- \hat{\tilde{I}}|^2= \tilde{O}\left(\frac{1}{T_{\mu,\nu}}\right).
\end{eqnarray}
We also assumed that the truncation error is corrected  similarly to the previous section with the parametric tail corrector. Then the final error is scaling as $\mathbb{E}|{\tr{( \boldsymbol{\rho}_{1}\boldsymbol{\rho}_{2})}}- \widehat{\tr{( \boldsymbol{\rho}_{1,nh}\boldsymbol{\rho}_{2,nh})}}|^2=\tilde{O}(1/T_{\mu,\nu})$.
\begin{example}
     Let us assume that $\boldsymbol{\rho}$ is a true state and $\widehat{\boldsymbol{\rho}_{nh}}$ is its KQSE. Then from the previous derivations we know that the MSE of the estimate is
\begin{eqnarray}
\mathrm{MSE}(\widehat{\tr{( \boldsymbol{\rho}\boldsymbol{\rho}_{nh})}})= \mathbb{E}|{\tr{( \boldsymbol{\rho}^2)}}- \widehat{\tr{( \boldsymbol{\rho}\boldsymbol{\rho}_{nh})}}|^2 = \tilde{O}(1/T_{\mu,\nu}).
\end{eqnarray}
If the state $\boldsymbol{\rho}$ is pure, then $D(\boldsymbol{\rho}, \boldsymbol{\rho})=0$ and 
\begin{eqnarray}
  \widehat{D(\boldsymbol{\rho},\boldsymbol{\rho}_{nh})}=\sqrt{1-\widehat{\tr{( \boldsymbol{\rho}\boldsymbol{\rho}_{nh})}}} .
\end{eqnarray}
The MSE of the estimate of the squared trace distance is exactly equal to the MSE of the overlap estimator:
\begin{eqnarray}
\mathrm{MSE}(\widehat{D^2(\boldsymbol{\rho},\boldsymbol{\rho}_{nh})})=\mathrm{MSE}(\widehat{\tr{( \boldsymbol{\rho}\boldsymbol{\rho}_{nh})}})=\tilde{O}(1/T_{\mu,\nu}).
\end{eqnarray}
Apply Chebyshev's inequality, we get
\begin{eqnarray}
    P\left(|\widehat{D^2(\boldsymbol{\rho},\boldsymbol{\rho}_{nh})})-{D^2(\boldsymbol{\rho},\boldsymbol{\rho})})|\geq \delta\right)\leq \frac{\mathrm{MSE}(\widehat{D^2(\boldsymbol{\rho},\boldsymbol{\rho}_{nh})})}{\delta^2}=\tilde{O}\left(\frac{1}{T_{\mu,\nu}\delta^2}\right),
\end{eqnarray}
since ${D^2(\boldsymbol{\rho},\boldsymbol{\rho})}=0$. 
\end{example}
{
\color{black}
\section*{S4. Simulation and Experimental Study}\label{sec_sim}

\par In this section, we present a comprehensive numerical study of tomographic reconstruction methods applied to both simulated and experimental data. The analysis focuses on the estimation of the tomographic probability distribution $\mathcal{W}(X|\mu,\nu)$, its characteristic function $\hat{\phi}(1;\mu,\nu)$ and the kernel of the quantum state $\rho(y,y')$, with particular emphasis on the robustness of nonparametric approaches in comparison with parametric alternatives.
\par We first perform a simulation study in which samples are generated from analytically known tomographic probability density functions. The true tomograms are obtained from the theoretical expressions derived below, and synthetic measurement data are generated using importance sampling. For each fixed pair of tomographic parameters $(\mu_k,\nu_k)$, we consider multiple independent datasets of size $n$ and investigate the reconstruction accuracy as a function of the sample size.
\par The KDE estimation is performed using: Gaussian kernels with the bandwidth selected via cross-validation (GridSearchCV), as well as Epanechnikov kernels with adaptively chosen bandwidth based on Silverman's rule of thumb~\cite{silverman2018density}. For comparison, histogram-based estimators are also included, with the bin width chosen according to the Freedman--Diaconis rule.
\par To quantify the reconstruction quality, we compute the MISE defined in Eq.~\eqref{1640_10}. For each configuration, we generate $L$ independent realizations $\{\hat{\mathcal{W}}^{(l)}_{nh}(x|\mu_k,\nu_k)\}_{l=1}^{L}$ of the estimator, each based on an independent dataset. The MISE is approximated numerically as
\begin{eqnarray}
     \widehat{\operatorname{MISE}}(\hat{\mathcal{W}}_{nh}(x|\mu_k,\nu_k)) =
     \frac{1}{L}\sum_{l=1}^{L}
     \left[
     \sum_{i=1}^{m}
     \bigl(\hat{\mathcal{W}}^{(l)}_{nh}(x_i|\mu_k,\nu_k)
     -\mathcal{W}(x_i|\mu_k,\nu_k)\bigr)^2
     \,\Delta x
     \right],
\end{eqnarray}
where the integral is approximated on a discrete grid $\{x_i\}_{i=1}^{m}$ with spacing $\Delta x$ using a Riemann sum. The study is performed for three sample sizes, $n=\{500,1000,2000\}$.
\par In addition to KDE-based reconstruction, we also consider MLE of the tomogram under parametric assumptions. Specifically, the tomographic distribution is modeled as a Gaussian mixture model (GMM) with a fixed number of components. We analyze the cases of two- and three-component mixtures (GMM$=2$ and GMM$=3$) in order to assess the sensitivity of MLE-based reconstruction to model misspecification. This comparison allows us to directly contrast nonparametric KDE with parametric MLE approaches in a controlled setting.
\par A similar analysis is carried out for the CF of the tomogram, reconstructed using the KCFE. The reconstruction accuracy is quantified by the mean squared error
\begin{eqnarray}
     \widehat{\operatorname{MSE}}(\hat{\phi}_{nh}(1;\mu_k,\nu_k)) =
     \frac{1}{L}\sum_{l=1}^{L}
     \left|
     \hat{\phi}^{(l)}_{nh}(1;\mu_k,\nu_k)
     -\phi(1;\mu_k,\nu_k)
     \right|^2 .
\end{eqnarray}
We also reconstruct the characteristic function using MLEs of the tomogram. Specifically, the CF is obtained as the Fourier transform of the MLE-reconstructed tomographic probability distribution under the Gaussian mixture model assumptions discussed above. The resulting CF estimates are then compared with those obtained directly via the KCFE. 
\par Finally, we investigate the performance of the proposed methods on noisy data and experimental homodyne measurements. We consider tomograms corresponding to noisy quantum states modeled as mixtures with additive Gaussian noise, as in Eq.~\eqref{1033}. From the noisy measurement data, the KCFE yields a set of estimates $\{\hat{\phi}_X(1;\mu_i,\nu_j)\}$ for $i=1,\dots,N_\mu$ and $j=1,\dots,N_\nu$, which are effectively purified at the level of the CF. These purified estimates are then used as input for KQSE, enabling the reconstruction of the density-matrix kernel $\rho(y,y')$ and associated trace characteristics.
\par In addition, for comparison, we also reconstruct the density-matrix kernel $\rho(y,y')$ using CFs obtained from MLEs of the tomogram.  This enables a direct comparison between density matrix reconstructions based on KQSE and MLE in the presence of noise and experimental imperfections.
\par The numerical error of the reconstructed kernel of the density matrix is quantified by the computational version of the bound~\eqref{1033_33},
\begin{eqnarray}\label{1033_333}
   \hat{L}_{\infty}(\boldsymbol{\rho}_{nh})
   \equiv
   \max_{y,y'\in[y_{\mathrm{min}},y_{\mathrm{max}}]}
   \mathbb{E}\!\left(
   |{\rho}(y,y')-\hat{\tilde{\rho}}_{\mathrm{max}}(y,y')|^2
   \right),
\end{eqnarray}
which provides an upper bound on the maximal reconstruction error over the sampled grid. In all numerical experiments we take $y,y'\in[-4,4]$.

\subsection*{4.1. Cat states}
\par We first performed a simulation study for the coherent cat state (CCS): 
\begin{eqnarray}\label{1642}
    \ket{\psi}_{\mathrm{CCS}}=N_c\sum\limits_{j=1}^{3} \ket{a_j},\quad  \ket{a_j} = \ket{a\, e^{\frac{2\pi i (j-1)}{3}}}, \quad j = 1, 2, 3,
\end{eqnarray}
where $\ket{a}$ are the Glauber coherent states  of the harmonic oscillator. The normalization constant is 
\begin{eqnarray}
|N_c|^{-2} = \sum_{j,k=1}^{3} \exp\Bigl(-|a|^2 + |a|^2 e^{2\pi i (k-j)/3}\Bigr).
\end{eqnarray}
The coordinate representation of the coherent state reads as
\begin{eqnarray}
    \psi_{a}(x)=\langle x|a\rangle=\pi^{-1/4}\exp{\left(-\frac{x^2}{2}-\frac{|a|^2}{2}+\sqrt{2}a x-\frac{a^2}{2}\right)},
\end{eqnarray}
and its tomogram is known to be equal to~\cite{manko1997quantum}:
\begin{eqnarray}\label{1300}
   \mathcal{W}(X|{a},\mu,\nu)&=&  \mathcal{W}_0(X|\beta) 
 e^{iXs}e^{d},\\\nonumber
s&\equiv &\frac{\sqrt{2}\left(
(\nu-i\mu)a^{\star} -(\nu+i\mu)a \right)}{\beta^2},\\\nonumber
d&=&-|a|^2+\frac{
(\nu+i\mu)^2a^2 +(\nu-i\mu)^2(a^{\star})^2}{2\beta^2},
\end{eqnarray}
where $\mathcal{W}_0(X|\beta)$ is given by \eqref{1530}, $\beta=\sqrt{\mu^2+\nu^2}$.
The corresponding CF is 
\begin{eqnarray}\label{1300_1}
  {\phi}(t;{a},\mu,\nu)&=&\phi_0(t+s;\beta)  e^{d}.
\end{eqnarray}
The coordinate representation of the cat state \eqref{1642} is
\begin{eqnarray}\label{848}
\rho(y,y')=%|N_c|^2\sum\limits_{j,k=1}^3\braket{y|a_j}\braket{a_k|y'}=
|N_c|^2 \sum\limits_{j,k=1}^3\psi_{ae^{2\pi i(j-1)/3}}(y)\psi^{\star}_{ae^{2\pi i(k-1)/3}}(y').
\end{eqnarray}
According to \eqref{1645_5} the tomogram of the cat states \eqref{1642} can be written as
\begin{eqnarray}\label{1328}
    \mathcal{W}_{\mathrm{CCS}}(X|a,\mu,\nu)
&=&\mathcal{W}_0(X|\beta)|N_c|^2\sum\limits_{j,k=1}^{3}e^{iXs(j,k)}e^{d(j,k)},\\\nonumber
s(j,k)&\equiv&\frac{\sqrt{2}\left(
(\nu-i\mu)a^{\star} e^{-\frac{2\pi i}{3}(k-1)}-(\nu+i\mu)a e^{\frac{2\pi i}{3}(j-1)}\right)}{\beta^2},\\\nonumber
d(j,k)&\equiv&-|a|^2+\frac{\left(
(\nu+i\mu)^2a^2 e^{\frac{4\pi i}{3}(j-1)}+(\nu-i\mu)^2(a^{\star})^2 e^{-\frac{4\pi i}{3}(k-1)}\right)}{2\beta^2}.
\end{eqnarray}
One can see that if $k,j=1$ then $s(1,1)=s$, $d(1,1)=d$ and the tomogram coincides with the single coherent state tomogram  \eqref{1300}. The PDF \eqref{1328} is a mixture of the Gaussian distributions~\cite{robertson1969descriptive}. Note that the sum of Gaussian PDFs in general is not a Gaussian PDF itself.  The CF of this cat state is
\begin{eqnarray}\label{1132}
    {\phi}_{\mathrm{CCS}}(t;{a},\mu,\nu)&=&|N_c|^2\sum\limits_{j,k=1}^{3}\phi_0(t+s(j,k);\beta)  e^{d(j,k)}.
\end{eqnarray}
The results of the KDE of the tomogram  are summarized in Fig.~\ref{fig:kde_comparison}, where we compare the estimates against the true tomogram for $n\in\{500,1000,2000\}$ sample points. The multimodal PDF is parameterized by the values:
 $\mu = 0.8$, $\nu = 1.2$,  $a = 1.0 + 0.5i$. 
\begin{figure}[h]
  \centering
  \begin{minipage}[b]{0.32\textwidth}
    \centering
    \includegraphics[width=\linewidth]{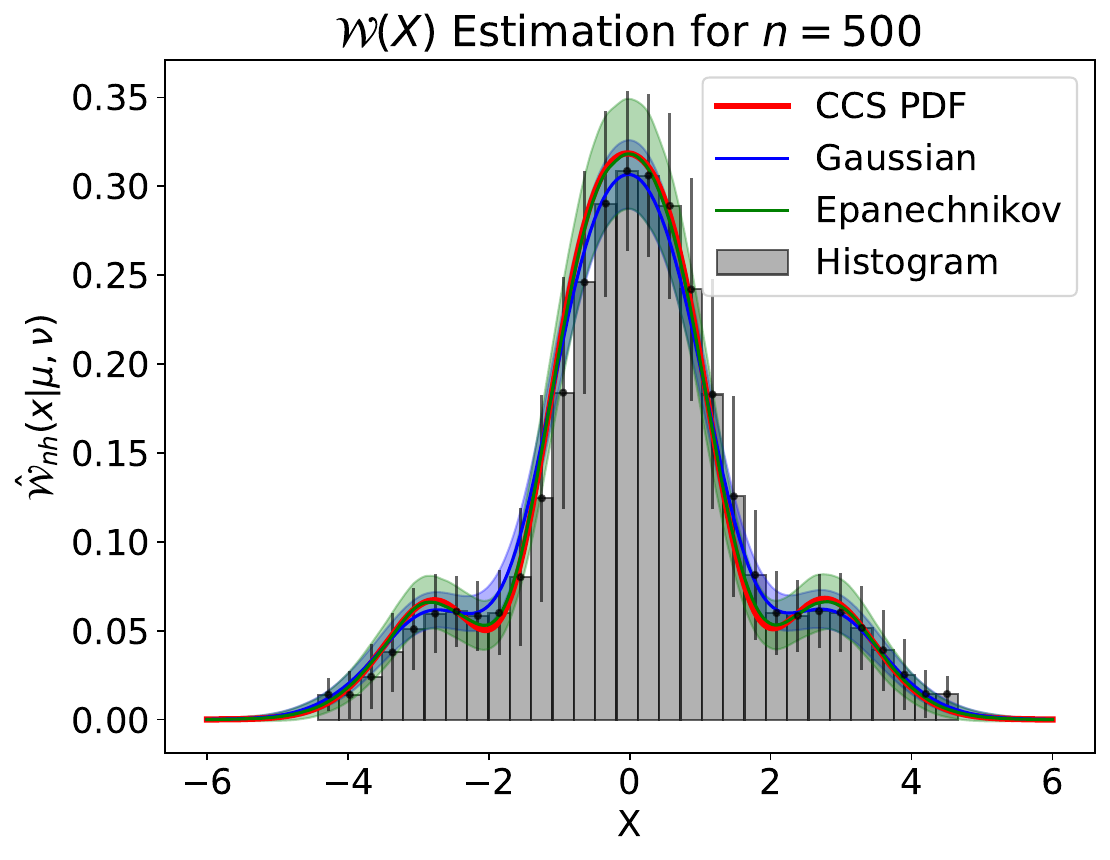}
  \end{minipage}\hfill
  \begin{minipage}[b]{0.32\textwidth}
    \centering
    \includegraphics[width=\linewidth]{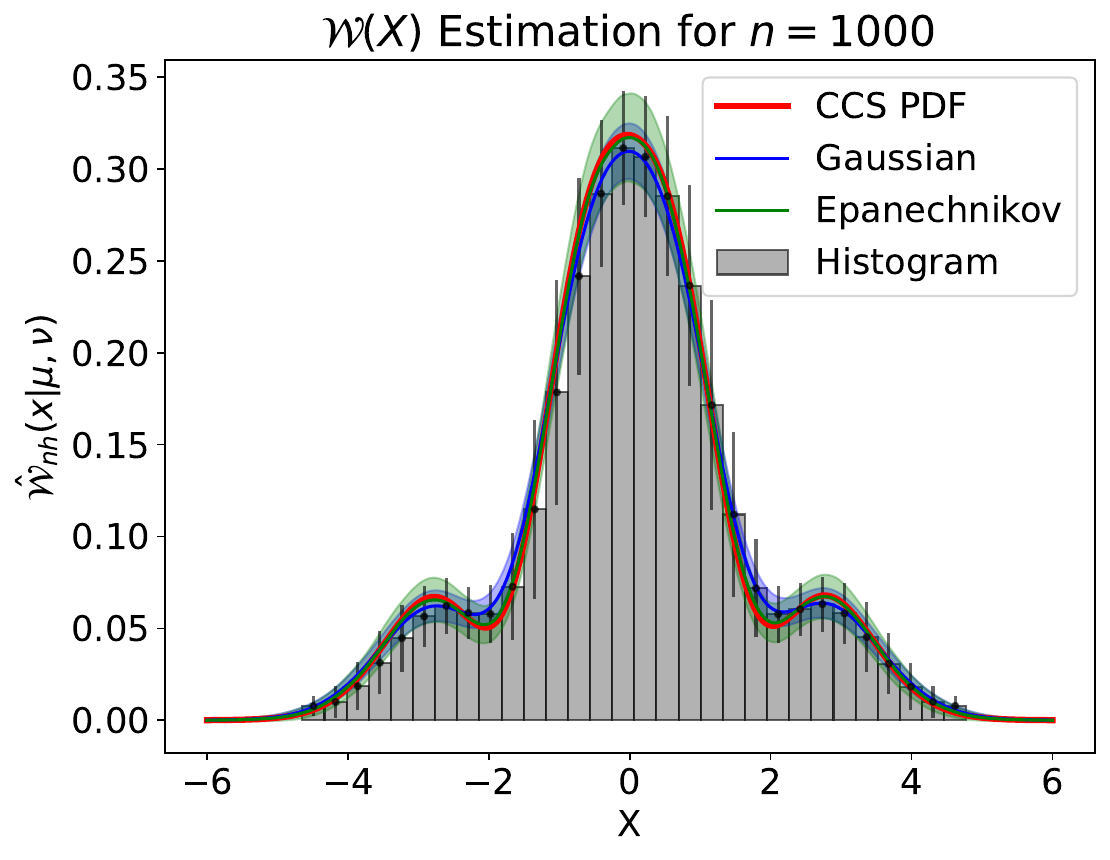}
  \end{minipage}\hfill
  \begin{minipage}[b]{0.32\textwidth}
    \centering
    \includegraphics[width=\linewidth]{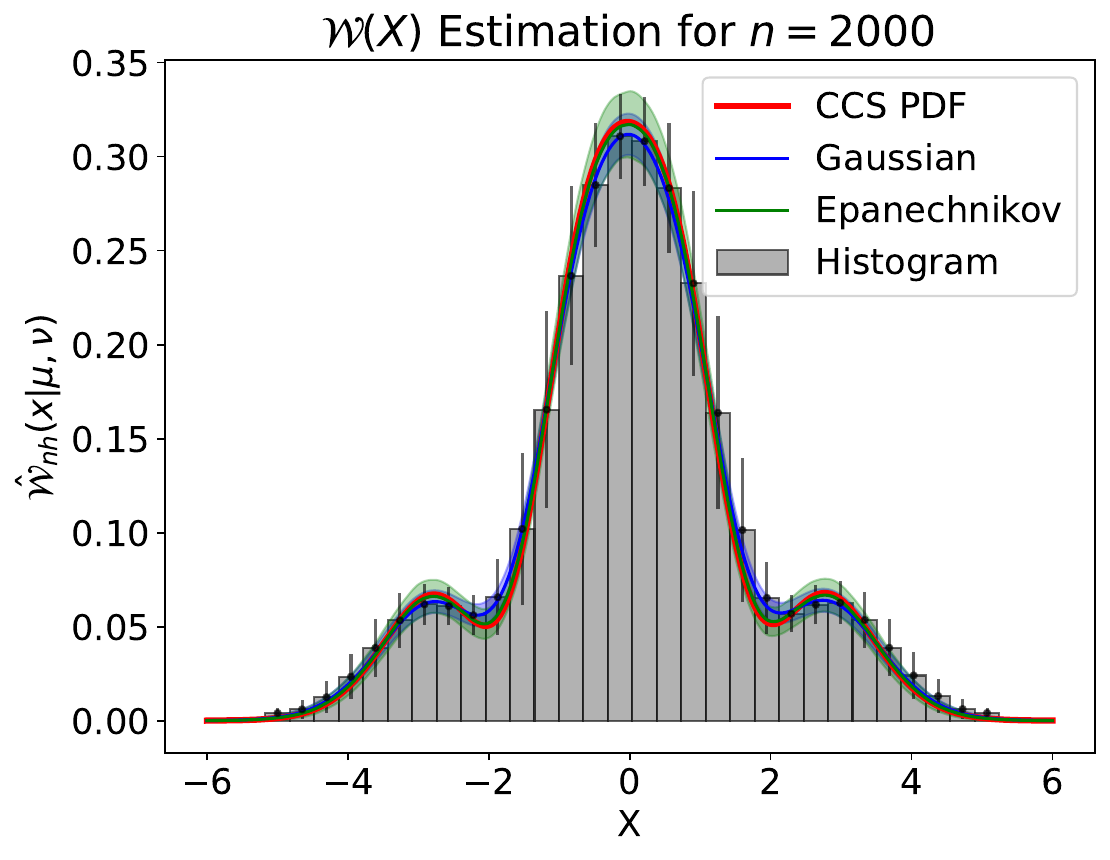}
  \end{minipage}
  \caption{Comparison of KDEs and histogram with the CCS tomogram (PDF) \eqref{1328} for different sample sizes $n=\{500,1000,2000\}$ for a fixed   $\mu = 0.8$, $\nu = 1.2$,  $a = 1.0 + 0.5i$.}
  \label{fig:kde_comparison}
\end{figure}
One can see that the KDE method is robust to the multimodal structure of the tomogram. 
In the average the PDF estimation with the Epanechnikov kernel performs
better than the one with the Gaussian kernel that can be seen in the plots and is predicted by classical KDE theory. However, Epanechnikov KDE has larger variance
(green error bars). That means that its MISE per repetition  is in general larger than for the Gaussian kernel (see the Table.~\ref{tab:kde_comparison}).
\begin{table}[h]
    \centering
    \begin{tabular}{|c|c|c|c|}
        \hline
        \textbf{n} & \textbf{KDE Type} & \textbf{Avg.\ Bandwidth} & \textbf{KDE}  \\
        \hline
        \multirow{2}{*}{500}  & Gaussian      & \(3.691\times10^{-1}\) & \(2.202\times10^{-3}\) \\[-1ex]
                              & Epanechnikov  & \(3.334\times10^{-1}\) & \(3.279\times10^{-3}\) \\
        \hline
        \multirow{2}{*}{1000} & Gaussian      & \(2.970\times10^{-1}\) & \(1.307\times10^{-3}\) \\[-1ex]
                              & Epanechnikov  & \(2.904\times10^{-1}\) & \(1.914\times10^{-3}\) \\
        \hline
        \multirow{2}{*}{2000} & Gaussian      & \(2.562\times10^{-1}\) & \(7.556\times10^{-4}\) \\[-1ex]
                              & Epanechnikov  & \(2.531\times10^{-1}\) & \(1.091\times10^{-3}\) \\
        \hline
    \end{tabular}
    \caption{Comparison of the optimal bandwidth and MISE of the KDE of the CCS tomogram \eqref{1328} for the different sample sizes $n=\{500,1000,2000\}$, averaged over $L=10^3$.}
    \label{tab:kde_comparison}
\end{table}
The results fully coincide with the theoretical prediction \eqref{1640_10}.
\par The KDE estimated CCS tomogram was compared with MLE based on parametric Gaussian mixture models (GMM) with two and three components ($GMM = 2$ and $GMM = 3$).
Fig.~\ref{fig:mle_comparison} shows the reconstructed CCS tomograms for a sample size $n=1000$. While the MLE based on a two-component Gaussian mixture fails to reproduce the true tomogram, increasing the model complexity to three components improves the fit but still leads to visible deviations. In contrast, the KDE-based reconstruction provides an accurate approximation of the true tomogram without requiring a specific parametric model.
\begin{figure}[h]
  \centering
    \includegraphics[width=0.5\linewidth]{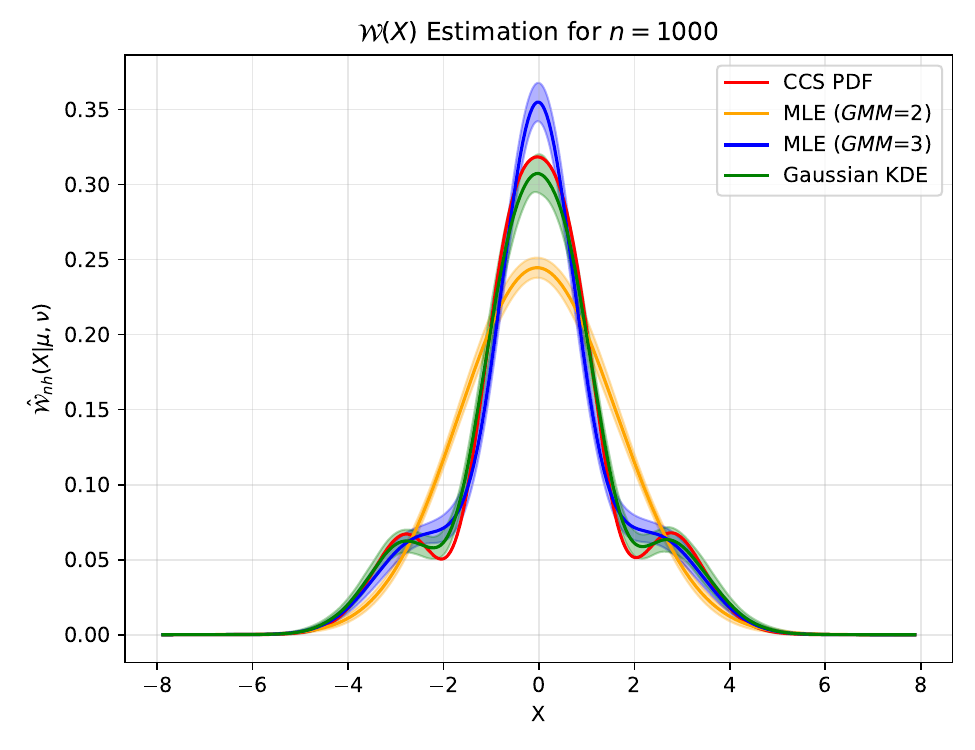}
  \caption{Comparison of CCSs tomogram reconstructions  obtained using MLE under Gaussian mixture model assumptions with two and three components (GMM = 2 and GMM = 3), and the corresponding KDE with the Gaussian kernel. The true tomographic probability density~\eqref{1328} is shown for reference. The reconstructions are based on a sample of size 
$n=1000$, with the parameters  $\mu = 0.8$, $\nu = 1.2$,  $a = 1.0 + 0.5i$. The results illustrate the sensitivity of MLE-based reconstructions to the choice of the parametric model.}
  \label{fig:mle_comparison}
\end{figure}
This qualitative observation is confirmed quantitatively in Table~\ref{tab:mle_comparison}, where the MISE of the tomogram estimate is reported for different sample sizes.
\begin{table}[h]
    \centering
    \begin{tabular}{|c|c|c|c|}
        \hline
        \textbf{n} & \textbf{MLE}, $GMM=2$ & \textbf{MLE}, $GMM=3$ & \textbf{Gaussian KDE} \\
        \hline
        500  & \(1.565\times10^{-2}\)       & \(3.323\times10^{-3}\) & \(2.202\times10^{-3}\) \\
        \hline
        1000 & \(1.588\times10^{-2}\)       & \(2.718\times10^{-3}\) & \(1.307\times10^{-3}\) \\
        \hline
        2000 & \(1.599\times10^{-2}\)       & \(1.975\times10^{-3}\) & \(7.556\times10^{-4}\) \\
        \hline
    \end{tabular}
    \caption{MISE of tomogram reconstructions for the CCS~\eqref{1328} obtained using MLE with Gaussian mixture models and Gaussian KDE, shown for different sample sizes $n=\{500,1000,2000\}$. The results are averaged over $L=10^3$ independent realizations. The tomographic parameters are fixed to $\mu=0.8$, $\nu=1.2$, $a=1+0.5i$.}
    \label{tab:mle_comparison}
\end{table}
For all considered values of 
$n$, the KDE-based estimator consistently outperforms both MLE-based reconstructions. Notably, the performance of the MLE is highly sensitive to the assumed number of mixture components, whereas the KDE approach remains robust and systematically yields lower reconstruction errors since its fully data-driven.
\par In Fig.~\ref{fig:kde2_comparison} the kernel CF estimate is presented for $n\in\{500,1000,2000\}$ sample points. We take $t=1$, fix $\nu=1.2$, $a = 1 + 0.5i$ and vary $\mu\in[0, 8]$ for $N_{\mu}=160$ points. Since the CF has real and imaginary parts, we plot them separately. Histogram method cannot directly estimate the CF, while the KCFE performs with the great accuracy.  In this simulation we only use the Gaussian kernel with the optimal bandwidth \eqref{1640_10}. Since this bandwidth is dependent from the unknown CF, first we apply the Silvermans's rule of thumb method, similarly to the procedure for the PDF estimation.
\begin{figure}[h]
  \centering
  \begin{minipage}[b]{0.32\textwidth}
    \centering
    \includegraphics[width=\linewidth]{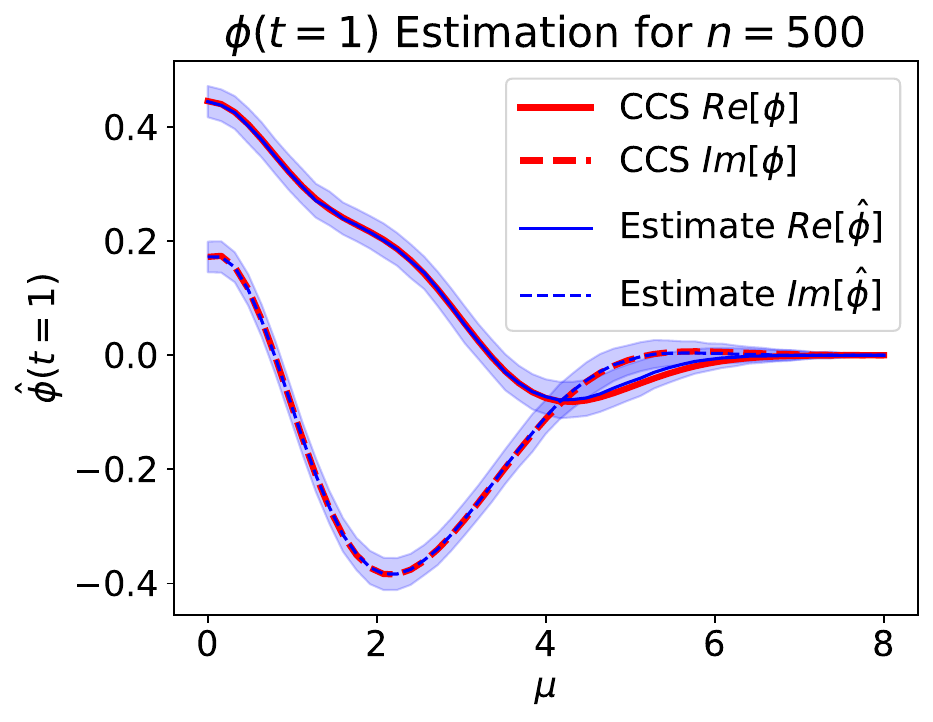}
  \end{minipage}\hfill
  \begin{minipage}[b]{0.32\textwidth}
    \centering
    \includegraphics[width=\linewidth]{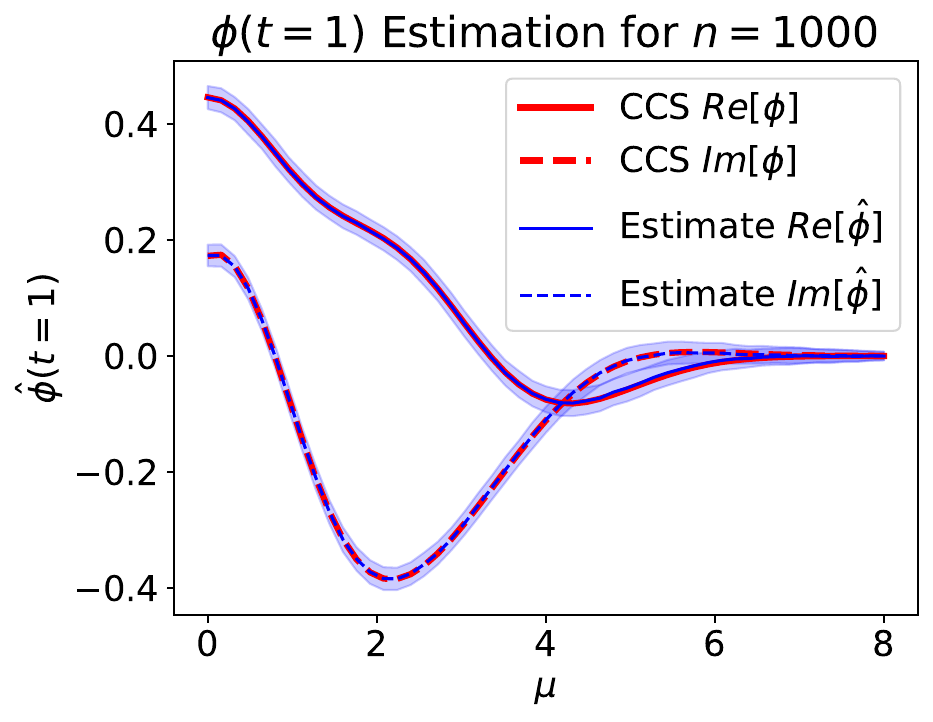}
  \end{minipage}\hfill
  \begin{minipage}[b]{0.32\textwidth}
    \centering
    \includegraphics[width=\linewidth]{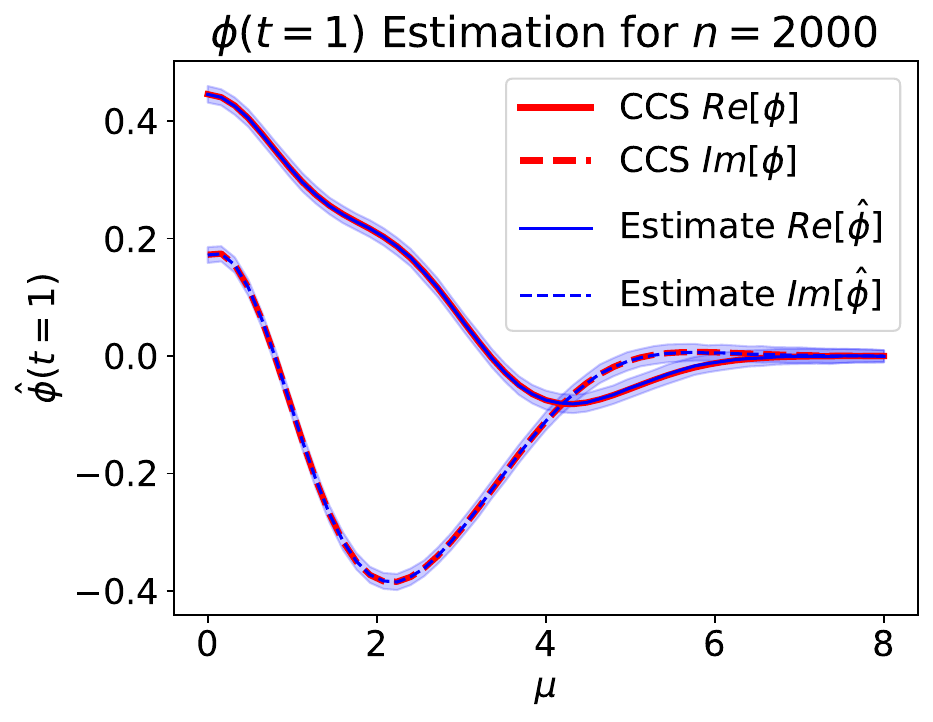}
  \end{minipage}
  \caption{Comparison of kernel estimate  with the CCS CF  \eqref{1132} for different sample sizes $n=\{500,1000,2000\}$ for the fixed $\nu=1.2$ from $\mu\in[0,8]$, $N_{\mu}=160$.}
  \label{fig:kde2_comparison}
\end{figure}
Table.~\ref{tab:kde2_comparison} presents the optimal bandwidth and  MSE for the KCFE averaged  for $L=10^3$ repetitions.
\begin{table}[h]
    \centering
    \begin{tabular}{|c|c|c|c|}
        \hline
        \textbf{n} & \textbf{Data type} & \textbf{Avg.\ Bandwidth} & \textbf{Gaussian KCFE} \\
        \hline
        \multirow{3}{*}{500}  & Noiseless  & \(9.387\times10^{-2}\) & \(1.756\times10^{-3}\) \\[-1ex]
                              & Noisy      & \(7.123\times10^{-2}\) & \(7.946\times10^{-3}\) \\[-1ex]
                              & Corrected  & \(8.157\times10^{-2}\) & \(1.898\times10^{-3}\) \\
        \hline
        \multirow{3}{*}{1000} & Noiseless  & \(6.400\times10^{-2}\) & \(9.036\times10^{-4}\) \\[-1ex]
                              & Noisy      & \(4.896\times10^{-2}\) & \(7.267\times10^{-3}\) \\[-1ex]
                              & Corrected  & \(5.534\times10^{-2}\) & \(9.070\times10^{-4}\) \\
        \hline
        \multirow{3}{*}{2000} & Noiseless  & \(4.408\times10^{-2}\) & \(4.247\times10^{-4}\) \\[-1ex]
                              & Noisy      & \(3.386\times10^{-2}\) & \(7.028\times10^{-3}\) \\[-1ex]
                              & Corrected  & \(3.822\times10^{-2}\) & \(4.505\times10^{-4}\) \\
        \hline
    \end{tabular}
    \caption{The optimal bandwidth and the MSE of the KCFE with the Gaussian kernel of the CCS states CF for the sample sizes $n=\{500,1000,2000\}$, averaged over $L=10^3$ repetitions. We compare two data types: ideal data and the linear mixture \eqref{1033} with the Gaussian noise, $\kappa=0.85$. For the noisy data we compare the CF reconstruction  without and with the noise correction procedure \eqref{CF_fil}.}
    \label{tab:kde2_comparison}
\end{table}
In Fig.~\ref{fig:mse} we present the MSE for the kernel tomogram estimator and the KCFE based on the same data sets averaged over $L=10^3$ estimations. The MSE for the CF estimator converges in accordance to the theoretical prediction \eqref{1224}. 
\begin{figure}[h]
  \centering
   \includegraphics[width=0.5\linewidth]{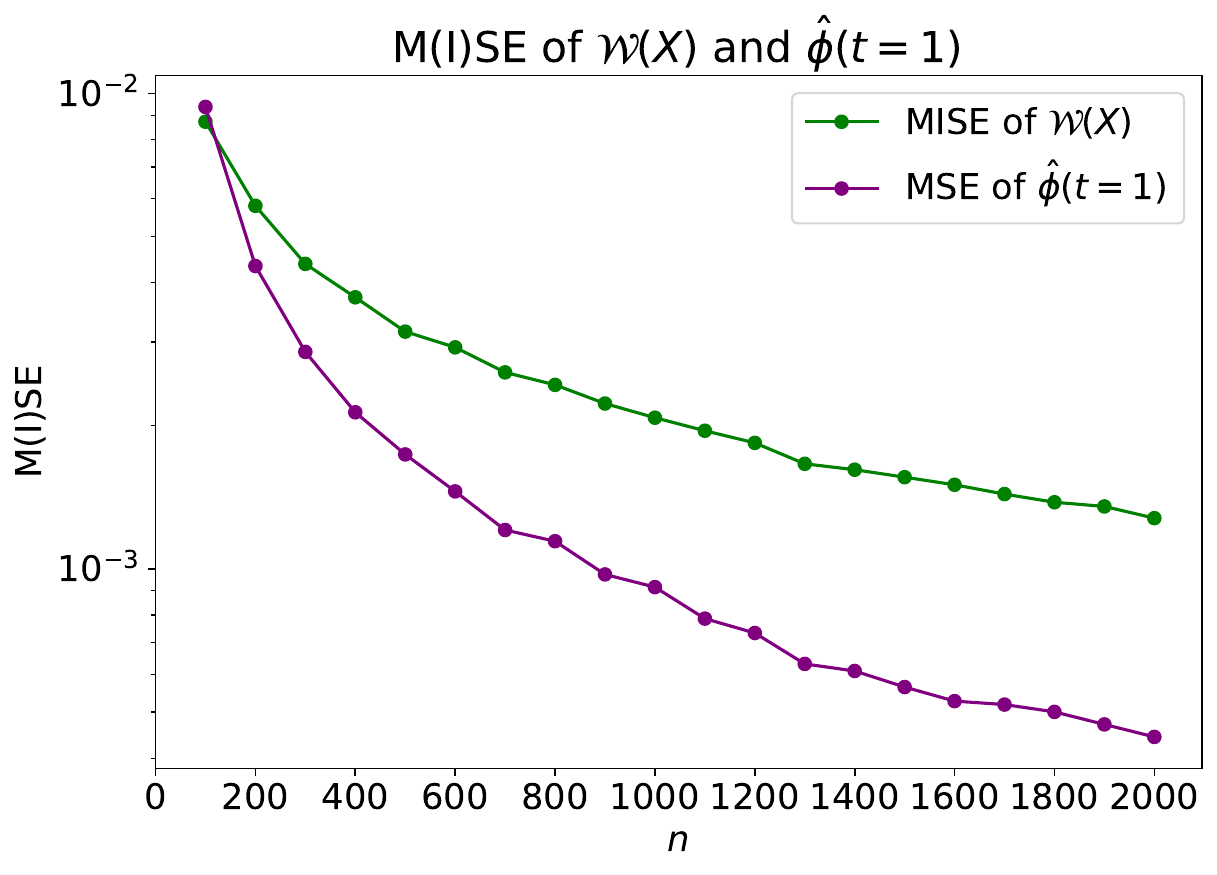}
   \caption{MISE the kernel tomogram estimator and the MSE of the KCFE based on the same data sets averaged over $L=10^3$ estimations against $n$ sample points.}
  \label{fig:mse}
\end{figure}
\par Next we consider the case when the observable data is the mixture \eqref{1033}. The observable of interest $X$ is  distributed with $ \mathcal{W}(X|a,\mu,\nu)$, the noise component $Y$ is a normally distributed random variable with $N(0,1)$, and we fix $\kappa=0.85$. The  KCFE is performed according to \eqref{CF_fil} and is shown in Fig.~\ref{fig:kde3_comparison}. One can see that the correction procedure significantly increases the performance of the estimation. 
\begin{figure}[h]
  \centering
  \begin{minipage}[b]{0.32\textwidth}
    \centering
    \includegraphics[width=\linewidth]{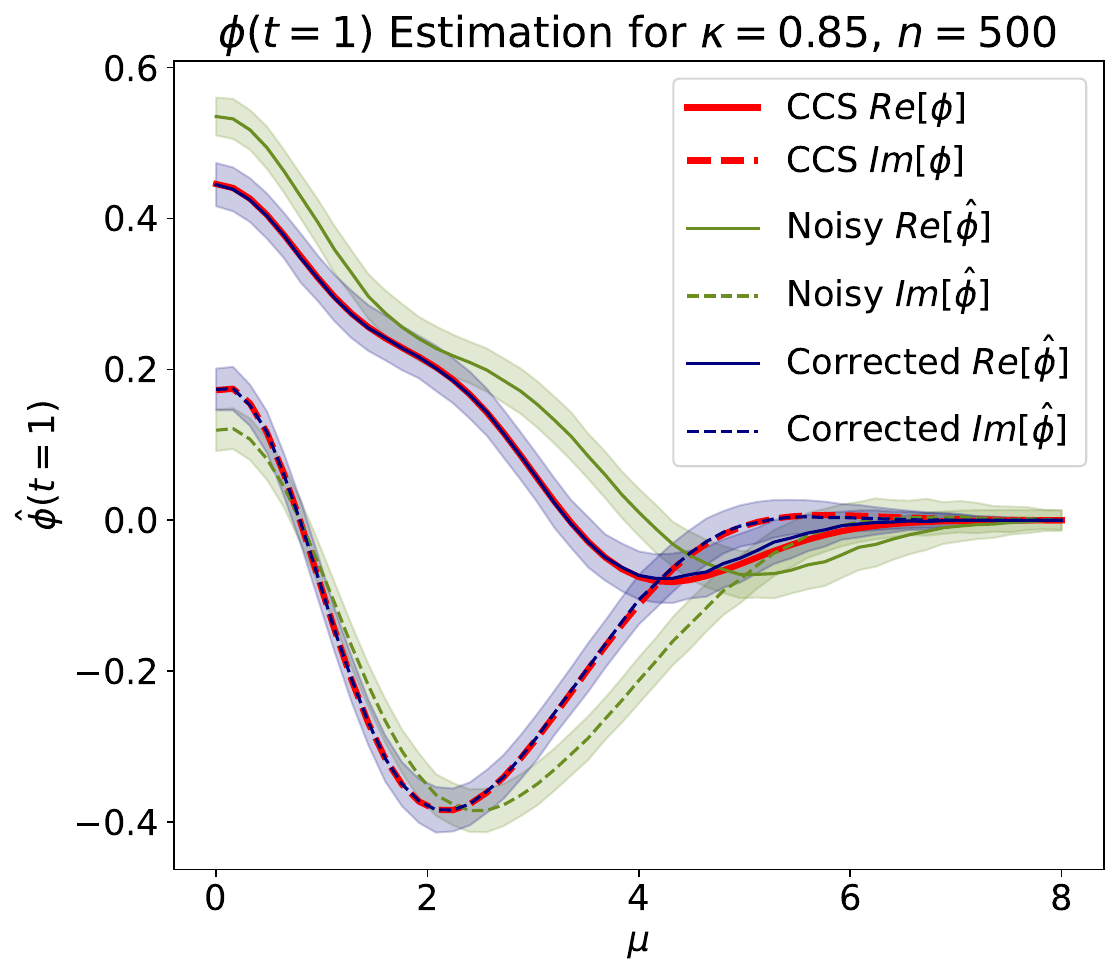}
  \end{minipage}\hfill
  \begin{minipage}[b]{0.32\textwidth}
    \centering
    \includegraphics[width=\linewidth]{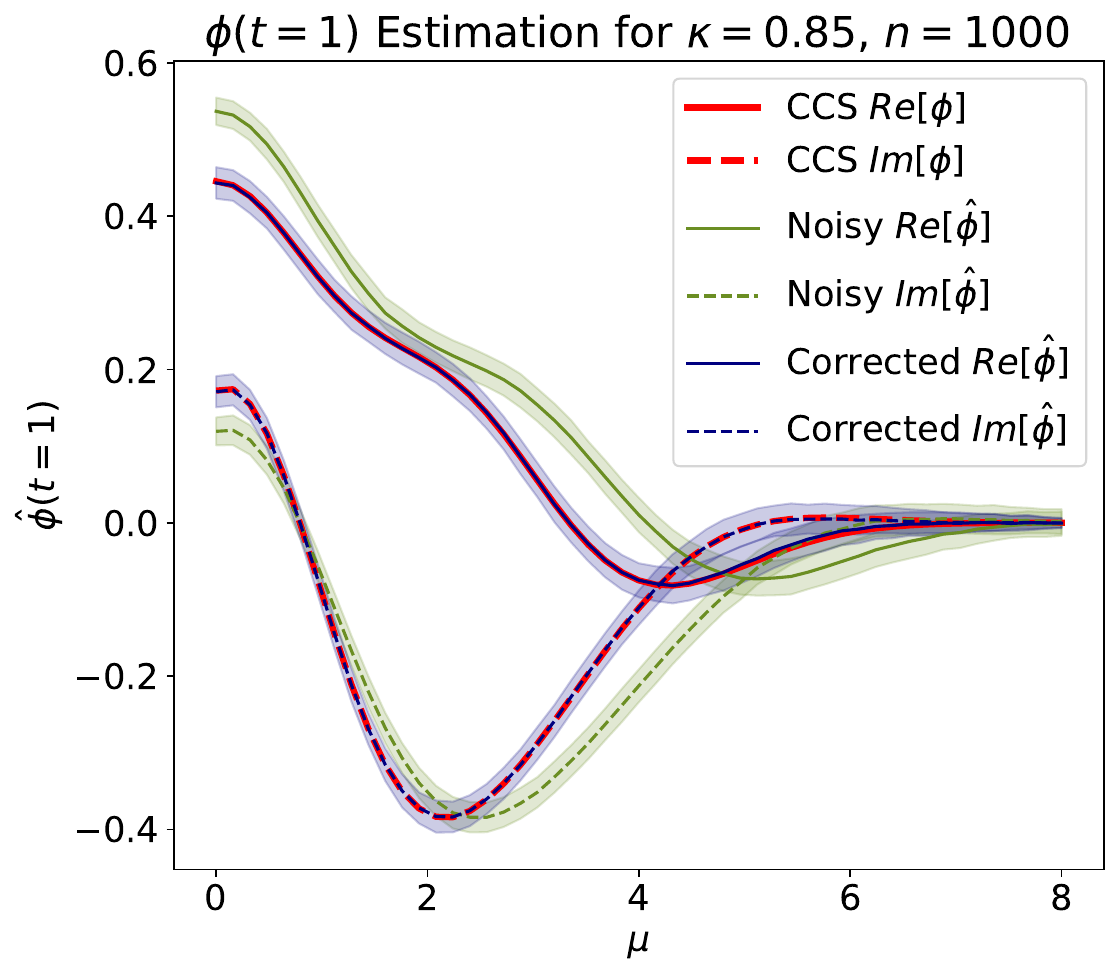}
  \end{minipage}\hfill
  \begin{minipage}[b]{0.32\textwidth}
    \centering
    \includegraphics[width=\linewidth]{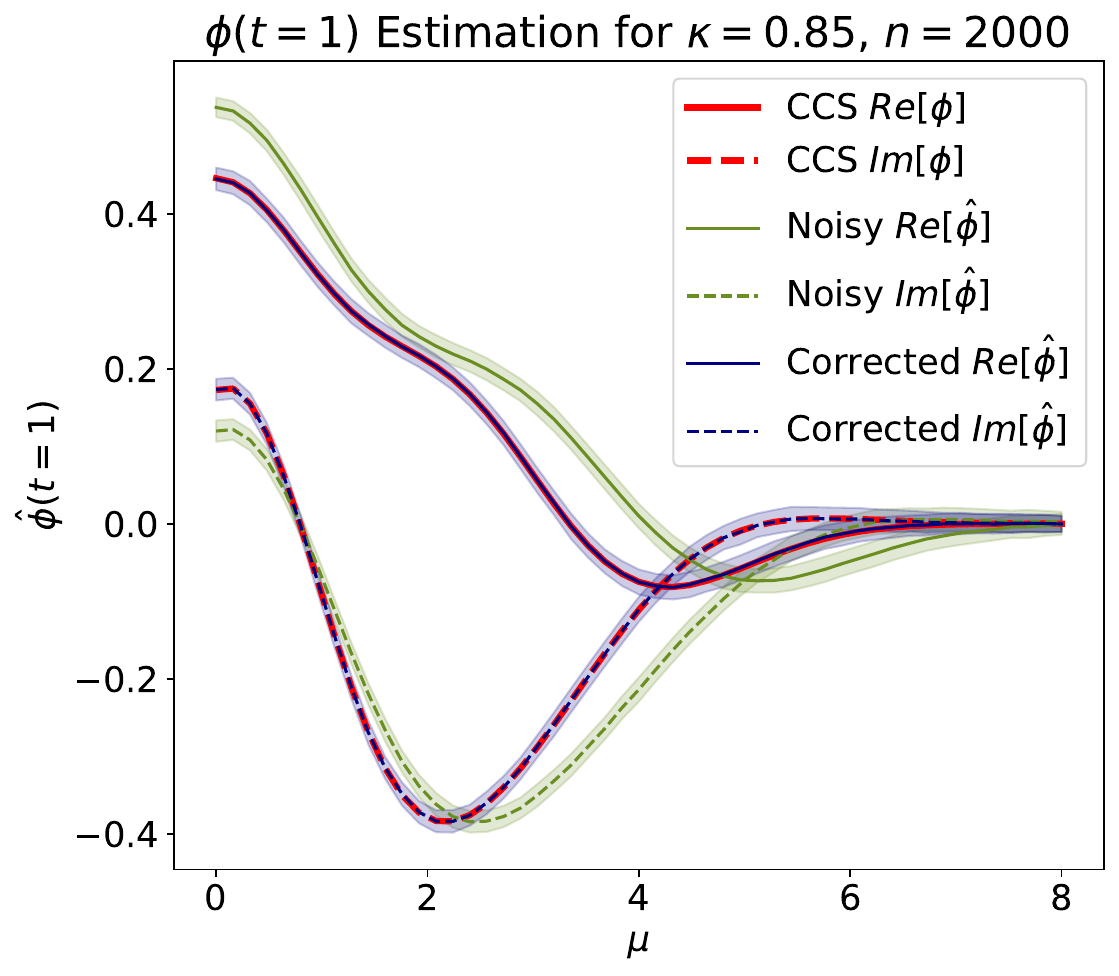}
  \end{minipage}
  \caption{Comparison of the theoretical CCS CF with the KCFE with the Gaussian kernel obtained from a noisy data without/with correction for the different sample sizes $n=\{500,1000,2000\}$, $\kappa = 0.85$.}
  \label{fig:kde3_comparison}
\end{figure}
The MSE of the KCFE using the noisy data is provided in Table.~\ref{tab:kde2_comparison}. We compare the estimation without and with correction  \eqref{CF_fil}. One can see that our correction build in the  kernel CF estimation  provides the results nearly identical to the noiseless case.
\par We also reconstruct the CF of the CCS by taking the Fourier transform of the MLE tomogram. The resulting CF estimates are shown in Fig.~\ref{fig:mle2_comparison}.
\begin{figure}[h]
  \centering
    \includegraphics[width=0.5\linewidth]{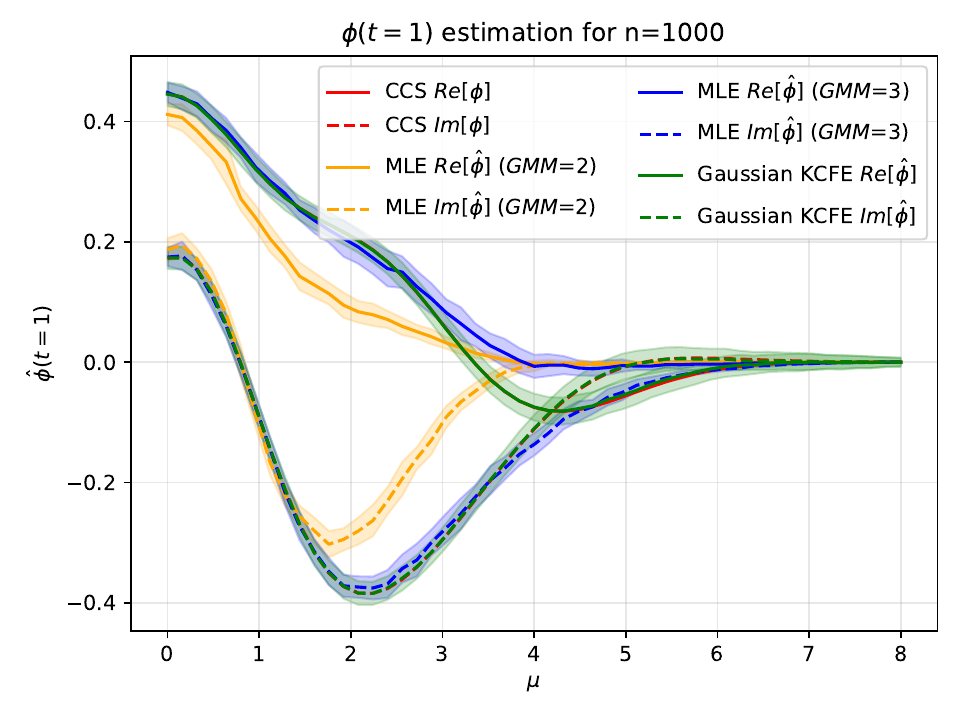}
  \caption{Comparison of CCS CF estimator obtained using MLE with Gaussian mixture models and KCFE for a sample size $n=1000$ and fixed $\nu=1.2$.}
  \label{fig:mle2_comparison}
\end{figure}The MLE reconstruction with $GMM = 2$ fails to reproduce the true CF, while increasing the model complexity to $GMM = 3$ leads to an improved but still inaccurate reconstruction in certain regions of the parameter $\mu$. In contrast, the KCFE provides an accurate reconstruction of the CF across the entire parameter range.
\par This  is confirmed in Table~\ref{tab:kde2_comparison}, where the MSE of the CF estimates is reported for different sample sizes. For all considered values of $n$, the KCFE consistently outperforms the MLE-based estimators with both $GMM = 2$ and $GMM = 3$.
Although both MLE and KCFE exhibit the same theoretical convergence rate under ideal assumptions, the performance of the MLE critically depends on the correctness of the underlying parametric model. In the present case, the Gaussian mixture model does not provide sufficient structural information for the MLE to achieve optimal performance, whereas the KCFE remains robust and accurate without requiring explicit model assumptions.
\begin{table}[h]
    \centering
    \begin{tabular}{|c|c|c|c|c|}
        \hline
        \textbf{n} & \textbf{Data type} & \textbf{MLE}, $GMM=2$ & \textbf{MLE}, $GMM=3$ &  \textbf{Gaussian  KCFE} \\
        \hline
        \multirow{3}{*}{500}  & Noiseless  & \(8.823\times10^{-3}\) & \(1.945\times10^{-3}\) & \(1.756\times10^{-3}\) \\[-1ex]
                              & Noisy      & \(3.493\times10^{-3}\) & \(9.729\times10^{-3}\) & \(7.946\times10^{-3}\) \\[-1ex]
                              & Corrected  & \(8.992\times10^{-3}\) & \(2.005\times10^{-3}\) & \(1.898\times10^{-3}\) \\
        \hline
        \multirow{3}{*}{1000} & Noiseless  & \(7.442\times10^{-3}\) & \(9.826\times10^{-4}\) & \(9.036\times10^{-4}\) \\[-1ex]
                              & Noisy      & \(1.841\times10^{-3}\) & \(7.844\times10^{-3}\) & \(7.267\times10^{-3}\) \\[-1ex]
                              & Corrected  & \(7.608\times10^{-3}\) & \(9.992\times10^{-4}\) & \(9.070\times10^{-4}\) \\
        \hline
        \multirow{3}{*}{2000} & Noiseless  & \(7.083\times10^{-3}\) & \(5.743\times10^{-4}\) & \(4.247\times10^{-4}\) \\[-1ex]
                              & Noisy      & \(1.360\times10^{-3}\) & \(7.217\times10^{-3}\) & \(7.028\times10^{-3}\) \\[-1ex]
                              & Corrected  & \(7.125\times10^{-3}\) & \(5.856\times10^{-4}\) & \(4.505\times10^{-4}\) \\
        \hline
    \end{tabular}
    \caption{MSE of CF reconstructions for the CCS obtained using MLE with Gaussian mixture models and KCFE with the Gaussian kernel, shown for different sample sizes $n=\{500,1000,2000\}$. The results are averaged over $L=10^3$ independent realizations. The parameters are fixed to $\mu=0.8$, $\nu=1.2$, $a=1+0.5i$.  Two types of data are considered: ideal data and noisy data modeled as a linear mixture with additive Gaussian noise, Eq.~\eqref{1033}, with $\kappa=0.85$. For the noisy case, CF reconstruction results are compared with and without the noise-correction procedure~\eqref{CF_fil}.}
    \label{tab:kde2_comparison}
\end{table}

\par 
 Using the set of $N_{\mu}$ KCFEs  estimates $\{\hat{\phi}_X(1;\mu_i,y-y'\}$, $i\in[1,N_{\mu}]$  we do KQSE procedure to estimate the kernel  $\rho(y, y')$, where $N_{\mu} $ different values of $\mu$ are uniformly distributed in the interval $[-\mu_{\max}, \mu_{\max}]$, for each $(y, y')$ pair. Then we evaluate the maximum by $(y,y')$ taken in the interval $(y,y')\in[-4,4] $ from the distances \eqref{1033_333} between the true kernel $\rho(y,y')$ and its estimate $\hat{\tilde{\rho}}_{\mathrm{max}}(y,y')$. Table~\ref{tab:combined} reports the latter supremum norm error for the reconstructed quantum state under noiseless, noisy, and noise-corrected conditions, while varying the sample size $n$, maximal parameter $\mu_{\max}$, or number of measurement settings $N_{\mu}$. When one parameter is varied, the others are fixed to $n=500$, $\mu_{\max}=8$, and $N_{\mu}=160$. Among the parameters varied, the sample size $n$  and number of measurement settings $N_{\mu}$ exhibit the strongest influence on the reconstruction accuracy, yielding over threefold reduction in the supremum norm error, while the effect of increasing $\mu_{\max}$ is slightly weaker but still significant. For $\mu_{\mathrm{max}}=4$ the oscillation behavior of the CF is still present and the integral truncation error is bigger then the one with $\mu_{\mathrm{max}}=\{6,8\}$. Since the oscillation behavior of CF is already small enough at $\mu_{\mathrm{max}}=6$, the estimation errors for $\mu_{\mathrm{max}}=\{6,8\}$ are of the similar rate. In all cases, the noise-corrected values closely match the noiseless benchmarks, demonstrating the efficacy of the correction procedure.
\begin{figure}[htbp]
    \centering
    \begin{minipage}{0.49\textwidth}
        \includegraphics[width=\linewidth]{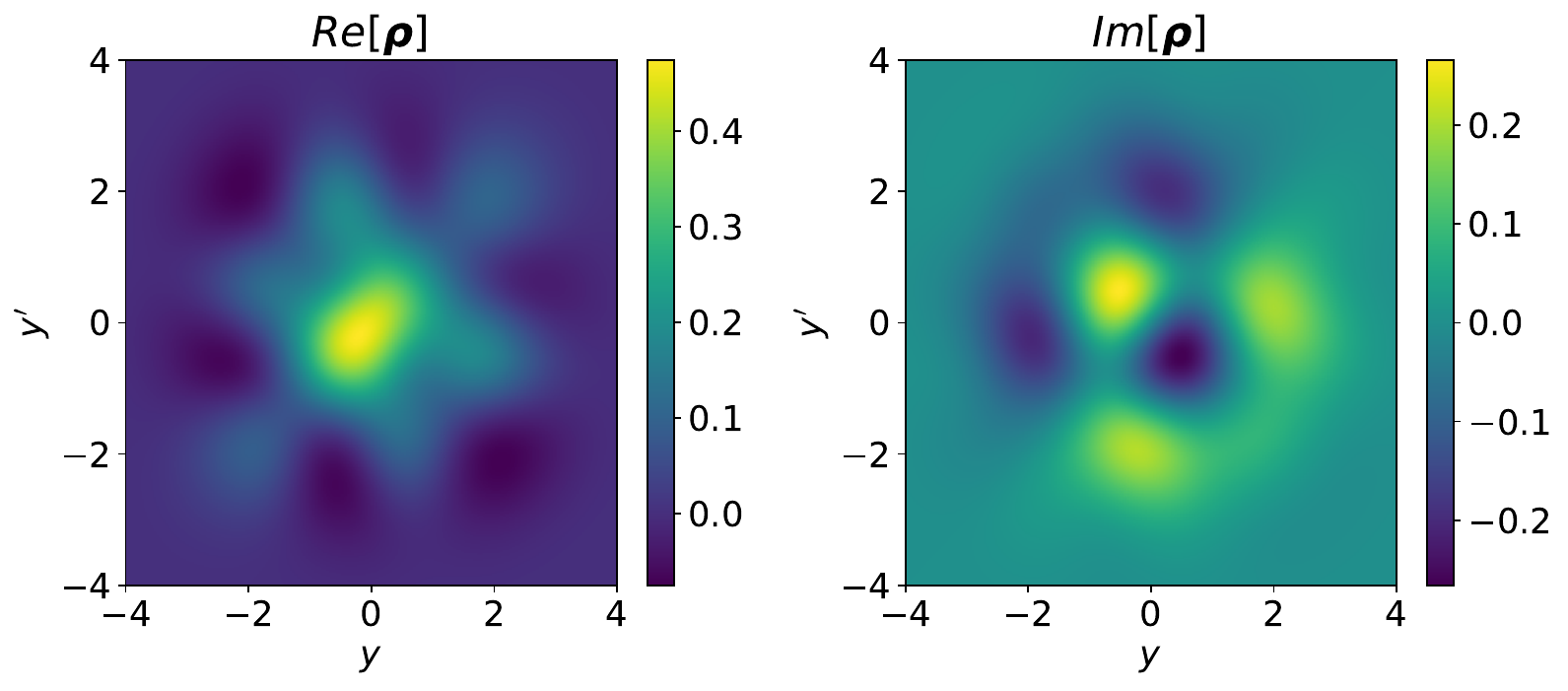}
        \caption*{(a)}
    \end{minipage}%
    \begin{minipage}{0.49\textwidth}
        \includegraphics[width=\linewidth]{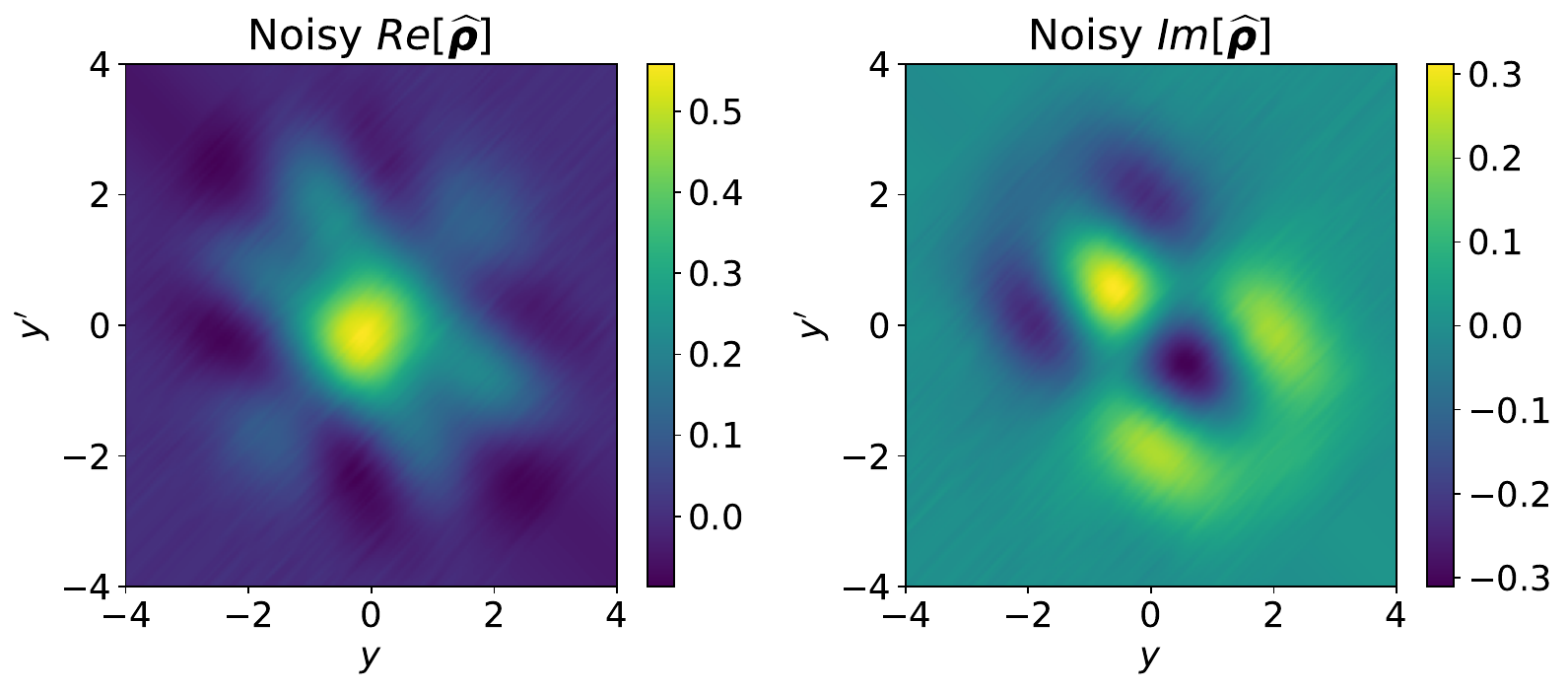}
        \caption*{(b)}
    \end{minipage}

    \vspace{0.3em}

    \begin{minipage}{0.49\textwidth}
        \includegraphics[width=\linewidth]{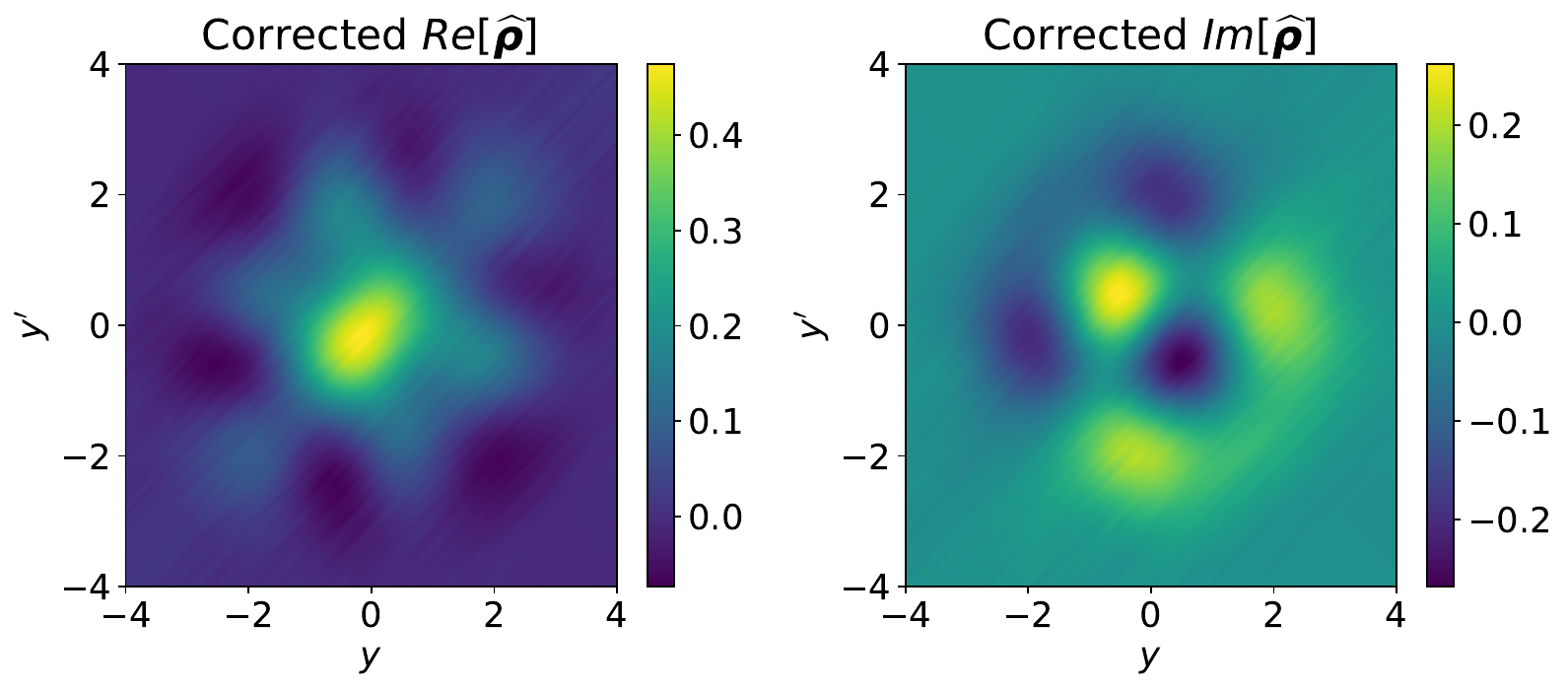}
        \caption*{(c)}
    \end{minipage}%
    \begin{minipage}{0.49\textwidth}
        \includegraphics[width=\linewidth]{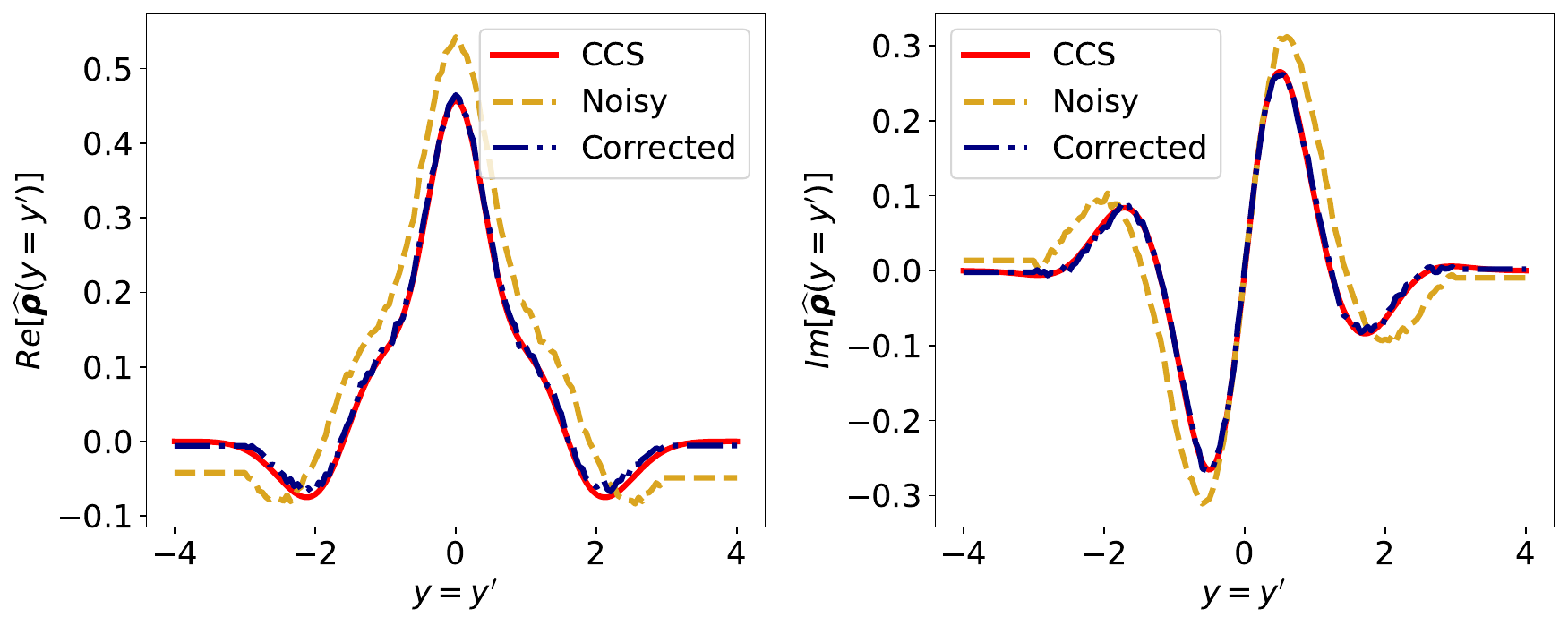}
        \caption*{(d)}
    \end{minipage}

\vspace{0.3em}
\begin{minipage}{0.49\textwidth}
        \includegraphics[width=\linewidth]{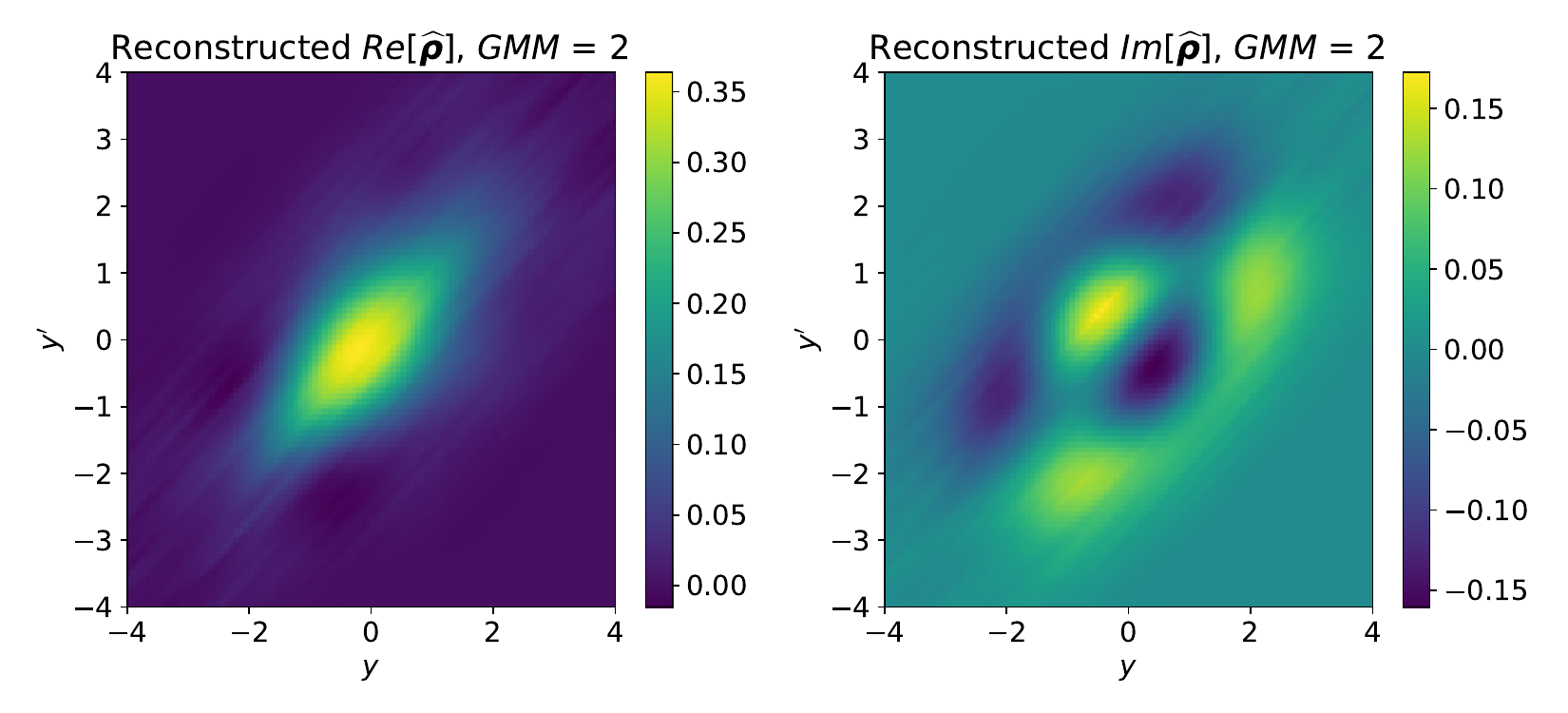}
        \caption*{(e)}
    \end{minipage}%
    \begin{minipage}{0.49\textwidth}
        \includegraphics[width=\linewidth]{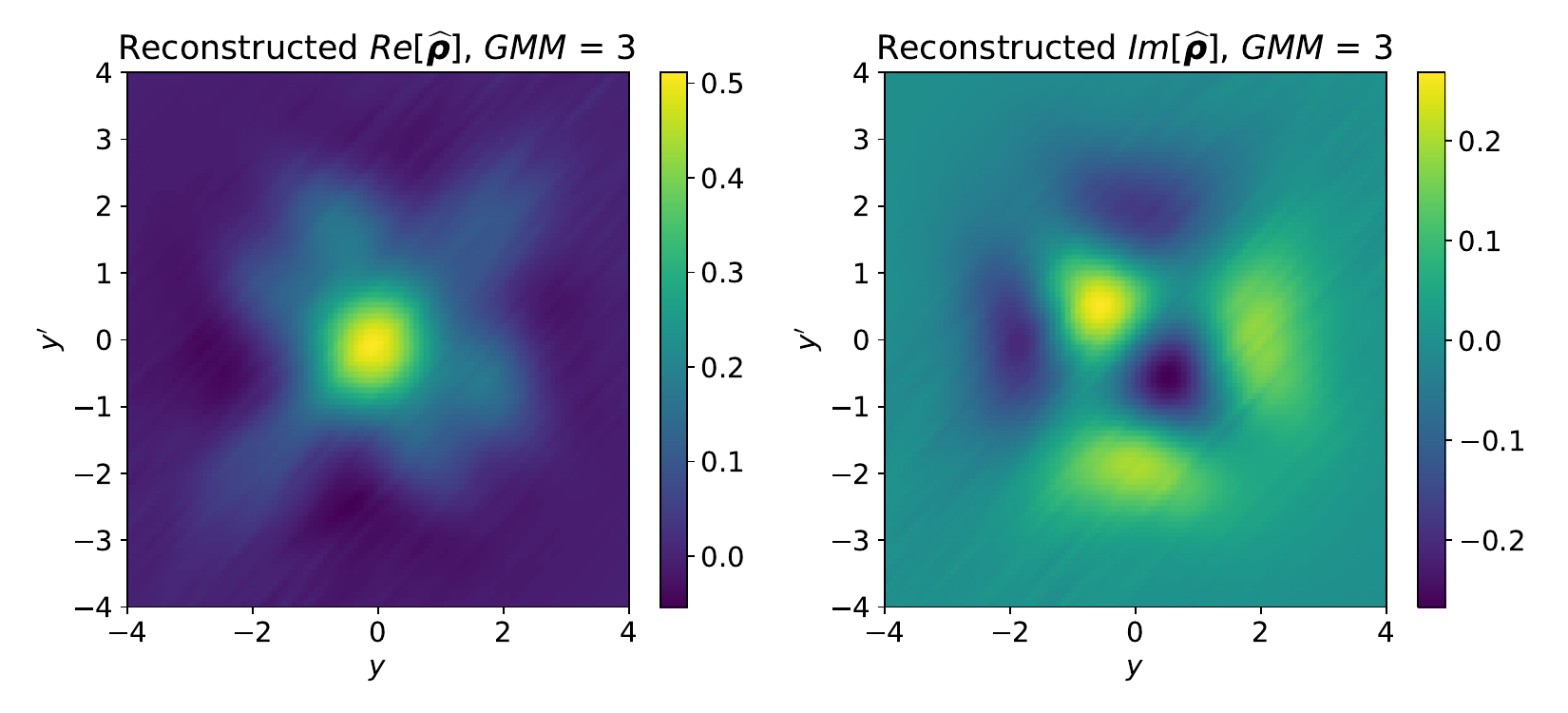}
        \caption*{(g)}
    \end{minipage}

 \caption{
Heatmaps of the CCS density matrix kernel: (a) theoretical kernel; 
(b) KQSE  from noisy data without filtering; 
(c) KQSE  from noisy data with filtering. 
Panel (d) shows a comparison of the diagonal elements of the theoretical kernel and the KQSE.
Panels (e) and (g) display CCS density matrix kernels estimated using MLE under Gaussian mixture model assumptions with $GMM=2$ and $GMM=3$, respectively.
The parameters are $n=500$, $\mu_{\mathrm{max}}=8$, $N_{\mu}=160$, and $a=1+0.5i$. 
For noisy data, $\kappa=0.85$.
}
\label{fig:rho-comparison}
\end{figure}
Fig.~\ref{fig:rho-comparison} illustrates the performance of the KQSE of \eqref{848} under different conditions. Panel (a) shows the theoretical kernel \eqref{848} for a complex amplitude $a = 1 + 0.5i$. Panel (b) demonstrates the degradation caused by noise when no filtering is applied. In contrast, panel (c) shows a visibly improved reconstruction obtained using our filtering procedure. Panel (d) provides a quantitative comparison of the diagonal elements, highlighting the accuracy gain achieved through filtering.
Panels (e) and (d) provide the MLE of the CCS kernel for $GMM=2$ and $GMM=3$. 
The reconstruction was performed with $n = 500$ samples, $\mu_{\mathrm{max}} = 8$, $N_{\mu} = 160$, and noise parameter $\kappa = 0.85$.
\begin{table}[h]
    \centering
    \begin{tabular}{|c|c|c|c|c|}
        \hline
        \textbf{n} & \textbf{Data type} & \textbf{MLE}, $GMM=2$ & \textbf{MLE}, $GMM=3$ &  \textbf{KQSE} \\
        \hline
        \multirow{3}{*}{500}  & Noiseless  & \(2.050\times10^{-2}\) & \(8.159\times10^{-3}\) & \(1.667\times10^{-4}\) \\[-1ex]
                              & Noisy      & \(1.375\times10^{-2}\) & \(2.674\times10^{-2}\) & \(1.279\times10^{-2}\) \\[-1ex]
                              & Corrected  & \(2.165\times10^{-2}\) & \(8.215\times10^{-3}\) & \(1.747\times10^{-4}\) \\
        \hline
        \multirow{3}{*}{1000} & Noiseless  & \(2.155\times10^{-2}\) & \(7.812\times10^{-3}\) & \(5.572\times10^{-5}\) \\[-1ex]
                              & Noisy      & \(1.384\times10^{-2}\) & \(2.555\times10^{-2}\) & \(1.246\times10^{-2}\) \\[-1ex]
                              & Corrected  & \(2.172\times10^{-2}\) & \(7.928\times10^{-3}\) & \(5.592\times10^{-5}\) \\
        \hline
        \multirow{3}{*}{2000} & Noiseless  & \(2.151\times10^{-2}\) & \(6.973\times10^{-3}\) & \(2.457\times10^{-5}\) \\[-1ex]
                              & Noisy      & \(1.361\times10^{-2}\) & \(2.508\times10^{-2}\) & \(1.222\times10^{-2}\) \\[-1ex]
                              & Corrected  & \(2.207\times10^{-2}\) & \(7.016\times10^{-3}\) & \(2.690\times10^{-5}\) \\
        \hline
    \end{tabular}
   \caption{
MSE of the CCS density matrix kernel estimated using KQSE with a Gaussian kernel and MLE based methods for sample sizes $n=\{500,1000,2000\}$. The parameters are $n=500$, $\mu_{\mathrm{max}}=8$, $N_{\mu}=160$, and $a=1+0.5i$. 
For noisy data, $\kappa=0.85$.
}
\label{tab:mle_rho_comparison}
\end{table}
The comparison of the numerical error~\eqref{1033_333} for the estimation of the CCS density matrix kernel obtained using KQSE with a Gaussian kernel and MLE based reconstructions is summarized in Table~\ref{tab:mle_rho_comparison}.
\par Using the similar steps, we perform KQSE procedure to estimate the purity of the CCS. Using it the squared trace distance is calculated. Table~\ref{tab:combined} presents the MSE of the KQSEs $\widehat{\mathrm{tr}(\boldsymbol{\rho}\boldsymbol{\rho}_{nh})}$ and $\widehat{D^2(\boldsymbol{\rho},\boldsymbol{\rho}_{nh})}$ under variation of key reconstruction parameters.
\begin{table}[h]
    \centering
    \begin{tabular}{|c|c|c|c|c|c|}
        \hline
        \textbf{Parameter varied} & \textbf{Value} & \textbf{Data type} & $\hat{L}_{\infty}(\boldsymbol{\rho}_{nh})$ & $\mathrm{MSE}(\widehat{\mathrm{tr}(\boldsymbol{\rho}\boldsymbol{\rho}_{nh})})$ & $\widehat{D^2(\boldsymbol{\rho},\boldsymbol{\rho}_{nh})}$ \\
        \hline
        \multirow{9}{*}{$n$} 
            &      & Noiseless & \(1.667\times10^{-4}\) & \(5.446\times10^{-5}\) & \(7.352\times10^{-3}\) \\
            & 500  & Noisy     & \(1.279\times10^{-2}\) & \(1.195\times10^{-2}\) & \(1.093\times10^{-1}\) \\
            &      & Corrected & \(1.747\times10^{-4}\) & \(6.099\times10^{-5}\) & \(7.744\times10^{-3}\) \\
            \cline{2-6}
            &      & Noiseless & \(5.572\times10^{-5}\) & \(9.253\times10^{-6}\) & \(2.889\times10^{-3}\) \\
            & 1000 & Noisy     & \(1.246\times10^{-2}\) & \(1.181\times10^{-2}\) & \(1.057\times10^{-1}\) \\
            &      & Corrected & \(5.592\times10^{-5}\) & \(9.526\times10^{-6}\) & \(3.040\times10^{-3}\) \\
            \cline{2-6}
            &      & Noiseless & \(2.457\times10^{-5}\) & \(1.983\times10^{-6}\) & \(1.083\times10^{-3}\) \\
            & 2000 & Noisy     & \(1.222\times10^{-2}\) & \(1.126\times10^{-2}\) & \(1.003\times10^{-1}\) \\
            &      & Corrected & \(2.690\times10^{-5}\) & \(2.301\times10^{-6}\) & \(1.454\times10^{-3}\) \\
        \hline
        \multirow{9}{*}{$\mu_{\max}$}
            &     & Noiseless & \(1.259\times10^{-3}\) & \(1.006\times10^{-3}\) & \(3.171\times10^{-2}\) \\
            & 4   & Noisy     & \(1.527\times10^{-2}\) & \(1.497\times10^{-2}\) & \(1.537\times10^{-1}\) \\
            &     & Corrected & \(1.277\times10^{-3}\) & \(1.105\times10^{-3}\) & \(3.246\times10^{-2}\) \\
            \cline{2-6}
            &     & Noiseless & \(1.701\times10^{-4}\) & \(5.867\times10^{-5}\) & \(7.605\times10^{-3}\) \\
            & 6   & Noisy     & \(1.397\times10^{-2}\) & \(1.131\times10^{-2}\) & \(1.062\times10^{-1}\) \\
            &     & Corrected & \(1.781\times10^{-4}\) & \(6.389\times10^{-5}\) & \(7.942\times10^{-3}\) \\
            \cline{2-6}
            &     & Noiseless & \(1.667\times10^{-4}\) & \(5.446\times10^{-5}\) & \(7.352\times10^{-3}\) \\
            & 8   & Noisy     & \(1.279\times10^{-2}\) & \(1.126\times10^{-2}\) & \(1.003\times10^{-1}\) \\
            &     & Corrected & \(1.747\times10^{-4}\) & \(6.099\times10^{-5}\) & \(7.744\times10^{-3}\) \\
        \hline
        \multirow{9}{*}{$N_{\mu}$}
            &     & Noiseless & \(3.159\times10^{-4}\) &\(1.916\times10^{-4}\) &\(1.248\times10^{-2}\) \\
            & 40  & Noisy     & \(1.321\times10^{-2}\) &\(1.202\times10^{-2}\) &\(1.094\times10^{-1}\) \\
            &     & Corrected & \(3.179\times10^{-4}\) &\(2.170\times10^{-4}\) &\(1.362\times10^{-2}\) \\
            \cline{2-6}
            &     & Noiseless & \(2.229\times10^{-4}\) &\(7.930\times10^{-5}\) &\(8.472\times10^{-3}\) \\
            & 80  & Noisy     & \(1.302\times10^{-2}\) &\(1.147\times10^{-2}\) &\(1.071\times10^{-1}\) \\
            &     & Corrected & \(2.247\times10^{-4}\) &\(8.223\times10^{-5}\) &\(8.993\times10^{-3}\) \\
            \cline{2-6}
            &     & Noiseless & \(1.667\times10^{-4}\) & \(5.446\times10^{-5}\) & \(7.352\times10^{-3}\) \\
            & 160 & Noisy     & \(1.279\times10^{-2}\) & \(1.126\times10^{-2}\) & \(1.003\times10^{-1}\) \\
            &     & Corrected & \(1.747\times10^{-4}\) & \(6.099\times10^{-5}\) & \(7.744\times10^{-3}\) \\
        \hline
    \end{tabular}
    \caption{The distances $\hat{L}_{\infty}(\boldsymbol{\rho}_{nh})$, $\mathrm{MSE}(\widehat{\mathrm{tr}(\boldsymbol{\rho}\boldsymbol{\rho}_{nh})})$ and the value $\widehat{D^2(\boldsymbol{\rho},\boldsymbol{\rho}_{nh})}$ under change of different parameters: sample size $n$, maximal parameter $\mu_{\max}$, and number of measurement settings $N_{\mu}$, with fixed $\kappa=0.85$. When not varied, the default values are $\mu_{\max}=8$, $N_{\mu}=160$, and $n=500$.}
    \label{tab:combined}
\end{table}
As expected, increasing the sample size $n$ or the resolution $N_{\mu}$ leads to improved reconstruction quality, reflected by a lower mean squared error and smaller distance to the true state.

\subsection*{4.2. Kitten state}
The experimental data analyzed in this work were obtained from the homodyne-tomographic measurements reported in Ref.~\cite{Lvovsky_2002,Lvovsky_2004}.
In that experiment, nonclassical single-mode optical states were prepared using conditional measurements on a highly unbalanced beam splitter. One input port of the beam splitter was fed with a weak coherent state of amplitude $\alpha$, while the second input port contained a conditionally prepared single-photon Fock state. The beam splitter reflectivity was $r^2=0.925$, corresponding to a small transmission amplitude $t=\sqrt{1-r^2}$.

A single-photon detector placed in one of the output ports served as a conditioning device: homodyne measurements in the other output channel were recorded only when a detection event (“click”) occurred. Due to quantum interference at the beam splitter, this conditioning erased the which-path information about the detected photon, resulting in the preparation of a nonclassical state in the signal channel. In the idealized scenario, this procedure generates a coherent superposition of the vacuum and single-photon states,
\begin{eqnarray}
    \ket{\psi_s}
    =c_0\ket{0}+c_1\ket{1},
    \qquad
    c_0=\frac{t}{\sqrt{t^2+\alpha^2}},
    \quad
    c_1=\frac{\alpha}{\sqrt{t^2+\alpha^2}},
\end{eqnarray}
commonly referred to as a \textit{Schr{\"o}dinger kitten} state.

The prepared states were characterized using balanced homodyne detection with a phase-scanned local oscillator, providing access to quadrature probability distributions for different local-oscillator phases. The experimental dataset consists of
$T=14153$ raw homodyne detector data points,  $\bigl\{x_{\theta_j},\theta_j\bigr\}_{j=1}^{T}$,
where $x_{\theta_j}$ denotes the measured quadrature value at phase $\theta_j$.
\par Due to experimental imperfections, including non-unit preparation efficiency of the single-photon state, optical losses, and dark counts of the single-photon detector, the reconstructed quantum state was not pure. Instead, it is well described by the mixed state
\begin{eqnarray}\label{933}
    \rho=(1-p)\ket{0}\bra{0}+p\ket{\psi_s}\bra{\psi_s}.
\end{eqnarray}
We note that $\langle 0|\psi_s\rangle=c_0$, so that
$|\langle 0|\psi_s\rangle|^2=|c_0|^2=1-|c_1|^2$. The purity of this state is therefore
\begin{eqnarray}
    \tr{\rho^2}
    =(1-p)^2+p^2+2p(1-p)|\langle0|\psi_s\rangle|^2
    =1-2p(1-p)\frac{\alpha^2}{t^2+\alpha^2}.
\end{eqnarray}
Solving this equation for $p$ yields
\begin{eqnarray}
    p=\frac{1}{2}\left(
    1\pm
    \sqrt{
    1-2(1-\tr{\rho^2})\frac{t^2+\alpha^2}{\alpha^2}
    }
    \right).
\end{eqnarray}
In Ref.~\cite{Lvovsky_2002}, for $\alpha=0.3$ the purity is reported to be $\tr{\rho^2}=0.77$, which leads to two possible values,
$p=\{0.3189537,\,0.6810463\}$. Both values correspond to the same purity; however, in the physically relevant “mostly-kitten with a vacuum admixture” scenario, we take $p=0.6810463$.

To construct the symplectic  tomographic representation, we define
$\beta=\sqrt{\mu^2+\nu^2}$, $x=r y$, and $\theta=\arg(\mu+i\nu)$.
The required quadrature tomograms for Fock states are given in Eq.~\eqref{1530}. Since a rotated quadrature eigenstate introduces a phase factor $e^{i n\theta}$ for the Fock state $\ket{n}$, the interference term between $\ket{0}$ and $\ket{1}$ carries the phase factor $e^{\pm i\theta}$. Collecting all contributions, we obtain
\begin{eqnarray}
    \mathcal{W}(x|\beta)
    =
    \mathcal{W}_0(x|\beta)
    \left(
    1-p|c_1|^2
    +2p|c_1|^2\frac{x^2}{\beta^2}
    +2\sqrt{2}\frac{p x}{\beta}\,
    \Re\!\left[c_0c_1^{\star}e^{i\theta}\right]
    \right),
\end{eqnarray}
where $\Re[c_0c_1^{\star}e^{i\theta}]=c_0c_1\cos\theta$. Using $\cos\theta=\mu/\beta$, the symplectic tomogram takes the form
\begin{eqnarray}
    \mathcal{W}(x|\beta)
    =
    \mathcal{W}_0(x|\beta)
    \left(
    1-pc_1^2
    +2pc_1^2\frac{x^2}{\beta^2}
    +2\sqrt{2}pc_0c_1\frac{\mu x}{\beta^2}
    \right).
\end{eqnarray}
Using Eq.~\eqref{1530}, the characteristic function of this tomogram is
\begin{equation}\label{934}
\phi(t;\mu,\nu)
=
\phi_0(t;\beta)
\left[
1
-\frac{p\,c_1^{\,2}\,\beta^{2} t^{2}}{2}
-\mathrm{i}\sqrt{2}\,p\,c_0 c_1\,\mu\, t
\right].
\end{equation}
Finally, the kernel of the density-matrix operator in the position representation is obtained as
\begin{equation}\label{935}
\rho(y,y')
=
\frac{1}{\sqrt{\pi}}
\exp\!\left(-\frac{y^{2}+y'^{2}}{2}\right)
\left[
1
- p\,c_1^{2}
- \sqrt{2}\,p\,c_0 c_1\,(y+y')
+ 2p\,c_1^{2}\,y\,y'
\right].
\end{equation}
\begin{figure}
    \centering
    \includegraphics[width=0.7\linewidth]{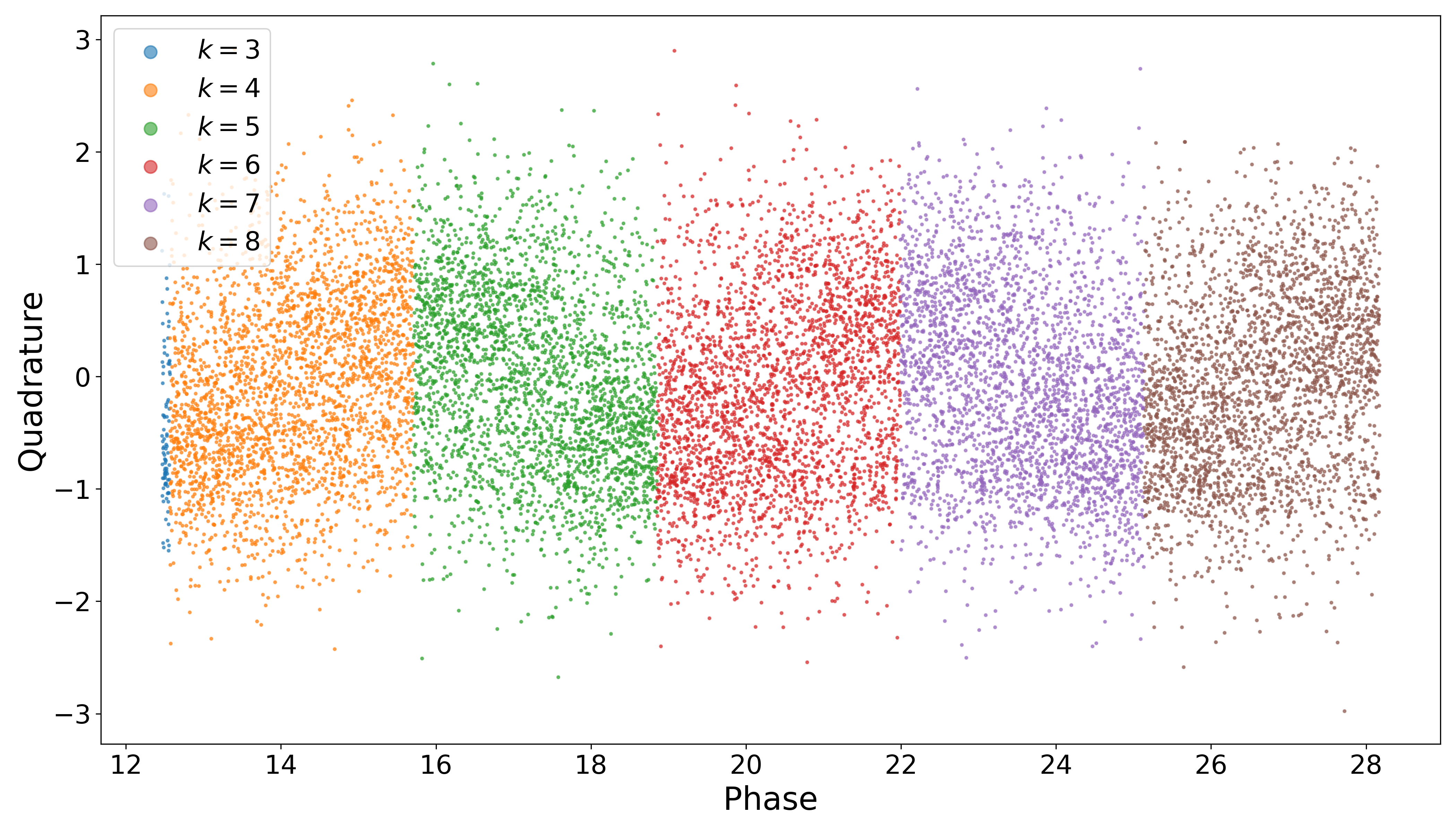}
    \caption{Raw homodyne quadrature samples plotted versus the local-oscillator phase and separated into consecutive 
$\pi$-intervals $k$. Each interval exhibits statistically equivalent quadrature distributions, demonstrating the absence of phase dependent drift and supporting the reduction of the data to a single phase range $[0,\pi)$.}
    \label{fig:data}
\end{figure}
We now describe how the obtained experimental data were analyzed. Homodyne detection measures the field quadrature, which satisfies the symmetry relation
$y_{\theta+\pi}=-y_\theta$. This property is clearly visible in the experimental data shown in Fig.~\ref{fig:data}. As a consequence, the data can be split into intervals of the form $[k\pi,(k+1)\pi]$. Each even interval is a sign-flipped mirror of the adjacent odd interval, so that all intervals collapse into a single consistent dataset on $[0,\pi)$.

After folding the data, we obtain several realizations of $x_{\theta_j}$ for each value of $\theta_j$. These realizations can be organized into bins in different ways. Here we choose bins of equal width in $\theta$, dividing the interval $[0,\pi)$ into $N_{\theta}=20$ subintervals of the same size, each containing approximately the same number of data points ($n\sim700$). The choice $N_{\theta}$ is motivated by the requirements of the KQSE method: since the effective cutoff in the symplectic parameter satisfies $\mu_{\max}\sim\log n$ and the phase discretization must scale as $N_{\theta}\sim (\log n)(\log n)$ for stable characteristic-function estimation, this value provides a consistent balance between angular resolution and statistical accuracy for $n$ data points per bin.
Thus, starting from the optical homodyne dataset $\{y_j,\theta_j\}_{j=1}^{T}$, we convert the data into symplectic variables $\{x_j,\mu_j,\nu_j\}_{j=1}^{T}$ using
\begin{eqnarray}
    x_j=\beta_j y_j,\quad \mu_j=\beta_j\cos{\theta_j},\quad \nu_j=\beta_j\sin{\theta_j}.
\end{eqnarray}

This representation is particularly convenient for the subsequent data analysis, as it allows us to apply the KQSE framework directly. In this approach, the symplectic data $\{x_j,\mu_j,\nu_j\}$ are treated as samples drawn from the symplectic tomogram of the state, and are used to construct a nonparametric estimator of the corresponding characteristic function via kernel smoothing. The estimated characteristic function then serves as the central object from which tomograms, phase-space distributions, and density-matrix elements can be consistently reconstructed. In this way, the folded and binned homodyne data provide the direct input for the KQSE procedure employed in this work.
\par Fig.~\ref{fig:kitten_pdf_comparison} compares the experimentally obtained tomogram with its theoretical prediction~\eqref{933}, together with histogram-based, MLE-based, and KDE-based reconstructions.
\begin{figure}[h]
  \centering
    \includegraphics[width=0.5\linewidth]{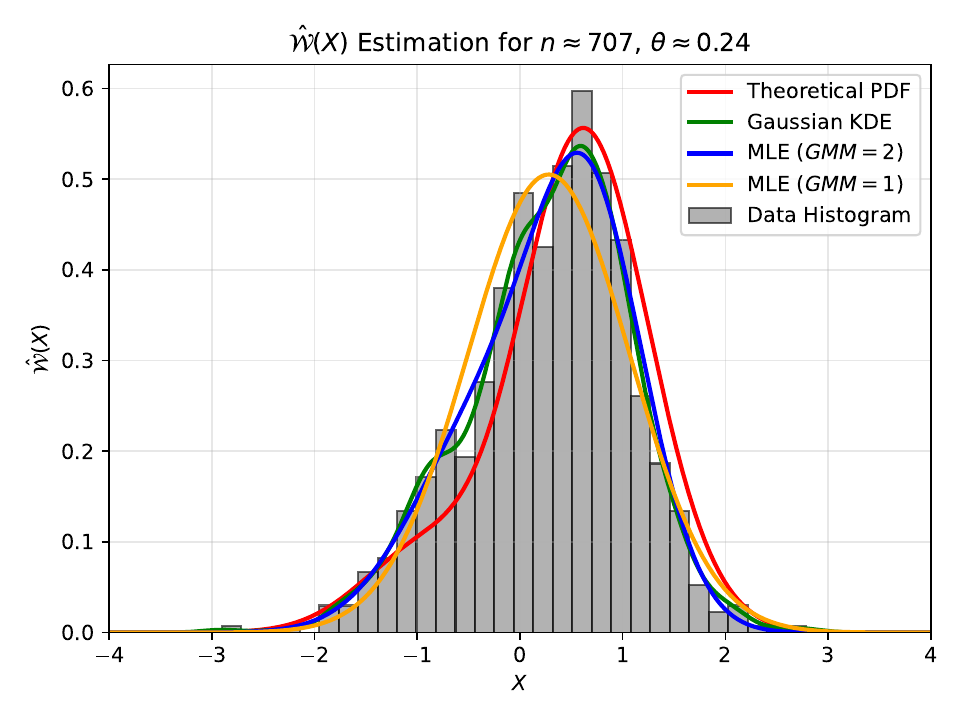}
 \caption{
Comparison of the experimentally obtained tomogram with the theoretical tomogram of the statistical mixture of the kitten state and vacuum~\eqref{933}. Shown are the experimental histogram, MLE-based reconstructions using one- and two-component Gaussian mixture models ($GMM=1$ and $GMM=2$), and the KDE-based reconstruction with a Gaussian kernel.}
  \label{fig:kitten_pdf_comparison}
\end{figure}
The experimental histogram indicates noticeable deviations from the theoretical tomogram, suggesting that the underlying physical state is not fully captured by the idealized model. In particular, while the theoretical tomogram corresponds to a bimodal probability density, the experimental data exhibit an asymmetry between the two modes, with the first peak being systematically more pronounced then on the theoretical curve. This feature is consistently reproduced by both the KDE-based estimator and the MLE reconstruction with a higher complexity Gaussian mixture model.

Such discrepancies are not unexpected in the experimental setting, where additional noise sources and imperfections may not be fully accounted for by the theoretical model. As a result, the theoretical tomogram should be regarded as an approximate reference rather than an exact ground truth in this case. Since only a single experimental dataset is available, we quantify the reconstruction accuracy using the integrated squared error (ISE) averaged by $\theta\in[0,\pi)$  with $N_{\theta}=20$, rather than the mean squared error:
% \begin{eqnarray}
%     \mathrm{ISE}_{KDE}=5.133\times10^{-3},\quad  \mathrm{ISE}_{GMM=2}=4.265\times10^{-3}, \quad \mathrm{ISE}_{GMM=1}=1.851\times10^{-2}.
% \end{eqnarray}
\begin{eqnarray}
    \mathrm{ISE}_{KDE}=7.430\times10^{-3},\quad  \mathrm{ISE}_{GMM=2}=7.037\times10^{-3}, \quad \mathrm{ISE}_{GMM=1}=2.242\times10^{-2}.
\end{eqnarray}
The numerical values of the ISE indicate that the KDE-based reconstruction and the MLE with a two-component Gaussian mixture model yield comparable performance, whereas the single-component MLE fails to adequately describe the experimental data. While parametric MLE is expected to be optimal when the assumed model is correct, the KDE approach remains competitive without relying on detailed model assumptions, making it particularly suitable for the analysis of experimental tomographic data.
\par Next, we estimate the CF using two different approaches. In the first approach, the CF is obtained by Fourier transforming the tomogram reconstructed via MLE under $GMM=1$ and $GMM=2$. In the second approach, the CF is reconstructed directly using the KCFE with a Gaussian kernel. The resulting CF estimates are shown in Fig.~\ref{fig:kitten_cf_comparison}.
The squared errors (SE) throughout $\mu\in[-8,8]$  with $N_{\mu}=100$ indicate a clear hierarchy in the reconstruction accuracy of the CF. The KCFE yields the smallest error,
$\mathrm{SE}_{\mathrm{KCFE}}=4.075\times10^{-3}$, followed by the MLE with a two-component Gaussian mixture model,
$\mathrm{SE}_{\mathrm{GMM}=2}=4.489\times10^{-3}$. In contrast, the single-component MLE exhibits a substantially larger error,
$\mathrm{SE}_{\mathrm{GMM}=1}=1.276\times10^{-2}$.

These results quantitatively confirm the qualitative behavior observed in Fig.~\ref{fig:kitten_cf_comparison}. While increasing the model complexity from $GMM=1$ to $GMM=2$ significantly improves the MLE-based reconstruction, the parametric approach remains less accurate than the kernel-based estimator.
\begin{figure}[h]
  \centering
  \begin{minipage}[b]{0.5\textwidth}
    \centering
    \includegraphics[width=\linewidth]{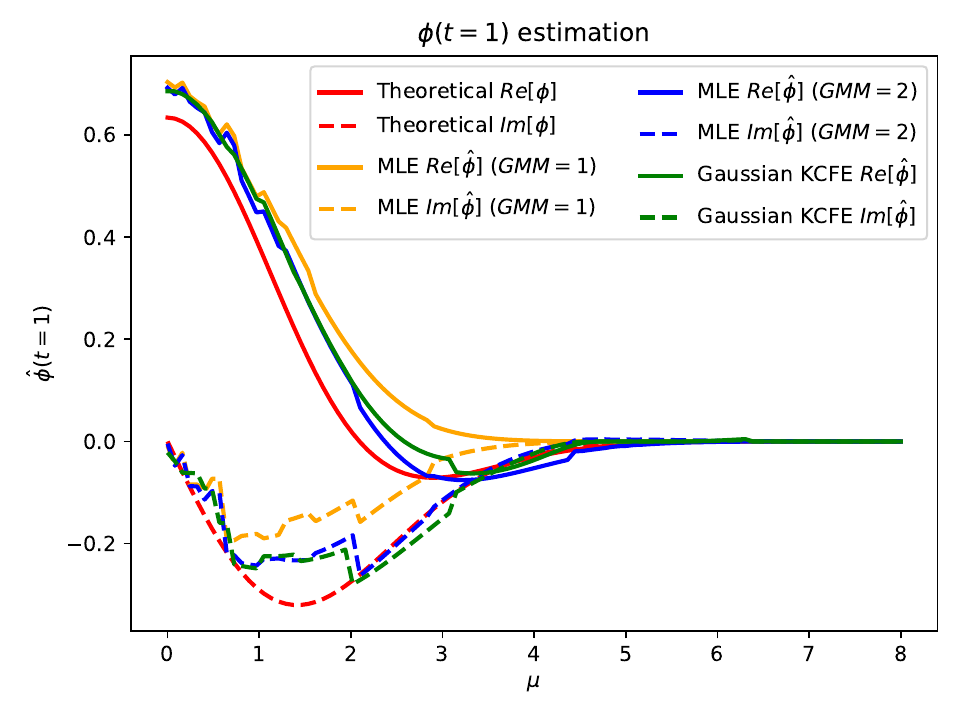}
   % \caption*{(a)}
  \end{minipage}\hfill
  %\begin{minipage}[b]{0.5\textwidth}
  %  \centering
  %  \includegraphics[width=\linewidth]{Figures/Lvovsky_KCFE_correction.pdf}
 %   \caption*{(b)}
 % \end{minipage}\hfill
 \caption{
Comparison of CF reconstructions for  the statistical mixture of the kitten state and vacuum~\eqref{933} state obtained using MLE with Gaussian mixture models ($GMM=1$ and $GMM=2$) and direct KCFE with a Gaussian kernel.}
 %$\mu=2$, $\nu=1$, $\kappa=0.9$ for Gaussian noise inversion in (b). 
 % With Gaussian noise inversion: $\mathrm{SE}_{KCFE}=1.477\times10^{-3}$, $\mathrm{SE}_{GMM=2}=1.383\times10^{-2}$, $\mathrm{SE}_{GMM=1}=4.342\times10^{-2}$. \xiaoyu{Change the notations the way you like, delete the 'with noise inversion' plot and data}}
  \label{fig:kitten_cf_comparison}
\end{figure}
%To further analyze the reconstruction accuracy across different regions of the CF, we report in Tab.~\ref{tab:mle_cf_diff_mu} the squared errors of the CF estimates obtained using the methods discussed above for three representative values of the parameter $\mu=1,2,4$. 
%\begin{table}[h]
%\centering
%\begin{tabular}{|c|c|c|c|}
%\hline
%$\mu$ & \textbf{MLE}, $GMM=1$    &\textbf{MLE}, $GMM=2 $   & \textbf{Gaussian KCFE}     \\ \hline
%1  & $1.851\times10^{-2}$ & $4.861\times10^{-3}$ & $8.043\times10^{-3}$ \\ \hline
%2  & $4.929\times10^{-2}$ & $1.692\times10^{-2}$ & $9.720\times10^{-3}$ \\ \hline
%4  & $1.771\times10^{-3}$ & $5.264\times10^{-4}$ & $5.329\times10^{-5}$ \\ \hline
%\end{tabular}
%\caption{SE of CF reconstructions evaluated at different values of $\mu$ for the statistical mixture of the kitten state and vacuum~\eqref{933}. The results are shown for MLE-based reconstructions using Gaussian mixture models ($GMM=1$ and $GMM=2$) and for KCFE with a Gaussian kernel.}\label{tab:mle_cf_diff_mu}
%\end{table}
\par Using the reconstructed CFs, we next estimate the density-matrix kernel~\eqref{935} of the statistical mixture of the kitten state and vacuum~\eqref{933}. The results are provided in Fig.~\ref{fig:kitten_rho_comparison}. The reconstruction accuracy is quantified by the $L_\infty$ error defined in Eq.~\eqref{1033_333}. The resulting errors are
$L_{\infty,\mathrm{GMM}=1}=3.424\times10^{-2}$,
$L_{\infty,\mathrm{GMM}=2}=1.372\times10^{-2}$,
and
$L_{\infty,\mathrm{KQSE}}=1.012\times10^{-2}$.
These results show a systematic improvement of the kernel reconstruction when moving from a single-component to a two-component Gaussian mixture model within the MLE approach. Nevertheless, the smallest reconstruction error is achieved by the KQSE method, indicating its enhanced robustness with respect to model mismatch and residual noise present in the experimentally motivated CF estimates.

\begin{figure}[h]
  \centering
    \begin{minipage}[b]{1.0\textwidth}
    \centering
    \includegraphics[width=\linewidth]{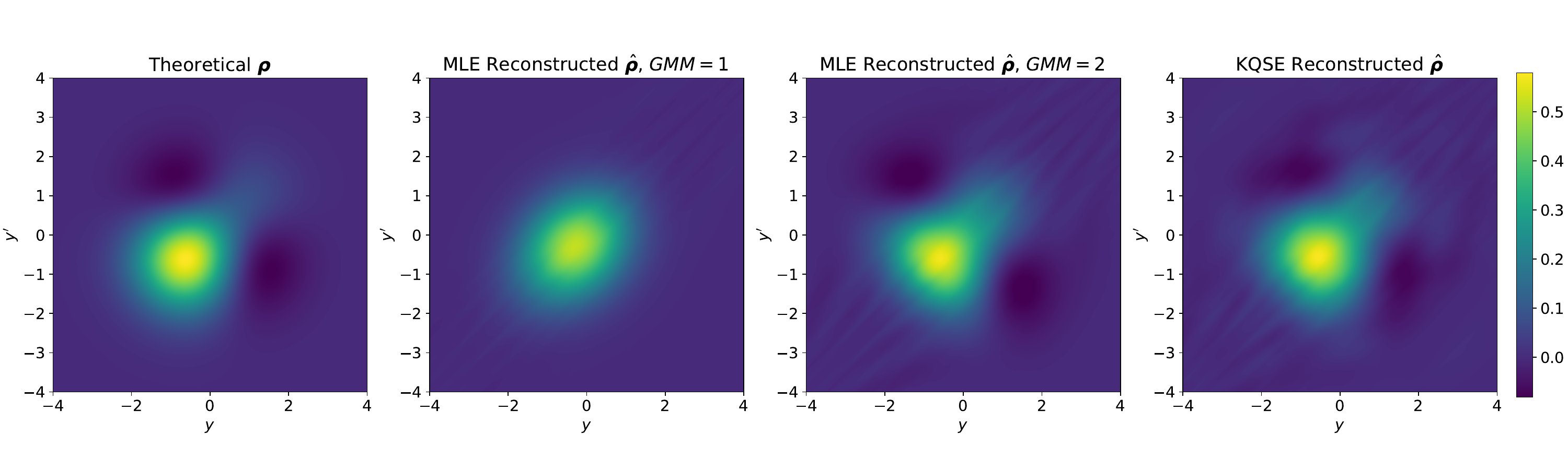}
 %   \caption*{(a) without noise inversion}
  \end{minipage}
%  \begin{minipage}[b]{1.0\textwidth}
 %   \centering
 %   \includegraphics[width=\linewidth]{Figures/rho_comparison_MLE_KQSE_correction.pdf}
%    \caption*{(b) with noise inversion}
 %   \end{minipage}
  \caption{Heatmaps of kernel \eqref{935} of the statistical mixture of the kitten state and vacuum~\eqref{933}: (a) theoretical kernel; 
(b) and (c) MLE under Gaussian mixture model assumptions with $GMM=1$ and $GMM=2$, respectively; (d) KQSE  
  }
%  With noise inversion ($\kappa=0.9$): $L_{\infty, GMM=1}=2.500\times10^{-2}$, $L_{\infty, GMM=2}=1.038\times10^{-2}$, $L_{\infty, KQSE}=8.314\times10^{-3}$.
%}
  \label{fig:kitten_rho_comparison}
\end{figure}}

\end{document}